\documentclass[preprint,12pt,authoryear]{elsarticle}

\usepackage{amssymb}
\usepackage{amsmath}
\usepackage{amsthm}
\usepackage{booktabs}
\usepackage{bbm}
\usepackage{tikz}
\usepackage{subcaption}
\usepackage{multicol}
\usepackage{multirow}
\usepackage{color,soul}
\usepackage{algpseudocode}
\usepackage{subcaption,float}
\usepackage[linesnumbered,ruled,lined, noend]{algorithm2e}
\newtheorem{definition}{Definition}
\newtheorem{proposition}{Proposition}

\newtheorem{claim}{Claim}
\def\Fig{{Figure}}
\def\H{\mathcal{H}}
\def\P{\mathcal{P}}
\def\L{\mathcal{L}}
\DeclareMathSymbol{\dv}{\mathord}{operators}{"3A}
\def\S{\mathbf{S}}

\journal{Nuclear Physics B}

\begin{document}

\begin{frontmatter}

%% Title, authors and addresses

%% use the tnoteref command within \title for footnotes;
%% use the tnotetext command for theassociated footnote;
%% use the fnref command within \author or \affiliation for footnotes;
%% use the fntext command for theassociated footnote;
%% use the corref command within \author for corresponding author footnotes;
%% use the cortext command for theassociated footnote;
%% use the ead command for the email address,
%% and the form \ead[url] for the home page:
%% \title{Title\tnoteref{label1}}
%% \tnotetext[label1]{}
%% \author{Name\corref{cor1}\fnref{label2}}
%% \ead{email address}
%% \ead[url]{home page}
%% \fntext[label2]{}
%% \cortext[cor1]{}
%% \affiliation{organization={},
%%            addressline={}, 
%%            city={},
%%            postcode={}, 
%%            state={},
%%            country={}}
%% \fntext[label3]{}

\title{Non-Parametric Bayesian Inference for Partial Orders with Ties from Rank Data observed with Mallows Noise} %% Article title

%% use optional labels to link authors explicitly to addresses:
%% \author[label1,label2]{}
%% \affiliation[label1]{organization={},
%%             addressline={},
%%             city={},
%%             postcode={},
%%             state={},
%%             country={}}
%%
%% \affiliation[label2]{organization={},
%%             addressline={},
%%             city={},
%%             postcode={},
%%             state={},
%%             country={}}

\author{Chuxuan (Jessie) Jiang and Geoff K. Nicholls} %% Author name

%% Author affiliation
\affiliation{organization={Department of Statistics, University of Oxford},%Department and Organization
            addressline={Wellington Square}, 
            city={Oxford},
            postcode={OX1 2JD}, 
            % state={},
            country={United Kingdom}}

%% Abstract
\begin{abstract}
%% Text of abstract
Partial orders may be used for modeling and summarising ranking data when the underlying order relations are less strict than a total order. They are a natural choice when the data are lists recording individuals' positions in queues in which queue order is constrained by a social hierarchy, as it may be appropriate to model the social hierarchy as a partial order and the lists as random linear extensions respecting the partial order. In this paper, we set up a new prior model for partial orders incorporating ties by clustering tied actors using a Poisson Dirichlet process. The family of models is projective. We perform Bayesian inference with different choices of noisy observation model. In particular, we propose a Mallow's observation model for our partial orders and give a recursive likelihood evaluation algorithm. We demonstrate our model on the `Royal Acta' (Bishop) list data where we find the model is favored over well-known alternatives which fit only total orders. 
\end{abstract}

%%Graphical abstract
% \begin{graphicalabstract}
% \centering
% \includegraphics[width=0.89\linewidth]{images/graphicalabstract.pdf}
% \end{graphicalabstract}

%%Research highlights
% \begin{highlights}
% \item A new partial order ranking model which allows for ties
% \item Non-parametric Bayesian inference for rank data in which random linear extensions are observed with Mallows-distributed noise.
% \end{highlights}

%% Keywords
\begin{keyword}
%% keywords here, in the form: keyword \sep keyword
Partial Order \sep DAGs \sep Ranking \sep Bayesian Non-Parametric Inference \sep Mallow's Noise Model \sep Social Hierarchy
%% PACS codes here, in the form: \PACS code \sep code

%% MSC codes here, in the form: \MSC code \sep code
%% or \MSC[2008] code \sep code (2000 is the default)

\end{keyword}

\end{frontmatter}

%% Add \usepackage{lineno} before \begin{document} and uncomment 
%% following line to enable line numbers
%% \linenumbers

%% main text
%%

%% Use \section commands to start a section
\section{Introduction}\label{sec:intro}
%% Labels are used to cross-reference an item using \ref command.

Ranking methods are widely used in research and industry to provide priority ranks and otherwise summarise data. They play a role in decision support, medical research, chemistry, and many other areas. Order ranking methods can be roughly classified into two categories - total order ranking and partial order ranking. In our application we are interested in recovering hierarchical relations between individuals or \textit{actors}.

A \textit{total order} coincides with one's common understanding of `ranking' where strict order relations exist between each pair of actors. Correspondingly, most probabilistic ranking models reconstruct an underlying total order. In the first Thurstonian model \citep{thurstone1927law}, actors are ranked based on random Gaussian attributes. This was followed by paired comparison models, including the Babington-Smith model \citep{smith1950discussion} and the Bradley-Terry model \citep{bradley1952rank}. These models assign a probability distribution to total orders given the probabilities of pairwise preference $p_{ij}, i,j\in [n]$ where $[n]=\{1,2,...,n\}$. \cite{bradley1952rank} introduced a special parameterisation of this preference such that $p_{ij} = {\nu_i}/({\nu_i+\nu_j}),\ \nu_i, \nu_j \in \mathbb{R}^+$. The Bradley-Terry model is widely studied and there are several extensions and generalisations. \cite{mallows1957non} proposed a tractable variant - the class of Mallows models. Mallows models are built on common ranking distance measures, for example, Spearman’s rho (Mallows $\theta$ model) and Kendall’s tau (Mallows $\phi$ model). \cite{plackett1975analysis} and \cite{luce1959possible} extended the Bradley-Terry model to compare multiple actors and introduced the Plackett-Luce model. The Mallows and Plackett-Luce models are the most commonly used statistical models for total order ranking in homogeneous populations. \cite{pearce2024bayesianrankclustering} developed a Rank-Clustered Bradley-Terry-Luce (BTL) model that allows for ties among actors in a BTL model. They used a partition-based spike-and-slab fusion prior. This is relevant to our ties idea as we will explain in section \ref{sec:motivation}. However, their model considers ties within total orders while we consider ties in rankings which are only required to be partially ordered. 

Mixtures of ranking models have been proposed and may be appropriate for heterogeneous populations. They allow us to fit data with more than one underlying total order in the observation model. Partial order models can be thought of as very general mixture models with one mixture component for every total order that respects the partial order (which can be a huge number, and varies as we change the partial order). Two popular mixture models are the Mallows mixture model and the Plackett-Luce mixture model. \cite{murphy2003mixtures} applies the Mallows mixture model to ranked list data and includes many useful references to earlier work in the area. \cite{meila2012dirichlet} proposes a Dirichlet process mixture of generalised Mallow Models over discrete incomplete rankings. \cite{lu2014effective} learns the Mallow model and its mixture variant from pairwise-preference data. For the Plackett-Luce mixture model, see \cite{mollica2017bayesian}, where they introduce a Bayesian finite mixture of Plackett-Luce models given partially ranked data. \cite{liu2019learning} learns the Plackett-Luce model and its mixtures from partial order data. \cite{caron2014bayesian} developed a Dirichlet process mixture in their non-parametric Plackett-Luce model given incomplete (top-$k$) ranking data. 

However, total order based models are not always naturally descriptive for many real-world problems, as one may have pairs of incomparable actors between whom no meaningful underlying order relation exists. There may be a temptation to think of the order on such pairs as simply being uncertain, but the physical reality may be that no order exists. If we want the elements of the model to correspond to elements of reality, relations between actors need to be represented by a partial order. 

A \textit{partial order} $h=\{[n],\prec_h\}$ is an order relation that assigns a binary order\footnote{The binary relation $\prec_h$ is both irreflexive (the relation $i \prec_h i$ does not exist) and transitive (if $i\prec_h j$ and $j\prec_h k$, then $i \prec_h k$), where $i,j,k\in [n]$ and $i\neq j\neq k$.} $\prec_h$ over a set of actor-labels $[n]$. Partial orders on $[n]$ are in one-to-one correspondence with transitively-closed directed acyclic graphs (DAG's) with nodes $V=[n]$ and edges $E$, such that $(i \prec_h j) \in E $ and $(j \prec_h k) \in E$ gives $(i \prec_h k) \in E$, $i,j,k \in [n], \, i\neq j\neq k$. Denote by $\H_{[n]}$ the set of all partial orders on actor labels $[n]=\{1,2,...,n\}$. An example partial order\footnote{In this article, we visualise partial orders via their transitive reduction - this omits all edges implied by transitivity and is unique.} is shown in \Fig~\ref{POex}. A \textit{linear extension} $l_h$ of a partial order $h$ is a permutation of $[n]$ which respects $h$, so that lesser actors in $h$ never come before greater. See \Fig~\ref{LEex} for examples. We denote the set of all linear extensions of partial order $h$ by $\mathcal{L}[h]$. 

% Figure 1 & 2
\begin{multicols}{2}
    \begin{minipage}{0.9\linewidth}
    \centering
    \begin{tikzpicture}[thick,scale=1.07, every node/.style={scale=0.8}]
        \node[draw, circle, minimum width=.1cm] (1) at (0, 1) {$1$};
        \node[draw, circle, minimum width=.1cm] (2) at (-1, -0.5) {$2$};
        \node[draw, circle, minimum width=.1cm] (3) at (1, 0) {$3$};
        \node[draw, circle, minimum width=.1cm] (4) at (1, -1) {$4$};
        \node[draw, circle, minimum width=.1cm] (5) at (0, -2) {$5$};
        \draw[-latex] (1) -- (2);
        \draw[-latex] (2) -- (5);
        \draw[-latex] (1) -- (3);
        \draw[-latex] (3) -- (4);
        \draw[-latex] (4) -- (5);
    \end{tikzpicture}
    
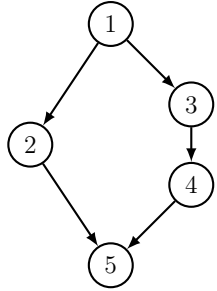
\captionof{figure}{A random partial order $h_0$ with 5 actors.}\label{POex}
    \end{minipage}%

    \begin{minipage}{0.9\linewidth}
    \centering
    \begin{tikzpicture}[thick,scale=0.8, every node/.style={scale=0.8}]
        \node[draw, circle, minimum width=.1cm] (1) at (0, 2) {$1$};
        \node[draw, circle, minimum width=.1cm] (2) at (0, 1) {$2$};
        \node[draw, circle, minimum width=.1cm] (3) at (0, 0) {$3$};
        \node[draw, circle, minimum width=.1cm] (4) at (0, -1) {$4$};
        \node[draw, circle, minimum width=.1cm] (5) at (0, -2) {$5$};
        \draw[-latex] (1) -- (2);
        \draw[-latex] (2) -- (3);
        \draw[-latex] (3) -- (4);
        \draw[-latex] (4) -- (5);
    \end{tikzpicture}
    \begin{tikzpicture}[thick,scale=0.8, every node/.style={scale=0.8}]
        \node[draw, circle, minimum width=.1cm] (1) at (0, 2) {$1$};
        \node[draw, circle, minimum width=.1cm] (2) at (0, 1) {$3$};
        \node[draw, circle, minimum width=.1cm] (3) at (0, 0) {$2$};
        \node[draw, circle, minimum width=.1cm] (4) at (0, -1) {$4$};
        \node[draw, circle, minimum width=.1cm] (5) at (0, -2) {$5$};
        \draw[-latex] (1) -- (2);
        \draw[-latex] (2) -- (3);
        \draw[-latex] (3) -- (4);
        \draw[-latex] (4) -- (5);
    \end{tikzpicture}
    \begin{tikzpicture}[thick,scale=0.8, every node/.style={scale=0.8}]
        \node[draw, circle, minimum width=.1cm] (1) at (0, 2) {$1$};
        \node[draw, circle, minimum width=.1cm] (2) at (0, 1) {$3$};
        \node[draw, circle, minimum width=.1cm] (3) at (0, 0) {$4$};
        \node[draw, circle, minimum width=.1cm] (4) at (0, -1) {$2$};
        \node[draw, circle, minimum width=.1cm] (5) at (0, -2) {$5$};
        \draw[-latex] (1) -- (2);
        \draw[-latex] (2) -- (3);
        \draw[-latex] (3) -- (4);
        \draw[-latex] (4) -- (5);
    \end{tikzpicture}
    
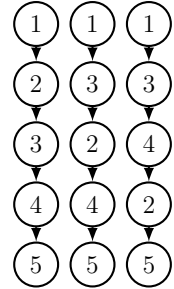
\captionof{figure}{All three linear extensions of partial order $h_0$ in \Fig~\ref{POex}.}\label{LEex}
    \end{minipage}
\end{multicols}

Partial orders express fundamental incomparability between actors, rather than a lack of knowledge of the underlying order relation or randomness in the realised order. Imagine the partial order $h_0$ in a hypothetical university setting, where actor 1 is the Head of Department, 2 is the Academic Administrator, 3 and 4 are the Deputy Head of Department and a lecturer respectively, and 5 is a student. Partial order $h_0$  represents a power hierarchy in which the relation $i\succ_{h_0} j$ expresses the reality that $i$ can direct $j$ and not {\it vis versa}. The Academic Administrator is not comparable with the Deputy Head of Department as neither manages the other. The lecturer answers to the Deputy Head, and therefore also to the Head, but not the Administrator, and so on. 

A total order is just a partial order with no incomparable pairs of actors, so assuming the underlying order is a partial order is weaker than assuming a total order. However, despite being a natural generalisation of rank, the use of partial orders in rank analysis is largely restricted to its use as a tool for summarising order relations obtained by fitting total order models, or as data in which partial orders replace lists, rather than as a component or ``parameter'' of the model itself. Here we list some exceptions. 

In settings in which the partial order is estimated as an explicit ``parameter'' of the model, both Frequentist and Bayesian methods have been applied. Work on Frequentist estimation includes \cite{beerenwinkel2007conjunctive} and \cite{gionis2006algorithms}. \cite{beerenwinkel2007conjunctive} uses partial orders to map the sequencing relations of different genome mutations. The paper develops a maximum likelihood partial order estimator for conjunctive Bayesian network (CBN) models given sets of binary mutation data occurring in orders constrained by a partial order. \cite{gionis2006algorithms} uses `bucket' orders\footnote{In a `bucket' order, the actors are arranged in layers - every actor is ordered with respect to actors in other layers, and any pair of actors in the same layer are either unordered or tied.}, a subclass of partial orders, to represent the temporal order of the fossil discovery sites in seriological data analysis. \cite{gionis2006algorithms} and \cite{feng2008discovering} introduce a pivot optimisation algorithm and a Bucket Gap algorithm respectively to approximate the `bucket' orders given pairwise/full rankings. Bayesian literature is limited. \cite{nicholls122011partial} estimates a partial order $h$ from list data assuming the data-lists are linear extensions drawn uniformly at random from the set of all linear extensions of $h$. \cite{nicholls2023bayesian} extends such model to the temporal setting. Unlike the above examples, \cite{nicholls122011partial} and \cite{nicholls2023bayesian} treats the partial order $h$ as a random variable. 

This paper extends \cite{nicholls122011partial} by proposing new prior and observation models for partial orders. In our setting the unknown partial order $h$ is observed indirectly through possibly incomplete lists of actor labels. The lists respect the order relations imposed by $h$, so observing the lists tells us something about $h$. For $i,j\in[n]$, if actor $i$ appears before actor $j$ in the list then we must have either $i\prec_h j$ or $i\|_h j$, i.e. either $j$ beats $i$ or $i$ and $j$ are incomparable in the partial order. Such a list is a linear extension of the partial order $h$ if it includes all $n$ actors and otherwise it is a linear extension of the suborder of actors in that list.
%. If the actor labels in a list form a subset $o$ of $[n]$ then that list is a linear extension of the \textit{suborder} $h[o]$ of $h$: is a restricted partial order to a subset $o\in [n]$ while inheriting the relevant order relation in $h$. 

In our Bayesian analysis, we will have a prior over random partial orders and a posterior informed by a set of rank-order lists, each one a permutation of actor labels. We refer to this as \textit{list data}.
We will see that one natural observation model relating the lists to the partial order is to treat the lists as uniform random linear extensions of the partial order. This is motivated by considering list data recording the positions of actors in a queue in which higher status individuals come before those of lower status. A status hierarchy is a partial order, so if all queue orders respecting the hierarchy are equally likely, the queue will be a uniform random draw from the linear extensions of the partial order. This model is also justified by writing down a simple stochastic process model for the queue formation process. We return to this point in Section~\ref{NF}.

The distinguishing features of our work are of threefold. First, it is Bayesian: we work with probability distributions over partial orders. This makes it straightforward to quantify uncertainty, at least in principle. Second, our models are defined over the set of all partial orders $\H_{[n]}$ of $n$ actors rather than some subset of $\H_{[n]}$ such as bucket orders.  Third, our data are linear extensions of suborders, subject to possible recording errors as we discuss later. We further extend the basic latent variable model for partial orders given by \cite{nicholls122011partial} in three ways - we incorporate ties between actors, we allow the dimension of the latent space to be unknown a priori, and we fit observation models in which much more general noise processes are acting. These are needed for model realism. 

The paper is structured as follows: in the rest of this section, we motivate our tie and observation error models, and introduce some basic concepts. In Section \ref{sec:po-prior}, we write down our partial order model with ties, the $(PDP,\rho)$-partial order model. In Section \ref{sec:obs}, we discuss possible choices for the observation model. We consider a Mallows ranking model that allows error in the observed lists, so the observed lists need not be linear extensions but must be ``close'' to being linear extensions. In Section \ref{sec:bayesian-inference}, we set out Bayesian inference and identify some computational challenges, notably the problem of counting linear extensions in likelihood evaluations. Finally in Section \ref{sec:application}, we apply our method to `Royal Acta', a dataset containing the order of witnesses in legal documents from the twelfth century. In particular, we study the social hierarchy between bishops in two specific time periods. We compare the prior model with ties with other possible ranking models and show our model effectiveness. 

%% Use \subsection commands to start a subsection.
\subsection{Motivation}\label{sec:motivation}

We begin by considering ties. Ties between actors are a common phenomenon in reality. For example, \Fig~\ref{Tieex} shows a hierarchy involving a field marshal (FM), two generals (G1 \& G2), an earl (E) and a bishop (B). We suppose their ranking is determined by their titles, and there is no order between military and religious or noble titles, so the generals are equal. This can be interpreted in two ways. First, they must have the same order relations with other actors in the partial order. Secondly, we might additionally require that tied actors appear as an unsplit group in any list, that is, in a random order but sequentially with no other actors between them. More details in \ref{appx:tie2}. In this paper we interpret ties in the first way. The second use of tie models is of interest to us but is future work.

\begin{minipage}{0.9\linewidth}
    \centering
    \begin{tikzpicture}[thick,scale=1.07, every node/.style={scale=0.8}]
        \node[draw, circle, minimum width=.7cm] (1) at (-0.5, 1) {$FM$};
        \node[draw, circle, minimum width=.7cm] (2) at (-1, 0) {$G1$};
        \node[draw, circle, minimum width=.7cm] (3) at (0, 0) {$G2$};
        \node[draw, circle, minimum width=.7cm] (4) at (1, 0) {$E$};
        \node[draw, circle, minimum width=.7cm] (5) at (1, 1) {$B$};
        \node[draw=blue!60, rectangle, minimum width=2.5cm,minimum height=1cm] (6) at (-0.5, 0) {};
        \draw[-latex] (1) -- (6);
        \draw[-latex] (5) -- (4);
    \end{tikzpicture}
    
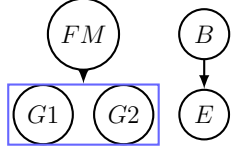
\captionof{figure}{The partial order between a field marshal ($FM$), two generals ($G1$ \& $G2$), a bishop ($B$) and an earl ($E$) if their order relation is determined by their titles.}\label{Tieex}
\end{minipage}

In the example the actor titles can be thought of as attributes or ``covariates'' which inform the partial order in a simple way. However, discrete attributes of this sort are often unknown to us (as in the examples below) and so we have to learn the tie structure. In our partial order model with ties - the $(PDP,\rho)$-partial order model - we use a Poisson Dirichlet process to model the clustering process. The partial order describes the order relation between the clusters. When there are unobserved attributes in the setting, this model respects the generative model. From a modelling perspective, it reduces the dimension of the latent variables needed to represent the partial order, and in turn increases the marginal likelihood. More importantly, the $(PDP,\rho)$-partial order model assigns higher prior probability on the `bucket' orders, or partial orders which are ``close'' to being `bucket' orders (see \ref{app:vspbo}). This is desirable as bucket orders are common in social hierarchies while allowing adequate prior probability on general partial orders. One nice feature of the families of prior models we write down for partial orders is that they are all \textit{projective} (every marginal of every distribution in the family is also in the family). 

We next motivate our models for ``noise'' in lists. Measurement or recording errors are common in data. For example, the lists we observe may be generated from linear extensions which respected the unknown underlying partial order but were then subject to random exchange of a small number of elements due to recording errors. There is little work fitting partial order models which include a noise component. \cite{beerenwinkel2007conjunctive} includes an error component in their mixture model. In \cite{nicholls122011partial} the error model treats errors as a random process of queue-jumping.  While this model is physically motivated, it only allows uni-directional queue-jumping in any given list realisation: errors occur when randomly chosen actors are promoted up (or down) the queue. \cite{jiang2023bayesian} extends \cite{nicholls122011partial} to propose a bi-directional queue-jumping model. However, this model is computationally costly and may only be suitable for partial orders with fast linear extension counting, e.g. the vertex-series-parallel partial orders\footnote{The vertex-series-partial orders are a class of partial orders that can be obtained by \textit{series} and \textit{parallel} operations. See \cite{jiang2023bayesian} for details.} in \cite{jiang2023bayesian}.

By contrast there is a large literature reconstructing an underlying total order from noisy list data. For example, in the Mallows model, the observed list will be ``close to'' but need not be equal to the central ranking under some distance measure. This suggests taking the Mallows model as a noise model, allowing errors in both directions away from a linear extension in a single list realisation. 
%The Plackett-Luce distribution is a model for random orders parameterised by given actor weights. The larger the actor weight is, the higher probability that the actor will be ranked higher up in any list realisation.  
% If the ranks of the Plackett-Luce actor weights are constrained to be linear extensions of the underlying partial order $h$ then a realisation of the Placket-Luce model will be ``close to'' being a linear extension of the partial order. 
We constrain the central ranking of the Mallows model to be a linear extension of a partial order. The dispersion parameter $\phi$ governs how ``close'' the rank is to being a linear extension of the partial order. Depending on the value of $\phi$, the generative model can form ranks that are pure noise, or ranks that are exactly linear extensions of partial order $h$, or in between. The dispersion parameter controls the level of noise in our model. 
% We will write down a prior on actor weights, conditioned on the partial order $h$, with the property that the ranks of the Plackett-Luce weights are uniform random linear extensions of $h$. 
% The more `spread-out' the weights are (in logit transform space) the more closely the observed list will follow the ranking of the weights, so the lists will distribute more like noise free random linear extensions of the partial order $h$. Our model allows the variance for actor weights in the prior to serve as a noise-control parameter that indicates how much error is present in the data. The Placket-Luce model is projective as is shown in \cite{hunter2004mm} and we ensure that this property is preserved in our model: the distribution of orders in short lists is the marginal of the order distribution of longer lists drawn from the same model.

\subsection{Partial orders}\label{PO}

We now provide formal notations on partial orders. So far, we have given two representations of a partial order, as a partially ordered set, and as a transitively closed DAG\footnote{Note in this paper we present the transitive reduction of a partial order for easier visualisation. }. For computation it is convenient to code a partial order $h\in\H_{[n]}\footnote{We use the same symbols $h$ and $\H_{[n]}$ for a partial order and its space in all three representations;
%represent both a realisation of partial order $H\in\H_{[n]}$ and a DAG (transitive reduction and transitive closure). 
we use $H$ to represent a random variable taking values in $\H_{[n]}$.} ,\ n\ge 1$ as a binary adjacency matrix $h \in \{0, 1\}^{n \times n}$. The rows and columns of $h$ correspond to the actors $[n]$. We take $h_{i,i}=0,\ i\in [n]$ and set $h_{i,j} = 1$ if and only if $i \succ_h j$. Two actors $i,j \in [n]$ are \textit{incomparable} $i\|_hj$, if neither $i\prec_h j$ nor $i \succ_h j$, and in this case $h_{i,j}=h_{j,i}=0$. A partial order is called a \textit{total order} if no pair of actors is incomparable. An \textit{empty order} is a partial order where no order relation is observed between any actors. An actor $m\in [n]$ is a maximal element of partial order $h$ if there exists no $i \in [n]$ that has $i \succ_h m$ in $h$. We denote the set of maximal elements of $h$ as $\max(h)=\{m\in [n]: h_{i,m}=0,\ \forall i\in[n]\}$.

A \textit{linear extension} $l = (l_1, l_2, \dots, l_n)$ of partial order $h$ is a permutation over $[n]$ given that its order does not violate $h$, so for $1 \leq a < b \leq n$ we must have either $l_a\succ_h l_b$ or $l_a\|_h l_b$ (``higher status come first''). Let $\P_n$ be the set of all possible permutations of $1,...,n$. The set of linear extensions of $h$ is then $$\mathcal{L}[h] = \{l\in \P_n|h_{l_b,l_a} \neq 1, \, \forall 1\le a<b\le n\}.$$ \Fig~\ref{LEex} shows all three possible linear extensions of the partial order $h_0$ in \Fig~\ref{POex}. 

A \textit{sub-order} $h[o]$ of a partial order $h\in \H_{[n]}$ restricts $h$ to a subset $o \subseteq [n],\ o=\{o_1,...,o_m\}$ say. All the order relation in $h$ are inherited in $h[o]$ so the matrix representation is $h_{o,o}$ - we simply retain rows and columns $o_1,...,o_m$ of $h$. As an example, the suborder of the partial order in \Fig~\ref{POex} restricted to $o=\{2,3,5\}$ is shown in \Fig~\ref{subPO} with its linear extensions in \Fig~\ref{subLE}. A \textit{chain} of $h$ is a sub-order that is also a total order. It extracts a strictly ordered sequence in a partial order. The \textit{length} of a chain is the number of nodes in this suborder. The length of the longest chain(s) is called the \textit{depth}, $D(h)$ say, of partial order $h\in \H_{[n]}$, with $1\le D(h)\le n$. For example the depth of $h_0$ in \Fig~\ref{POex} is four.

\begin{multicols}{2}
    \begin{minipage}{0.9\linewidth}
    \centering
    \begin{tikzpicture}[thick,scale=1.07, every node/.style={scale=0.8}]
        \node[draw, circle, minimum width=.1cm] (2) at (-1, 0.4) {$2$};
        \node[draw, circle, minimum width=.1cm] (4) at (1, 0.4) {$3$};
        \node[draw, circle, minimum width=.1cm] (5) at (0, -1) {$5$};
        \draw[-latex] (2) -- (5);
        \draw[-latex] (4) -- (5);
    \end{tikzpicture}
    
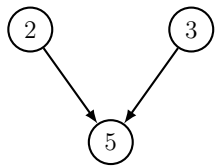
\captionof{figure}{A suborder $h_0[o]$ with actors $o=\{2,3,5\}$ of the partial order $h_0$ in \Fig~\ref{POex}.}\label{subPO}
    \end{minipage}%
    
    \begin{minipage}{0.9\linewidth}
    \centering
    \begin{tikzpicture}[thick,scale=0.8, every node/.style={scale=0.8}]
        \node[draw, circle, minimum width=.1cm] (2) at (0, 1) {$2$};
        \node[draw, circle, minimum width=.1cm] (4) at (0, 0) {$3$};
        \node[draw, circle, minimum width=.1cm] (5) at (0, -1) {$5$};
        \draw[-latex] (2) -- (4);
        \draw[-latex] (4) -- (5);
    \end{tikzpicture}
    \begin{tikzpicture}[thick,scale=0.8, every node/.style={scale=0.8}]
        \node[draw, circle, minimum width=.1cm] (2) at (0, 1) {$3$};
        \node[draw, circle, minimum width=.1cm] (4) at (0, 0) {$2$};
        \node[draw, circle, minimum width=.1cm] (5) at (0, -1) {$5$};
        \draw[-latex] (2) -- (4);
        \draw[-latex] (4) -- (5);
    \end{tikzpicture}
    
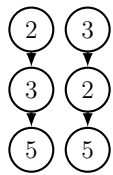
\captionof{figure}{All two linear extensions for the suborder $h_0[o]$ in \Fig~\ref{subPO}.}\label{subLE}
    \end{minipage}
\end{multicols}

The \textit{intersection order} $h^{int}$ of $K \in \mathbb{Z}^+$ permutations $l^{(k)}\in \P_n,\ k=1,...,K$ is a partial order $h^{int}\in \H_{[n]}$ such that for $i,j\in [n],\ i\ne j$ we have $i \prec_{h^{int}} j$ if and only if $i$ appears before $j$ in all $K$ permutations. If we think of the permutations as total orders then the intersection order shows all the order relations which are attested in every permutation. If a pair of actors appear in different orders in different permutations then there will be no relation between the actors in the intersection order. If $d^{(k)}_i=\{j\in [n]: l_j^{(k)}=i\}$ is the index of $i$ in list $k$, then the matrix representing the partial order defined by the intersection of lists satisfies $h^{int}_{i,j}=1$ if and only if $d^{(k)}_i<d^{(k)}_j,\ k=1,...,K$.

We define the \textit{dimension} of partial order $h$ to be the minimum number of linear extensions of $h$ whose intersection is $h$. If the observation model for the data is that the data are uniform random linear extensions of $h$, and the data is noise-free then the intersection-order is the maximum likelihood estimator of the true partial order  \citep{beerenwinkel2007conjunctive}. This works because the intersection-order has every order relation allowed by the data, and none which are not, so it is the partial order with the most order relations that admits every list in the data as a linear extension. It has the smallest number of linear extensions, and hence, maximises the likelihood for this observation model.\footnote{See \ref{appx-MLE} for proof.} In this setting the intersection order converges in probability to the true partial order as the number of sampled linear extensions goes to infinity.

\section{Prior Models for Random Partial Orders}\label{sec:po-prior}

Several probability distributions have been proposed for random partial orders in the combinatorics literature. See \cite{brightwell1993models} for an overview. Of these the most promising for development as a prior for Bayesian statistical work is \cite{winkler1985random}'s latent variable model. Following \cite{nicholls122011partial} we refer to this model as the $(K,n)$-partial order model. Before defining this model we introduce two prior properties which are desirable in our setting.

First of all in our application, and we expect this to be common, the depth $D(H)$ of the unknown true partial order $H$ is of particular interest. In some cases prior knowledge is given in terms of depth and similarly hypotheses are often statements about partial order depth. We seek a prior $H\sim \pi_{[n]}(\cdot),\ H\in \H_{[n]}$ for $H$ with the property that $H$ has approximately uniform depth distribution, $D(H)\sim \mathcal{U}\{1,...,n\}$. This is done so that the prior is non-informative with respect to hypotheses concerning depth. However, the analyst may use this degree of freedom (the prior for $\rho$, latent dimension and ties) to impose some non-uniform prior weight on the depth.

Secondly, a family of priors $\pi_{[n]}(h)$ is \textit{marginally consistent} or \textit{projective} if every marginal of every distribution in the family is also in the family. For every $n\ge 2$ and every $o,o'\subseteq [n]$ with $o \subset o'$, marginal consistency holds if 
\begin{equation}
    \pi_o(h)=\sum_{\substack{h'\in\H_{o'}\\h'[o]=h}} \pi_{o'}(h')\quad \mbox{for all $h\in \H_{o}$}.
\end{equation}\label{MarC}
Physical considerations suggest that our prior beliefs about relations between a subset of actors should be marginally consistent. We are assuming that the relation between any pair of actors in a social hierarchy is not affected by the presence or absence of any other actor.

\subsection{Properties of the uniform distribution on Partial Orders}
We illustrate a ``natural'' family of priors which is {\it not} projective and is additionally strongly informative on depth. 
It may seem appealing to take as prior the uniform distribution over partial orders.

\begin{definition}[The Uniform Distribution over Partial Orders]

For partial order $h\in \H_{[n]}$, the uniform distribution over partial orders $\pi_{u,[n]}(h)$ say, assigning equal prior probability to each partial order in $\H_{[n]}$, is
\begin{equation}
    \pi_{u,[n]}(h) = \frac{1}{|\H_{[n]}|}.
\end{equation}
    
\end{definition}

First of all the uniform priors $\pi_{u,[n]},\  n\ge 1$ have low sampling depth. As is shown in \cite{kleitman1975asymptotic}, if $h\sim \pi_{u,[n]}$ then $Pr(D(h)=3)\to 1$ as $n\rightarrow\infty$. This surprising result (why 3?) rules out the uniform distribution as a prior for most applied work since, even at small $n$, it favours partial orders of low depth, and has no parameter we can use to control the distribution over depth. 

Secondly, the family of uniform priors $\pi_{u,[n]},\  n\ge 1$ is not projective \citep{winkler1985random}. For a counter-example, let $h_2=\{V_2,E_2\}\in\H_{[2]}$ be a partial order on $2$ actors with $V_2=\{1,2\}$ and edge set $E_2$. There are $3$ choices for $E_2$: $E_2 = \emptyset$ (the empty partial order), $E_2 = \{(1,2)\}$ and $E_2 = \{(2,1)\}$. If there are three actors $h=(V_3,E_3)\in\H_{[3]}$ with $V_3=\{1,2,3\}$ then there are 19 possible edge sets $E_3$, depicted in the \ref{appx-edge}. In equation \ref{MarC}, $o=[2]$, $o'=[n]=[3]$ and 
$$\text{LHS} = \pi_{u,[2]}(h)=\frac{1}{3}$$ 
while 
\begin{align*}
\text{RHS} &= \sum_{\substack{h'\in\H_{[3]}\\h'[o]=h}} \pi_{u,[3]}(h') \\
&=\frac{|\{h'\in\H_{[3]}:  h'[o]=h\}|}{19}\neq \text{LHS}, 
\end{align*}
where $|A|$ denotes the cardinality of set $A$, and the last step holds because we cannot divide 19 into 3 groups of equal size. This counter-example holds between $n=2,3$ and extends straightforwardly to many larger $n$. The uniform distribution over partial orders is therefore not projective.

\subsection{The $(K,n)$-Partial Order Model}
 
The $(K,n)$-partial order model is defined in \cite{winkler1985random} and extended in \cite{nicholls122011partial}. This generative model adapts a latent matrix approach to represent a random partial order. The $(K,n)$-partial order model proposes a parameter $K\in \mathbbm{Z}^+$ - the fixed column dimension parameter in a random latent matrix $Z \in {[0,1]}^{n\times K}$ - to control the depth distribution. 
\begin{definition}[The $(K,n)$-Partial Order Model \citep{winkler1985random}]
    Let $Z \in \mathbb{R}^{n\times K}$ be a random latent matrix with fixed column dimension $K\in\mathbbm{Z}^+$. The $(K,n)$-partial order model defines the entries of $Z$ as $$Z_{i,k} \overset{i.i.d.}{\sim} \mathcal{U}[0,1], \forall (i,k)\in [n]\times [K]. $$
    Write $Z_{i,{1\dv K}}\succ Z_{j,{1\dv K}}$ if $Z_{i,k} > Z_{j,k} \, \forall k \in \{1, \dots, K\}, i,j \in [n]$. We define mapping $h:{[0,1]}^{n\times K}\to \H_{[n]}$ such that 
    \begin{equation}\label{eq:hz}
        h(Z)=\{h\in \H_{[n]}:i \succ_h j \text{ if and only if } Z_{i,{1\dv K}}\succ Z_{j,{1\dv K}} \}.
    \end{equation}
    (The original paper by \cite{winkler1985random} took $K=2$).
\end{definition}
The prior for partial orders under the $(K,n)$-partial order model is 
\begin{equation}
    \pi_{[n]}(h)=\int_{[0,1]^{n\times K}}\mathbbm{1}_{\{h(Z)=h\}}\pi(Z)dZ.
\end{equation}
The integral in $\pi(h)$ would be awkward. However, \cite{nicholls122011partial} carry out inference using the latent variable $Z$. This is in general non-identifiable, but this is not a concern as $h$ is the parameter of interest. \cite{nicholls122011partial} call $Z_{i,{1\dv K}}\in [0,1]^K$ the ``latent path'' (over the column index) for actor $i$. The condition $Z_{i,{1\dv K}}\succ Z_{j,{1\dv K}}$ means the paths are non-crossing and $Z_{i,{1\dv K}}$ lies entirely above $Z_{j,{1\dv K}}$. \Fig~\ref{fig:latentp1} shows some possible latent paths that generate the partial order $h_0$ in \Fig~\ref{POex}. 
The mapping $h:{[0,1]}^{n\times K}\to \H_{[n]}$ is surjective: one partial order can be represented by many latent matrices; however, each latent matrix uniquely determines a partial order. 

\begin{figure}[h!]%% placement specifier
\centering%% For centre alignment of image.
\includegraphics[width=0.7\linewidth]{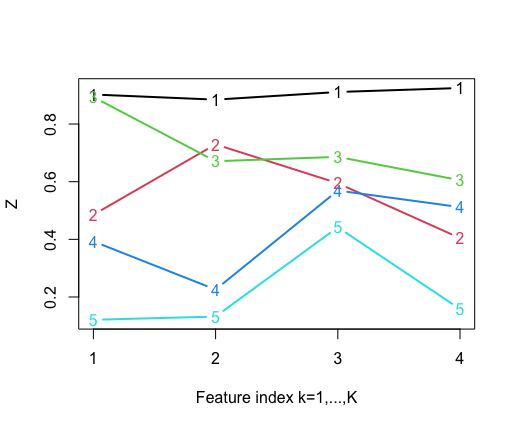}
%% Use \caption command for figure caption and label.
\caption{Latent paths of a random $Z$ with $K = 4, n=5$ defined using the $(K,n)$-partial order model and yielding the partial order $h_0$ in Figure~\ref{POex}.}
\label{fig:latentp1}
\end{figure}

Any partial order can be realised under this prior if we take $K$ at least $\lfloor n/2\rfloor$ \citep{watt2015inference,nicholls2023bayesian}. This can be seen in an equivalent construction of the prior. For $j\in [K]$, let $l^{(j)}(Z)=(R_{1,j},...,R_{n,j})$ give the rank-vector of the $Z$-entries in column $j$, so $R_{a,j}=\sum_{i'=1}^n\mathbbm{1}_{\{Z_{i',j}\ge Z_{a,j}\}}$. Since these are exchangeable, $l^{(j)}\sim \mathcal{U}(\P_{[n]})$ is a uniform random permutation of $[n]$. As discussed in \cite{watt2015inference}, the rule defining $h=h(Z)$ is identical to taking $h$ to be the intersection order of $l^{(j)}$ over $j=1,...,K$ (thinking of these as total orders). It follows that any partial order can be realised if we take $K$ at least $\lfloor n/2\rfloor$: a partial order $h$ on $n$ nodes can always be expressed as the intersection of $\lfloor n/2\rfloor$ total orders \citep{Hiraguchi51},  simply repeating total orders if fewer are needed, and any set of rank vectors $l^{(j)}$ over $j=1,...,K$ can be generated with non-zero probability. %This point is discussed in \cite{brightwell1993models} and elaborated in \cite{watt2015inference} who show that $K\ge \lfloor n/2\rfloor$ is actually sufficient.

\cite{winkler1985random} shows that the $(K,n)$-partial order model is projective.
This is intuitively obvious:
removing an actor is equivalent to removing one of the paths in Figure~\ref{fig:latentp1}. Since the paths are independent, and removing one doesn't alter the relations between the paths which remain, the probability to get any particular partial order on $n-1$ labeled actors is the marginal of the probability to get that partial order on any set of $n$ labeled actors containing the original $n-1$ actor labels. This kind of marginal consistency is typical for latent variable models for networks and graphs.

We now state and prove this result, partly to show how the notation works in this simple case. The proof for the more interesting case of the prior with ties below follows a similar path.
\begin{proposition}\label{prop:projective-kn}
The $(K,n)$-partial order model distributions are projective.
\end{proposition}
Proof: See \ref{appx-kr-projectivity}.

Although projective, the $(K,n)$-partial order model offers limited control over depth. In this model, $K$ controls the typical depth $D(h(Z))$ of the partial order. When $K=1$, $\Pr(D(h(Z))=n)=1$ since there are no paths to cross. At small $K\ge 1$ the probability that $Z_{i,{1\dv K}}\succ Z_{j,{1\dv K}}$ is relatively high, as paths are short, so an order between $i$ and $j$ is more likely and partial orders have greater depth. As $K$ increases, random $(K,n)$-partial orders tend to have lower depth, and $\lim_{K\to \infty}\Pr(D(h(Z))>1)= 0$. However, $K$ needs to be at least $\lfloor n/2\rfloor$ if we want $h(Z)$ to be able to realise any partial order in $\H_{[n]}$. At these $K$-values, $D(h(Z))\ll n$, typically. 

\subsection{The $(K,n,\rho)$-Partial Order Model}

% A latent variable $Z\in\mathbb{R}^{n\times K}$ and a mapping $h:\mathbb{R}^{n\times K} \to \H_{[n]}$. 

% \cite{watt2015inference} builds on from \cite{nicholls122011partial}'s model to propose the $(k,n,\rho)$-partial order model by adding a depth parameter that imposes more depth control over the realised distribution of random partial orders. This parameter manages creating approximately uniform depth distribution over partial orders. This is desirable as the depth of partial orders is of interest in many settings.

The issue of a depth distribution concentrated on small values is addressed in \cite{nicholls122011partial} and \cite{watt2015inference} where $K$ is fixed and a depth parameter $\rho\in [0,1)$ is introduced to control depth. This $(K, n, \rho)$-model for random partial orders is an extension to the $(K, n)$-model. It replaces the uniform random entries in $Z$ with a distribution which correlates entries within rows (but retains independence between rows). The $i$-th row of $Z$, $Z_{i,{1\dv K}}=(Z_{i,1},...,Z_{i,K})$, gives $K$ attribute-values for actor $i \in [n]$. If the entries in $Z_{i,{1\dv K}}$ are more strongly correlated, then the paths $Z_{i,{1\dv K}}$ are relatively ``flatter'', and the probability two paths do not cross is higher, so finally $h(Z)$ tends to have relatively higher depth. 

\begin{definition}[The $(K,n,\rho)$-Partial Order Model \citep{nicholls122011partial}]\label{def:knrho}
    Let $\Sigma_{\rho}$ be a $K\times K$ correlation matrix with $\Sigma_{\rho,i,i} = 1$ and $\Sigma_{\rho,i,j} = \rho,\ i \neq j$, $i,j\in [K]$. In the $(K,n,\rho)$-model
\begin{equation}
    Z_{i,{1\dv K}}|\rho \stackrel{i.i.d.}{\sim} \mathcal{MVN}(\mathbf{0}, \Sigma_\rho), \, i \in \{1, \dots, n\}. 
\end{equation}
Let 
\[\pi_Z(Z|\rho)=\prod_{i=1}^n \mathcal{MVN}(Z_{i,{1\dv K}};\mathbf{0}, \Sigma_\rho)\] 
denote the joint prior for $Z$ given $\rho$. The partial order $h=h(Z)$ is constructed according to map (\ref{eq:hz}). 
\end{definition}

There is considerable freedom in choosing the $Z$-distribution, but correlating entries within rows while maintaining independence across rows seems key to controlling depth whilst ensuring distributions are projective. In \Fig~\ref{fig:latentp2}, we show another set of latent paths for the partial order $h_0$ in Figure~\ref{POex} simulated under the $(K,n,\rho)$-partial order model with depth parameter $\rho=0.9$. 

%For $h\in \H_{[n]}$ we would like the marginal prior
%\[
%
%\]
We would like the marginal prior for $H$ determined by $\rho\sim\pi_\rho(\cdot)$, $Z\sim \pi_Z(\cdot|\rho)$ and $H=h(Z)$ 
\begin{equation}\label{prior-h}
    \pi_{[n]}(h)=\int_0^1 \int_{Z:h(Z)=h}\pi_Z(Z|\rho)\pi_\rho(\rho)dZd\rho,
\end{equation}
to have approximately uniform marginal depth distribution. Simulations reported in \cite{nicholls122011partial} suggest that $\rho\sim \mbox{Beta}(1,1/6)$ gives a reasonably
flat distribution for $D(H)$ for $n$ in the range of interest to us.

\begin{figure}[ht]%% placement specifier
\centering%% For centre alignment of image.
\includegraphics[width=0.7\linewidth]{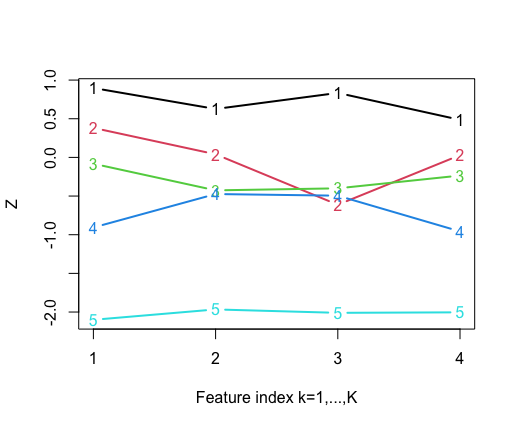}
%% Use \caption command for figure caption and label.
\caption{Latent paths of a random $Z$ defined using the $(K,n,\rho)$-partial order model for the partial order $h_0$ where $K = 4, n = 5$ and $\rho = 0.9$.}
\label{fig:latentp2}
\end{figure}

Like the $(K,n)$-model, the $(K, n,\rho)$-model can be expressed as an intersection of $K$ linear extensions, which are the rank-vectors of the $Z$-columns. Any set of rank vectors can be realised, but they are now correlated. The choice of $K$ remains an important model hyperparameter. We discuss below how $K$ can be treated as a parameter and estimated alongside the other parameters. %We sometimes fix $K = [n/2]$ as this is the least value known to be sufficient to recover any unknown true partial order \citep{watt2015inference}. 
Because the rows of $Z$ are conditionally independent given $\rho$ in the $(K, n, \rho)$-model, we can show that the family of priors $\pi_{[n]}(h|\rho)$ is projective for every fixed $\rho$, by the same reasoning given in Proposition~\ref{prop:projective-kn} for the $(K, n)$-model. 
\begin{proposition}
    The $(K,n,\rho)$-partial order distributions are projective. 
\end{proposition}
\begin{proof} 
In the notation of Equation~\ref{MarC}, if $o\subset o'\subseteq [n]$, $h'\sim \pi_{o'}(\cdot|\rho)$ and $h\sim \pi_o(\cdot|\rho)$ then the conditionals $H'[o]|\rho\sim H|\rho$ are same for every $\rho$ (the proof is essentially unchanged from Proposition~\ref{prop:projective-kn}), so the marginals $H'[o]\sim H$ are also equal in distribution.
\end{proof}
%We use simulation to check the prior distribution for depth and other features of the prior.
%

\subsection{The $(PDP,\rho)$-Partial Order Model}\label{Tie}

Depth control in the $(K, n, \rho)$-model can be inadequate in certain scenarios. In addition, its continuous latent matrix representation makes it impossible to model equality between actors. This paper extends the $(K, n, \rho)$-model to consider ties and to provide further control over depth distribution. We call this extended model \textit{the $(PDP,\rho)$-partial order model}. 

\subsubsection{Partial order with Ties}

Latent variable models are powerful and expressive tools for modelling partial orders. However, in the model we have defined the rows of $Z$ to be almost surely distinct. In some settings we may have prior information that some actors are equivalent in the sense that some groups of actors have the same order relations with other groups of actors. We don't know the groups, or how they are related, but we expect there are groups. %This is the case in a bucket order. 
As illustrated in the field marshal and generals example in Section \ref{sec:intro}, this scenario leads to a prior distribution over partial orders where we allow ties between some actors. We define a tied relation as follows.
\begin{definition}[Tie Relations in partial orders]
In a partial order with ties $h^*$, if two actors $i, j \in [n]$ are tied ($i \sim_{h^*} j$), then for any $v\in [n], v\neq i \neq j$,
\begin{enumerate}
    \item $i \succ_{h^*} v$ if and only if $j \succ_{h^*} v$;
    \item $i \prec_{h^*} v$ if and only if $j \prec_{h^*} v$;
    \item $i\|_{h^*} v$ if and only if $j\|_{h^*} v$.
\end{enumerate}
\end{definition}
In Section \ref{PO}, we represent a partial order $h$ using a binary adjacency matrix $h \in \{0,1\}^{n\times n}$, such that $h_{i,j}=1$ if and only if actor $i\succ_h j$. We extend this notation so that if there is a tie between actors $i$ and $j$ in $h^*$, then $h^*_{i,j}=h^*_{j,i}=1$. We use $\H_{[n]}^*$ to represent the class of \textit{partial orders with ties}. The original partial order class $\H_{[n]} \subseteq \H_{[n]}^*$.
\begin{definition}[Binary Adjacency Representation for Partial Orders with Ties]
A partial order with ties $h^*\in\H_{[n]}^*$ can be represented by a binary adjacency matrix $h^* \in \{0,1\}^{n\times n}$, where for $i, j \in [n], i\neq j$,
    \begin{enumerate}
        \item $i \succ_{h^*} j \text{ if and only if } H_{i,j}=1 \text{ and } H_{j,i}=0$;
        \item $i \sim_{h^*} j \text{ if and only if } H_{i,j}=H_{j,i}=1$; and 
        \item $i \:\|_{h^*}\: j \text{ if and only if } H_{i,j}=H_{j,i}=0$. 
    \end{enumerate}
\end{definition}
As an example, if we take the partial order $h_0$ in \Fig~\ref{POex} and tie actors $3 \sim_{h_0^*} 4$ the resulting tied partial order is shown in \Fig~\ref{POTex}. The new partial order with ties $h_0^*$ is shallower (the depth is 3) and has more linear extensions as is shown in \Fig~\ref{LETex}. 
% Figure 4 & 5
\begin{multicols}{2}
    \begin{minipage}{0.9\linewidth}
    \centering
    \begin{tikzpicture}[thick,scale=1.07, every node/.style={scale=0.8}]
        \node[draw, circle, minimum width=.1cm] (1) at (0, 1) {$1$};
        \node[draw, circle, minimum width=.1cm] (2) at (-1, -0.5) {$2$};
        \node[draw=blue!60, rectangle, minimum width=2cm,minimum height=1cm] (6) at (1, -0.5) {};
        \node[draw, circle, minimum width=.1cm] (3) at (0.6, -0.5) {$3$};
        \node[draw, circle, minimum width=.1cm] (4) at (1.4, -0.5) {$4$};
        \node[draw, circle, minimum width=.1cm] (5) at (0, -2) {$5$};
        \draw[-latex] (1) -- (2);
        \draw[-latex] (2) -- (5);
        \draw[-latex] (1) -- (6);
        \draw[-latex] (6) -- (5);
    \end{tikzpicture}
    
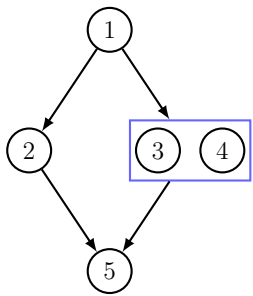
\captionof{figure}{A tied partial order of partial order $h_0$ with cluster $S = (1,2,(3,4),5)$.}\label{POTex}
    \end{minipage}%
    
    % modify figure to the Z matrix potentially
    \begin{minipage}{0.9\linewidth}
    \centering
    \begin{tikzpicture}[thick,scale=0.8, every node/.style={scale=0.8}]
        \node[draw, circle, minimum width=.1cm] (1) at (0, 2) {$1$};
        \node[draw, circle, minimum width=.1cm] (2) at (0, 1) {$2$};
        \node[draw, circle, minimum width=.1cm] (3) at (0, 0) {$3$};
        \node[draw, circle, minimum width=.1cm] (4) at (0, -1) {$4$};
        \node[draw, circle, minimum width=.1cm] (5) at (0, -2) {$5$};
        \draw[-latex] (1) -- (2);
        \draw[-latex] (2) -- (3);
        \draw[-latex] (3) -- (4);
        \draw[-latex] (4) -- (5);
    \end{tikzpicture}
    \begin{tikzpicture}[thick,scale=0.8, every node/.style={scale=0.8}]
        \node[draw, circle, minimum width=.1cm] (1) at (0, 2) {$1$};
        \node[draw, circle, minimum width=.1cm] (2) at (0, 1) {$3$};
        \node[draw, circle, minimum width=.1cm] (3) at (0, 0) {$2$};
        \node[draw, circle, minimum width=.1cm] (4) at (0, -1) {$4$};
        \node[draw, circle, minimum width=.1cm] (5) at (0, -2) {$5$};
        \draw[-latex] (1) -- (2);
        \draw[-latex] (2) -- (3);
        \draw[-latex] (3) -- (4);
        \draw[-latex] (4) -- (5);
    \end{tikzpicture}
    \begin{tikzpicture}[thick,scale=0.8, every node/.style={scale=0.8}]
        \node[draw, circle, minimum width=.1cm] (1) at (0, 2) {$1$};
        \node[draw, circle, minimum width=.1cm] (2) at (0, 1) {$3$};
        \node[draw, circle, minimum width=.1cm] (3) at (0, 0) {$4$};
        \node[draw, circle, minimum width=.1cm] (4) at (0, -1) {$2$};
        \node[draw, circle, minimum width=.1cm] (5) at (0, -2) {$5$};
        \draw[-latex] (1) -- (2);
        \draw[-latex] (2) -- (3);
        \draw[-latex] (3) -- (4);
        \draw[-latex] (4) -- (5);
    \end{tikzpicture}
    \begin{tikzpicture}[thick,scale=0.8, every node/.style={scale=0.8}]
        \node[draw, circle, minimum width=.1cm] (1) at (0, 2) {$1$};
        \node[draw, circle, minimum width=.1cm] (2) at (0, 1) {$2$};
        \node[draw, circle, minimum width=.1cm] (3) at (0, 0) {$4$};
        \node[draw, circle, minimum width=.1cm] (4) at (0, -1) {$3$};
        \node[draw, circle, minimum width=.1cm] (5) at (0, -2) {$5$};
        \draw[-latex] (1) -- (2);
        \draw[-latex] (2) -- (3);
        \draw[-latex] (3) -- (4);
        \draw[-latex] (4) -- (5);
    \end{tikzpicture}
    \begin{tikzpicture}[thick,scale=0.8, every node/.style={scale=0.8}]
        \node[draw, circle, minimum width=.1cm] (1) at (0, 2) {$1$};
        \node[draw, circle, minimum width=.1cm] (2) at (0, 1) {$4$};
        \node[draw, circle, minimum width=.1cm] (3) at (0, 0) {$2$};
        \node[draw, circle, minimum width=.1cm] (4) at (0, -1) {$3$};
        \node[draw, circle, minimum width=.1cm] (5) at (0, -2) {$5$};
        \draw[-latex] (1) -- (2);
        \draw[-latex] (2) -- (3);
        \draw[-latex] (3) -- (4);
        \draw[-latex] (4) -- (5);
    \end{tikzpicture}
    \begin{tikzpicture}[thick,scale=0.8, every node/.style={scale=0.8}]
        \node[draw, circle, minimum width=.1cm] (1) at (0, 2) {$1$};
        \node[draw, circle, minimum width=.1cm] (2) at (0, 1) {$4$};
        \node[draw, circle, minimum width=.1cm] (3) at (0, 0) {$3$};
        \node[draw, circle, minimum width=.1cm] (4) at (0, -1) {$2$};
        \node[draw, circle, minimum width=.1cm] (5) at (0, -2) {$5$};
        \draw[-latex] (1) -- (2);
        \draw[-latex] (2) -- (3);
        \draw[-latex] (3) -- (4);
        \draw[-latex] (4) -- (5);
    \end{tikzpicture}
    
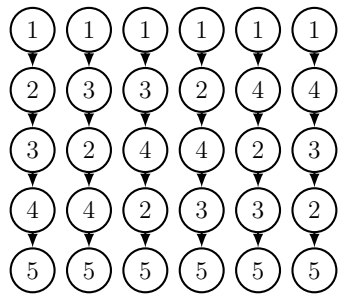
\captionof{figure}{All six linear extensions for the tied partial order in \Fig~\ref{POTex}.} \label{LETex}                                                           
    \end{minipage}
\end{multicols}

Notice the difference between ``equal'' 
$\sim_h$ and ``unordered'' $\|_h$. 
If we remove the tie in \Fig~\ref{POTex} to get a new partial order in which $3$ and $4$ are unordered (but all other relations are unchanged) then the linear extensions would be unchanged from those in \Fig~\ref{LETex}. 
Since the linear extensions are the data, the new and old partial orders would not be identifiable. So here there is no difference. However, we cant simply take all unordered pairs in a partial order and replace them with ties, as tied actors must have identical relations to other actors while unordered actors need not. On the other hand, as we saw, if we replace tied relations with unordered relations (maintaining relations to other actors) then the resulting partial order wont be identifiable from the tied order. From the perspective of the observation model, ties dont really change much. However, introducing ties changes the prior: we should think of the ties model as a way to put more weight in our prior on partial orders in which relations are between groups rather than between individuals. 

We now give a generative model for random partial orders with ties, which we call the $(PDP,K,\rho)$-partial order model with fixed latent column dimension $K$. This model involves partitioning the actors into tied clusters and assigning probability to a random partial order with the clusters as nodes. In other respects the same latent variable setup is used. This ensures marginal consistency and effective control over depth. Actor-partitioning is achieved using a Poisson-Dirichlet process mixture model over realisations of the rows of $Z$, a generalisation of the Dirichlet process to two parameters.  The Poisson-Dirichlet process groups the tied actors into clusters with equal $Z$-paths. Define partition $\S=(S_1, S_2, \dots, S_C)\in \Xi_{[n]}$ where $C=|\S|$ is the number of clusters and $\Xi_{[n]}$ is the set of all partitions of $[n]$.  Let $\eta = (\eta_a, \eta_b)$ be the discount and strength parameters respectively. The process assigns a prior distribution over $S$ according to
\begin{equation}
    P_{\eta,[n]}(\S)=\frac{\Gamma(\eta_b)}{\Gamma(\eta_b+n)}\frac{\eta_a^C\Gamma(\eta_b/\eta_a+C)}{\Gamma(\eta_b/\eta_a)}\prod_{c=1}^C \frac{\Gamma(n_c-\eta_a)}{\Gamma(1-\eta_a)},
\end{equation}
where $n_c=|S_c|$ denotes the number of actors in cluster $S_c$. Both $\eta_a \in [0,1)$ and $\eta_b>-\eta_a$ are constants. 

\begin{definition}[The $(PDP,K,\rho)$-Partial Order Model]
    Let $\mathbf{S}\sim PDP(\eta_a,\eta_b)\in\Xi_{[n]}$ and $|\mathbf{S}|=C$. Define latent matrix $Z^*\in \mathbb{R}^{C \times K}$, with distribution
    \begin{equation}
        Z^*_{c,1:K}|\rho\overset{i.i.d.}{\sim} \mathcal{MVN}(\mathbf{0},\Sigma_\rho), c\in [C], 
    \end{equation}
    and $\Sigma_\rho$ as defined in Definition \ref{def:knrho}. For $i\in [n]$, let $c(i; S)=\{c\in[C]: i\in S_c\}$ be the cluster containing $i$. Define map $Z:\mathbb{R}^{C\times K}\to\mathbb{R}^{n\times K}$,
\begin{equation}\label{eq:unpackZtoZstar}
Z(S,Z^*)=(Z^*_{c(i; S),{1\dv K}})_{i=1}^{n}.    
\end{equation}
We extend the map $h$ to handle ties as follows.
Let 
\[
G(S,Z^*)=\{i\succ_{h^*} j \text{ iff } Z(S,Z^*)_{i,k}> Z(S,Z^*)_{j,k},\ \forall\ k\in[K]\}_{i,j\in [n]}
\]
and
\[
E(S,Z^*)=\{i\sim_{h^*} j \text{ iff } Z(S,Z^*)_{i,k}= Z(S,Z^*)_{j,k},\ \forall\ k\in[K]\}_{i,j\in [n]}
\]
Then
\begin{align}
    h^*(S,Z^*) & =([n], G(S,Z^*)\cup E(S,Z^*)). 
\end{align}
This reduces to the previous definition in the case of a partial order without ties.
\end{definition}

Each row of $Z^*$ represents a cluster. Actors from the same cluster have the same attributes in the latent matrix $Z$. The $(PDP,K,\rho)$-partial order model therefore gives a probability distribution over the class of partial order with ties $\H^*_{[n]}$. 
%It would be natural to define the map directly from the space of $(Z^*,S)$ to $\H^*_n$ but the proposed rule is simpler to express and almost surely the same for continuous $Z^*$.

\subsubsection{Estimating the Latent Matrix Dimension $K$}

We now introduce a second generalisation of the model. This was motivated by our desire to get good control of the depth distribution in the context of a model with ties. We have two parameters influencing depth: $K$ and $\rho$.
The $(K,n,\rho)$ model took $K$ fixed and took a prior distribution for $\rho$ which gave an approximately uniform distribution for $D(H)$.
We found, in our $(PDP,K,\rho)$-partial order model, that we could not get a good control of depth by varying $\rho$, $\eta_a$ and $\eta_b$ alone. We now allow $K$ to vary. We take a prior $K\sim \pi_K(\cdot)$. Small $K$ promotes deeper partial orders.
%Having defined a new probability distribution over partial orders with ties, it is important to check it is marginally consistent, gives good control over depth and to display any other distinctive properties of the distribution it determines over partial orders.
% Grouping actors using the Dirichlet process gives a remarkable advantage as it resembles the data generating process and reduces the mixing time.
%Simulation experiments showed that the depth parameter $\rho$ itself cannot sufficiently govern the depth distribution to be approximately uniform. We further impose a random variable $k\in\mathbb{Z}^+$ as the column dimension for $Z$. Both a decreased $k$ and an increased $\rho$ favour higher partial order depth. We vary both $k$ and $\rho$. The reduced latent matrix $Z^*\in \mathbbm{R}^{C\times k}$ therefore follows
\begin{definition}[The $(PDP,\rho)$-Partial Order Model]\label{def:gen-proc-ties-model}
The full generative model for a random partial order with ties is as follows:
$K\sim \pi_K(\cdot)$, $\rho\sim \pi_\rho(\cdot)$, $S\sim P_{\eta,[n]}(\cdot)$ (with hyperparameter $\eta=(\eta_a, \eta_b)$), 
\begin{equation}\label{eq:Zstar-prior-given-rks}
    Z^*|\rho,K,S \sim \prod_{c=1}^C\mathcal{MVN}(Z^*_{c,1\dv K}; \mathbf{0}, \mathbf{\Sigma}_\rho), 
\end{equation}
with $\mathbf{0}\in \mathbb{R}^K$ and $\mathbf{\Sigma}_\rho$ as defined above (ie, variance one, covariance $\rho$) and $H^*=h^*(S,Z^*)$.
\end{definition}

We now write down the marginal distribution for $H^*$ generated by the process in Definition~\ref{def:gen-proc-ties-model}. For $h^*\in \H^*_{[n]}$ and $S\in \Xi_{[n]}$ let
\begin{equation}\label{eq:Zstar-space-given-hstar}
Z^*[h^*;S]=\{Z^*\in \mathbb{R}^{C\times K}: h^*(S,Z^*)=h^*\}
\end{equation}
be the set of $Z^*$ matrices which give $h^*$ for a given actor partition $S$.
The tie-model for random partial orders determines a prior distribution $\pi_{[n]}(h^*)$ over $h^*\in \H_{[n]}^*$. The conditional is
\begin{align}\label{candprior}
    \pi_{[n]}(h^*|\rho,K) = \sum_{S\in \Xi_n}  \left[ \int_{Z^*[h^*;S]} \pi(Z^*|\rho,K,S) dZ^*\right] P_{\eta,[n]}(S),
\end{align}
so that marginally,
\begin{equation}\label{eq:marginal-prior-ties}
    \begin{split}
        \pi_{[n]}(h^*) =  & \sum_{k=1}^\infty  \left[\int_0^1 \pi_H(h^*|\rho,K=k) \pi_\rho(\rho)  d\rho. \right] \pi_K(k).
    \end{split}
\end{equation}
The prior $\pi_{[n]}(h^*)$ is not tractable but can be explored by simulating the generative model. When we come to fit the model we work in the latent representation $(Z^*,S)$. 

%We give priors for the hyperparameters $K$ and $\rho$, so our prior for $h$ is determined implicitly by $\pi(Z|\rho,k,S)\pi(\rho)\pi(k)P_{\eta,[n]}(S)$. 
We choose the priors $\pi_K(k)=\mbox{Geo}(k,1,\eta_K)$ (with $k\ge 1$) and $\pi(\rho)=\mbox{Beta}(\rho;1,\eta_\rho)$. The hyperparameters $\eta_K$ abd $\eta_\rho$ are chosen by experiment so that the prior predictive distribution for the depth, $D(h^*(S,Z^*))$ determined by taking $\rho\sim \pi_\rho(\cdot)$, $K\sim \pi_K(\cdot)$, $S\sim P_{\eta,[n]}(\cdot)$ and $Z^*\sim \pi(\cdot|\rho,K,S)$ is reasonably flat. The prior distribution for partial order depths for the case $n=15$ ($\eta_a = 0.7, \eta_b=3, \eta_\rho=1/6, \eta_K=0.05$) is shown in \Fig~\ref{fig:depthp}. It is not perfectly uniform but this is not a concern. We simply need to avoid exponentially large weightings which can arise in this setting so that the prior at least ``allows'' all depths. Any further re-weighting can be made in the analysis using for example Bayes factors. In other settings the priors may need to be adjusted.

\begin{figure}[h!]%% placement specifier
\centering%% For centre alignment of image.
\includegraphics[height=7cm]{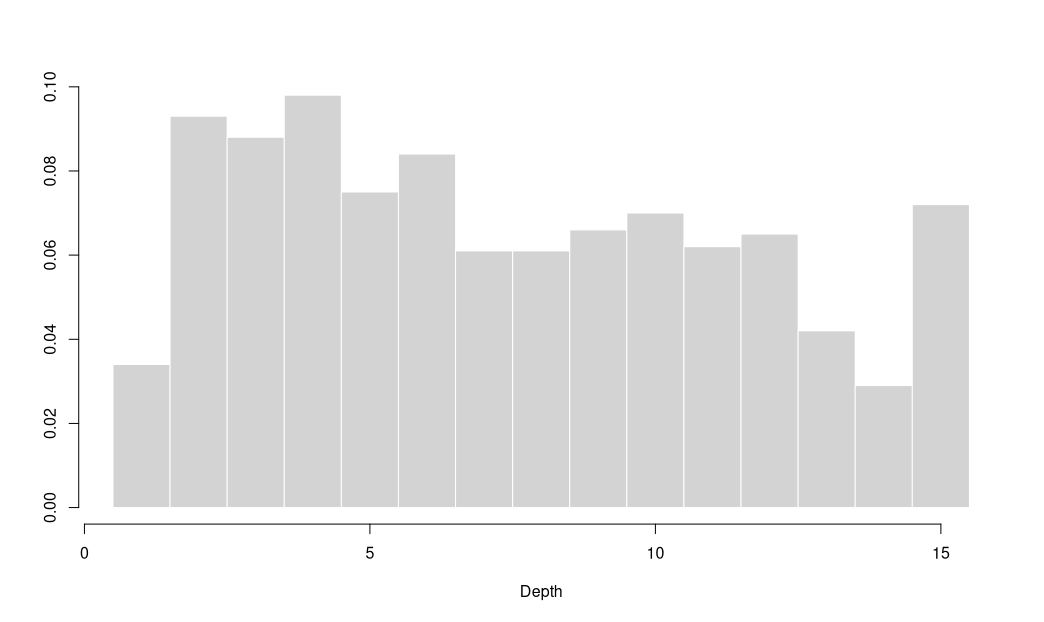}
%% Use \caption command for figure caption and label.
\caption{The prior distribution for partial order depths with $n=15$ actors (with hyperparameters $\eta_a = 0.7, \eta_b=3, \eta_\rho=1/6, \eta_K=0.0625$). }
\label{fig:depthp}
\end{figure}

The $(PDP,\rho)$-partial order model is projective.
This follows because the Poisson-Dirichlet process generating $Z$ is projective. The proof is very similar to that for the $(K,n)$-model in Proposition~\ref{prop:projective-kn}. 

\begin{proposition}\label{prior-mc}
For $h^*\in \H^*_{[n]}$ let $\pi_{[n]}(h^*)=\Pr(H^*=h^*)$ be the distribution of the random partial order with ties generated by the process in Definition~\ref{def:gen-proc-ties-model} and given in Equation~\ref{eq:marginal-prior-ties}.
The family of distributions $\pi_{[n]}, n\ge 1$ is projective.
\end{proposition}

Proof: see \ref{appx-projectivity}.

\section{Observation Models for list data}\label{sec:obs}

Suppose our data contains $N$ lists,  $\mathbf{y} = (y_1, \dots, y_N)$. Let $$o_i=(o_{i,1}, o_{i,2},...,o_{i,n_i}),\ i=1,...,N,$$ with $o_{i,j}<o_{i,j+1}, j=1,...,n_i-1$ be a subset $o_i\subseteq [n]$ of $n_i\ge 2$ actor labels recording the actors present in list $i$. Let $\P_{o_i}$ be the set of all permutations of $o_i$. The lists $y_i\in\P_{o_i}$ record the order in which the actors appear in the list. Note the setup: the list $y_i$ is not formed by sampling a distribution on permutations of all $n$ actors and then reducing it to a shorter list containing only those in $o_i$; rather the list is formed from the start using a distribution over permutations of the actors in $o_i$ only. We condition on the values of $o_i,\ i=1,...,N$ throughout. In the context of preference orders, assessors are presented with a choice set $o_i$ and return an ordered preference list $y_i$ which ranks all the elements in the choice set. We will be interested in the distribution of the random lists $y_i$ given the label content $o_i$. We think of the lists $y_i$ as total orders given by the sequence in which the labels appear from highest to lowest in the order. In the $i$-th data-list, $y_{i,1}$ came first and $y_{i,n_i}$ last and this will be evidence that actor $y_{i,1}$ is in some sense more important than actor $y_{i,n_i}$.

\subsection{Noise-Free Model}\label{NF}

In our observation model for noise-free lists we take $y_i$ to be a linear extension sampled uniformly at random from the set of all linear extensions of the suborder $H[o_i]$ of unknown true partial order $H$, independently for $i=1,...,N$. In this section we write down the likelihood and motivate this choice.

Suppose $H=h$ for some $h\in \H_{[n]}$. Let $o_j\subseteq [n]$ be some generic subset of actor-labels, $j=1,\dots,N$. The set of all linear extensions of the sub-order $h[o_j]$ is $\mathcal{L}[h[o_j]]$.  Let $y_j\in \mathcal{L}[h[o_j]]$ be some generic noise free `data', an ordered list of the actor-labels in $o_j$. 
The generative model in the noise-free case is 
\begin{equation}\label{noise-free}
    p^{(P)}(y_j|h[o_j])=\frac{\mathbbm{1}_{y_j\in \mathcal{L}[h[o_j]]}}{|\mathcal{L}[h[o_j]]|}.
\end{equation}

This model is motivated by thinking of the list as recording the place of actors $(o_{j,1},...,o_{j,n_j})$ in a queue. The queue forms, and then neighboring pairs of actors in the list swap places at random so long as the exchange does not violate any order relation in $h[o_j]$, until the moment the queue is recorded. This process defines an irreducible Markov chain of random lists converging in distribution to the uniform distribution over linear extensions of $h[o_j]$. 

The following result is well known, for example, it appears in \cite{Karzanov91} who show a closely related chain is rapidly mixing.
\begin{proposition}\label{Prop:MC-unif}
Consider a Markov chain $\{X_t\}_{t\ge 1}$, with state space $X_t\in \mathcal{L}[h[o_j]]$ for some fixed $j\in\{1,\dots,N\}$. Suppose that at step $t$ we have $X_t=x$. An entry $i\sim U\{1,...,n_j\}$ is chosen uniformly at random. If $i=n_j$ then we reject and set $X_{t+1}=x$ and otherwise $x'=(x_1,...,x_{i-1},x_{i+1},x_i,x_{i+2},...,x_{n_j})$.
If $x'\in \mathcal{L}[h[o_j]]$ then we set $X_{t+1}=x'$ and otherwise $X_{t+1}=x$.
The process converges in distribution to the uniform distribution over linear extensions of $h[o_j]$, that is, for $l\in \L[h[o_j]]$,
\begin{equation*}
    p^{(P)}(X_n=l)\overset{t\to\infty}{\to}\frac{1}{|\mathcal{L}[h[o_j]]|}\mathbbm{1}_{l\in \mathcal{L}[h[o_j]]}. 
\end{equation*}
\end{proposition}

Proof: see \ref{appx-MC}. 

\subsection{Error Models}

Obtaining error-free data is obviously ideal. However measurement/recording errors are often present. In order to incorporate noise in the observation model, \cite{nicholls122011partial} propose a simple queue-jumping error model that allows an actor to jump up the queue with probability $p\in[0,1]$, ignoring any order constraints. 

\subsubsection{Queue-Jumping Error Model}

Let $L(h)=|\mathcal{L}[h]|$ be the number of linear extensions of $h$ and let $L_i(h)=|\{l\in\mathcal{L}[h]:l_1=i\}|$ give the number of linear extensions headed by actor $i$. The queue-jumping (upwards) model (\cite{nicholls122011partial}) defines
\begin{equation}\label{Queue-jumping}
   p^{(Q)}(y_j|h,p)=\prod_{i=1}^{n_j-1}\left(\frac{p}{n_j-i+1}+(1-p)\frac{L_{y_i}(h[y_{j,i:n_j}])}{L(h[y_{j,i:n_j}])}\right). 
\end{equation}
We can interpret this as modelling the process by which the list is formed. The process is as follows:
\begin{enumerate}
    \item[1.] Set $i=1$, $s=[n]$ and $h'=h$. 
    \item[2.] Let $q= (L_j(h')/L(h'))_{j\in s}$ be a probability vector weighted by the numbers of linear extensions headed by each actor. 
    \item[2] With probability $p$ sample $Y_i\sim \mathcal{U}(s)$ (choose one of the remaining actors at random) and otherwise sample $Y_i\sim\mbox{multinom}(q)$.
    \item[3.] Set $s\leftarrow s\setminus {Y_i}$.
    \item[4.] If $s$ is empty return $Y_1,...,Y_n$ and otherwise set $i\leftarrow i+1$, $h'\leftarrow h[s]$ and go to step 2.
\end{enumerate}
The output $(Y_1,...,Y_n)\sim p^{(Q)}(\cdot|h,p)$ is a random list of $n$ elements distributed according to the queue-jumping model. This follows because the probabilities to choose entries in $Y$ at each step are just the factors in $p^{(Q)}$ in Equation~\ref{Queue-jumping}.
If we set $p=0$, then we get a telescoping product and $p^{(Q)}(l|h,p=0)=1/|\L(h)|$ for $l\in \L[h]$.
so the model reduces to the error-free observation model in Equation~\ref{noise-free}. We can turn the model around and build the list from the bottom, allowing ``queue jumping" downwards. \cite{jiang2023bayesian} proposes a `bi-directional' queue-jumping model that makes displacements in both directions plausible in the similar queue-like setting. If we denote the computational complexity of counting the number of linear extension of a partial order with $n$ nodes as $O_n$ (already at least exponential growth with $n$), the computational complexity of evaluating the `bi-directional' queue-jumping likelihood is at least $O(O_n\times2^n)$. This makes it formidable to evaluate the `bi-directional' queue-jumping likelihood on general partial orders, as opposed to its subclass - the vertex-series-parallel orders that admit fast linear extension counting - as considered in \cite{jiang2023bayesian}. 

In order to build a more tractable `bi-directional' error model, we introduce Mallows and Plackett-Luce noise-model variants, focusing on the Mallows case which we favor. See \ref{app:PL} for some remarks on Plackett-Luce and some pros and cons.

\subsubsection{The Mallows Model}\label{sec:mallows}

The Mallows model \citep{mallows1957non} assigns probability to a ranking $y$ based on its distance from a `central ranking' or `reference ranking' $l$. Consider the ``unknown true'' data $l \in \P_{[n]}$ as a linear extension of a generic partial order $h$ where $h\in \H_{[n]}$. We observe $y\in \P_{[n]}$ with ``Mallows noise" but ``centred'' on $l$
(if $y\in\P_o$ then $l\in\L_{h[o]}$ and we are observing a suborder $h[o]$ for some $o\subseteq [n]$, but the setup is the same with $h\to h[o]$ $etc$). In the Mallows-$\phi$ model, the probability $p(y|l,\theta)$ is given in terms of a symmetric divergence $d(l,y)\ge 0$. Let $\sigma(l,j)=\{k\in [n]:l_k=j\}$ and let the Kendall-tau distance be $d(y,l)$ so that
\[
d(y,l) = \sum_{i=1}^{n-1}\sum_{j=i+1}^n \mathbb{I}_{\sigma(l[y_{i:n}],y_{i})>\sigma(l[y_{i:n}],y_{j})}.
\] 
Denote the \textit{dispersion parameter} as $\theta\in (0,\infty)$. The Mallows $\phi$-model gives
\begin{align}
    p^{(M)}(y|l,\theta)=\frac{\exp({-\theta d(l,y)})}{\Psi_n(\theta)}, \quad y\in \P_{[n]},
    \label{eq:mallows_lkd_for_order}
\end{align}
where $\Psi_n(\theta)$
is a normalising constant available in closed for as $\Psi_n(\theta)=\sum_{z\in \P_{[n]}} e^{-\theta d(l,z)}=\prod_{i=1}^{n} \psi_i$
with
$\psi_i=\sum_{j=1}^i e^{-(i-1)\theta}$ under the Kendall-tau distance.
The Mallows-$\phi$ is a sequential choice model as
\begin{align}
    p^{(M)}(y|l,\theta)&=\prod_{i=1}^{n-1}\frac{\exp({-\theta\sum_{j=i+1}^n \mathbb{I}_{\sigma(l[y_{i:n}],y_i)>\sigma(l[y_{i:n}],y_j)}})}{\psi_{n-i+1}(\theta)}
    \label{eq:mallows-select-first}\\
    &=\prod_{i=1}^{n-1} q^{(M)}(y_i|l[y_{i:n}],\theta),
    \label{eq:mallows_lkd_for_order_seq}
\end{align}
where $q^{(M)}(y_i|l[y_{i:n}],\theta)$ is the probability $y_i$ is selected next from the remaining choices. 

Conditioning on the partial order $h$, the likelihood is the marginal probability to observe $y$, if $l$ is a noise-free draw from the linear extensions of $h$, so we define
\begin{align}\label{eq:mallows-likelihood-for-PO}
    p^{(M)}(y|h,\theta) &\equiv \sum_{l\in\L[h]} p^{(M)}(y|l,\theta)p^{(P)}(l|h)\\
    \intertext{using \eqref{noise-free} and the symmetry of $d(l,y)$ in its arguments}
    &= \frac{1}{|\L[h]|}\sum_{l\in\L[h]} p^{(M)}(l|y,\theta).\nonumber
    \end{align}
    
    We now give a recursion for evaluating the sum. Since $\L[h]=\cup_{k\in \max(h)}\L_k[h]$,
    \begin{align}
    \sum_{l\in\L[h]} p^{(M)}(l|y,\theta)&=\sum_{k\in \max(h)}\sum_{l\in\L_k[h_{-k}]} p^{(M)}(l|y,\theta),\nonumber\\
    \intertext{and now $l$ starts with $k$, so use Equation~\ref{eq:mallows_lkd_for_order_seq} to split off the first factor}
    %&=\sum_{k\in \max(h)}\sum_{l\in\L_k[h_{-k}]} q^{(M)}(k|y,\theta) p^{(M)}(l_{-k}|y_{-k},\theta)\nonumber\\
    &=\sum_{k\in \max(h)} q^{(M)}(k|y_{-k},\theta) \sum_{l'\in\L[h_{-k}]}  p^{(M)}(l'|y_{-k},\theta),
    \label{eq:mallows-likelihood-recursion}
\end{align} 
where $h_{-k}=([n]\setminus\{k\},\succ_h)$ and similarly for $y_{-k}$. Let \[f(y,h,\theta)=\sum_{l\in\L[h} p^{(M)}(l|y,\theta).\] The recursion in Equation~\ref{eq:mallows-likelihood-recursion} is
\[
f(y,h,\theta)=\sum_{k\in \max(h)} q^{(M)}(k|y_{-k},\theta) f(y_{-k},h_{-k},\theta).
\]
Algorithm \ref{alg:mallows} adapts the algorithm of \cite{knuth1974structured} for recursive counting of linear extensions to evaluating this weighted sum. We evaluate the count of linear extensions $|\L[h]|$ in the same pass. Once we are done the likelihood is 
$
p^{(M)}(y|l,\theta)={f(y,h,\theta)}/{|\L[h]|}
$.\\

\begin{algorithm}[H]
    \caption{Evaluating $f(y,h,\theta)=\sum_{l\in\L[h} p(l|y,\theta)$ under the Mallows observation model}\label{alg:mallows}
    \SetKwInOut{Input}{input}
    \SetKwFunction{f}{f}
    %\SetKwFunction{length}{length}
    \SetKwProg{Fn}{Function}{:}{}
        
    \Input{$h,\theta,y$}
    \Fn{\f{$y$,$h$,$\theta$}}{
        $n\leftarrow$ number of elements in $h$\\
        \If{$D(h)=n$ ($h$ is a total order)}{
          $l\leftarrow$ list ordered as elements of $h$\\
          $g\leftarrow p^{(M)}(l|y,\theta)$ (see Equation~\ref{eq:mallows-select-first})\\
          {{\bfseries return} $(f,\text{count})=(g,1)$}
        }
        \If{$D(h)=1$ ($h$ is an empty order)}{
          {\bfseries return} $(f,\text{count})=(1,n!)$
        }
        Set count $\leftarrow 0$, $f\leftarrow 0$\\
        \ForEach{$k\in\max(h)$}{
          $g_k \leftarrow q^{(M)}(k|y,\theta)$\\
          $(f_k, c_k)\leftarrow \f(y_{-k},h_{-k},\theta)$\\
          $f \leftarrow f+g_k\times f_k$ (Equation~\ref{eq:mallows-likelihood-recursion})\\
          $\text{count} \leftarrow \text{count}+c_k$
        }
        {\bfseries return} $(f,\text{count})$
    }
    ({\it the partial order likelihood $p(y|h,\theta)$ in Equation~\ref{eq:mallows-likelihood-for-PO} is then $f/\text{count}$.})
\end{algorithm}
\vspace*{1ex}
Algorithm \ref{alg:mallows} gives a recursion in which a function $f(y,h,\theta)$ returns the sum $\sum_{l\in\L[h} p^{(M)}(l|y,\theta)$ over orders of $m$ elements, by evaluating the sum over $k$ in the last line of Equation~\ref{eq:mallows-likelihood-recursion} and calling itself to evaluate $f(y_{-k},h_{-k},\theta)=\sum_{l'\in\L[h_{-k}]}  p^{(M)}(l'|y_{-k},\theta)$ on orders of length $n-1$. The recursion stops if $f$ is called with $h$ a total order (one extension, so return \eqref{eq:mallows_lkd_for_order}) or if $h$ is the empty order (then $\L[h]=\P_{[n]}$ so the sum is one as $p^{(M)}(l|y,\theta)$ is normalised over $l\in\P_{[n]}$) so the recursion stops at $n\ge 2$ at lowest. 
% \[
% |\L[h]|=\sum_{j\in \max(h)} |\L[h_{-j}]| 
% \]
% so $f(h, y,\theta)$ also returns $|\L[h]|$, calling itself to count $|\L[h_{-j}]|$.

Note that this paper uses the Kendall-tau distance as the distance metric for the Mallows model. It is chosen not only for its tractable normalising constant, in addition, \cite{a21c2824-5af7-3a80-b044-fe5b689dcd65} suggests Kendall-tau is the metric of choice considering its interpretability, tractability, invariance, sensitivity and available theory. However, from a modelling perspective, other distance metric choices may be more natural for a given data set. For example, \cite{vitelli2018probabilistic} prefers the footrule distance. Although it is physically well motivated, we didn't explore the footrule distance due to the challenge of its intractable normalising constant in our context. 

In our implementation, we impose a Gamma prior over the dispersion parameter $\theta$. We discuss in more detail in Section \ref{sec:application}.
\section{Bayesian Inference}\label{sec:bayesian-inference}

Given a prior model for partial orders and an observation model for ranking lists, Bayesian inference is straightforward in principle. 
We demonstrate Bayesian inference for the $(PDP,\rho)$-partial order model with queue-jumping error model ($(PDP,\rho)\backslash Q$) and the Mallows error model ($(PDP,\rho)\backslash M$). 
Here we write down the posterior distributions for these two models respectively.

The posterior distribution of the $(PDP,\rho)\backslash Q$ model is 
\begin{equation}\label{B}
    \pi(h^*|\mathbf{y}) \propto \sum_{S\in\Xi_n}\sum_{k=1}^\infty \int_{Z^*[h^*;S]}\int_{p=0}^1 \int_{\rho=0}^1 \pi(Z^*,\rho,k,S,p|\mathbf{y}) d\rho dp dZ^*, 
\end{equation}
where
\begin{equation*}
    \pi(Z^*,\rho,k,S,p|\mathbf{y}) \propto p^{(Q)}(\mathbf{y}|h(Z),p)\pi(Z^*|\rho,k,S)\pi_\rho(\rho)\pi_K(k)\pi_P(p)P_{\eta,[n]}(S).
\end{equation*}

The posterior distribution of the $(PDP,\rho)\backslash M$ model is
\begin{equation}\label{C}
    \pi(h^*|\mathbf{y}) \propto \sum_{S\in\Xi_n}\sum_{k=1}^\infty \int_{Z^*[h^*;S]}\int_{\theta=0}^\infty \int_{\rho=0}^1 \pi(Z^*,\rho,k,S,\theta|\mathbf{y}) d\rho d\theta dZ, 
\end{equation}
where
\begin{equation*}
    \pi(Z,\rho,k,S,\theta|\mathbf{y}) \propto p^{(M)}(\mathbf{y}|\theta,h(Z))\pi_\theta(\theta)\pi(Z^*|\rho,k,S)\pi_\rho(\rho)\pi_K(k)P_{\eta,[n]}(S).
\end{equation*}

We employ Gibbs Metropolis-Hastings for the inference. Varying $K$ changes the dimension of $Z^*$ and so we use reversible jumping to target the conditional for $K$, and the classic Gibbs sampling algorithm of \cite{neal2000markov} for the partition $S$ with a Poisson Dirichlet process as prior and simple random-walk Metropolis-Hasting for most of the remaining the parameters. We assume priors as $k\sim Geo(\eta_k), \rho\sim Beta(1,\eta_\rho)$ and $\theta\sim Gamma(\eta_\theta,1)$. The detailed algorithm is included in \ref{MCMC}. 

\subsection{Asymptotic posterior distributions in the noise free case}\label{sec:asymptotic-post}
\def\htrue{{h^\dagger}}

It is of interest to consider the behavior of these posteriors as $N\to \infty$, the large data limit. Suppose the true partial order is $\htrue$. Since the space of partial orders is discrete, the natural question is whether $\pi(\htrue|y)\to 1$ as $N\to\infty$ so the posterior model is a consistent estimator for the unknown true partial order when the observation model is correct (so ignoring model misspecification).

In the first proposition $O=(o_k)_{k=1}^{2^n-1}$ is the set of all non-empty subsets of $[n]$. We are interested in the setting where we can only make observations on suborders $\htrue[o_i]$ of $\htrue$ for $i\in I$ where $I\subseteq [2^n-1]$
is the set of subset-indices for subsets of actors appearing together in observable groups.
We observe $N$ linear extensions $\mathbf{y}_{1:N}$ in total. We set up the asymptotics so that for $i\in I$ there are $N_i$ lists $\mathbf{y}_i=(y_{i,1},\dots,y_{i,N_i})$ associated with each suborder $\htrue[o_i]$. Here each $y_{i,j}=(y_{i,j,1},\dots,y_{i,j,k_i})\in\P_{o_i}$ is an entire list for each observation $j = 1,\dots,N_i$ of the group $o_i,\ i\in I$. We will set $N_i=0$ for $i\in [2^n-1]\setminus I$ and take the limit as $\min_{i\in I} N_i\to \infty$. The total number of observed lists is $N = \sum_{i=1}^{2^n-1}N_i=\sum_{i\in I}N_i$. 

\begin{proposition}(Consistency for the $(K,n,\rho)\backslash P$ model)\label{prop:PONE-Consistency}
    Let $h^\dagger\in \H_{[n]}$ be given and suppose $y_{i,j}\sim p^{(P)}(\cdot|h^\dagger[o_i])$ jointly independent for all $i\in I$ and $j=1,\dots N_i$.
    Let $\pi(h|\mathbf{y}_{1:N})$ be the posterior of the $(K,n,\rho)\backslash P$ model with $K\ge \lfloor n/2\rfloor$, so that
    $$\pi(h|\mathbf{y}_{1:N}) = \int_{Z:h(Z)=h}\int_{\rho=0}^1 \pi(Z,\rho|\mathbf{y}) d\rho dZ$$
    where
    $$\pi(Z,\rho|\mathbf{y}_{1:N}) \propto p^{(P)}(\mathbf{y}_{1:N}|h(Z))\pi(Z|\rho)\pi_\rho(\rho).$$
    If for each pair of actors $(i,j)\in[n]\times [n]$ with $i\neq j$, there exists $k\in I$ such that $\{i,j\}\subseteq o_k$ then $\pi(\htrue|\mathbf{y}_{1:N})\to 1$ as $\min_{i\in I} N_i\to \infty$ for every $h^\dagger\in\H_{[n]}$.
\end{proposition}

Proof: see \ref{C-EF}.

In order to recover $\htrue$ with certainty we must be able to observe at least one list $o_k$ informing the relation between each pair $(i,j)$ of actors in $\htrue$ infinitely often. We have an infinite number of observations of a random variable informing every possible edge of the DAG representation of $h$.

% Below proposition considers the Mallows observation model. 
% \begin{proposition}(Consistency for the $(K,n,\rho)\backslash M$ model)\label{POPL-Consistency}
%     Let $\pi(h|\mathbf{y}_{1:N})$ be the $(K,n,\rho)$-posterior in the Mallows error model, such that
%     $$\pi(h|\mathbf{y}_{1:N})=\int_{Z\in\mathbbm{R}^{n\times K}}\int_{\theta=0}^{\infty} \int_{\rho=0}^1 \pi(Z,\rho,\theta|\mathbf{y}_{1:N})d\rho d\theta dZ$$
%     where 
%     $$\pi(Z,\rho,\theta|\mathbf{y})\propto p^{(M)}(\mathbf{y}_{1:N}|\theta,h(Z))\pi_\theta(\theta)\pi(Z|\rho)\pi_\rho(\rho).$$
%     Suppose we observe full length lists $o_k=[n]$ for $k=1,...,N$. Then there exists priors on $\theta$ such that $\pi(\htrue|y)\to 1$ as $N\to\infty$ when $\theta>0$.
% \end{proposition}

% Proof: see \ref{C-PL}.

% We believe the consistency of the partial order with queue-jumping observation model can be shown in a similar manner. 

These results are not given for the posteriors arising in a prior model for a partial order with ties. In fact ties are not identifiable as they are treated as unordered in the observation model. 
%If we take a pair of actors $i,j\in [n]$ in the true partial order that have the same order relations to all other actors, so $i\succ_\htrue k$ iff $j\succ_\htrue k$ etc, and are themselves unordered, so $i\|_\htrue j$, then we may add a tie between $i$ and $j$ without changing the distribution of the data $p(y|\htrue)$. 
When we fit the $(PDP,\rho)$-partial order model in the noise free observation model, the posterior will concentrate on the \textit{tie class} of the true partial order - the set of partial orders which are equivalent to the true partial order up to ties. The posterior distribution within this tie class of partial orders will just be the prior distribution. There are a number of reasons why the tied model remains useful for us. Firstly, adding ties changes the prior distribution in a fundamental and desirable way,
shifting probability mass onto orders with common relations between groups (such as ``bucket" orders and Vertex Series Parallel partial orders).
Social hierarchies of this sort are common in human society.
We do not restrict the class of partial orders we consider to bucket orders, but we can use ties to give them more weight in the prior.
Secondly, taken with the randomly variable column dimension $K$, the latent variables $Z^*$ in a model for partial orders with ties can potentially fit the data with lower parameter dimension ($C\times K$ with $C \le n$ and potentially $K<n/2$) than the latent variables $Z$ ($n\times n/2$ fixed) in the $(K,n,\rho)$-model for partial orders.  

\section{Application and Model Comparison}\label{sec:application}

We demonstrate our models by applying them to a social hierarchy study. We fit both the $(PDP,\rho)\backslash Q$ and the $(PDP,\rho)\backslash M$ on two different datasets from the `Royal Acta' data (see section \ref{sec:data}). We analyse the inference results, and conduct reconstruction-accuracy test with synthetic data (section \ref{sec:reconstruction}). We then perform model comparison with 1) models with no ties $(K,n,\rho)$; and 2) other ranking methods not based on partial orders - the Plackett-Luce mixture and the Mallows mixture models. 

\subsection{The `Royal Acta' Data}\label{sec:data}

The `Royal Acta' is a database created for `The Charters of William II and Henry I' project by the late Professor Richard Sharpe and Dr Nicholas Karn (\cite{Sharpe14}). It collects royal acts (mainly charters but also writs and other letters) issued in the names of two English kings, William II (reigned 1087 to 1100), and his brother Henry I (reigned 1100 to 1135). Each royal act is identified to a certain time period. Some may be specific to a single year, while some with more uncertainty are assigned to a range of time. Each royal act comes with a witness list - an ordered list of names of individual witnesses. For some examples, see \ref{app:lists}. The order reflects the relative social importance of individual witnesses. For example, an order Archbishops $\succ$ Bishops $\succ$ Earls can be clearly observed. The historians are interested to study the social hierarchy among the bishops specifically. We infer these relations by partial orders, which are natural choices for social hierarchy representation given their transitivity and generality compared to total orders. Power hierarchies were rigidly observed at this time. However it is natural to expect some incomparabilities or ties between bishops.  We extract the sub-lists that only contain bishops (as these are listed in a block together). For more background and details on data processing, please refer to \cite{nicholls122011partial} and \cite{nicholls2023bayesian}. 

The data is temporal. We take `snapshots' for the time periods 1131-1133 and 1100-1103 to study the social hierarchy among bishops in these two time periods. Table \ref{table:data-summery} summarises some key statistics of the datasets. The full witness lists for these two periods are shown in \ref{app:lists}. We assign each bishop in a single time period with a unique label. Assume the ground set of bishops in a certain time period is $[n]$. We assume each witness list to be an observed ranking list $y_i, i=1,\dots,N$ drawn from the `true' social hierarchy partial order $h$ with possible recording errors. The lists $Y=\{y_1,\dots,y_N\}$ are incomplete, in the sense that the membership in list $i$ is $o_i\subset [n]$ in a certain time period. The witness lists are of varied lengths. 

\begin{table}[h!]
\centering
\begin{tabular}{*3c}
\toprule
Time Period &  1131-1133 & 1100-1103\\
\midrule
 Number of Lists & 21 & 13\\
 Number of Actors & 15 & 9\\
 Length of the Longest List & 8 & 8 \\
 Length of the Shortest List & 3 & 2 \\
 \bottomrule
\end{tabular}
\caption{Information on data structure.}
\label{table:data-summery}
\end{table}
 Some contradictions can be observed in the witness lists. For example, in 1131-1133, $7, 11, 12$ and $13$ appear in both lists 2 and 16 (along with others). However in list $2$ the order is $13, 12, 7, 11$ while list 16 has $12$ and $13$ swapped but otherwise the same. Is this just noise, are $12$ and $13$ unordered or are they tied? Firstly, from a prior modelling perspective but not pretending to any historical expertise, bishop 13 (Henry, de Blois, Bishop of Winchester, 1129-1171) and bishop 12 (Gilbert, the Universal, bishop of London) may be incomparable. They come from two rather independent power systems as Winchester at the time was an important and independent administrative centre and seat of power, on a par with London, so may be incomparable (like administrators and academics). It defined a separate hierarchy of tied cities. Secondly, they might be tied, i.e. they are each identified by all other bishops as top powers and possess the same power relation with all other bishops and are ranked with the same importance. Thirdly, the switch in order in the lists might simply be a recording error. Perhaps Henry (13) simply arrived late for witnessing document 16 while his true status was higher than that of Gilbert (12). It should be said that historians do not expect many errors in these data as the deeds involved concerned properties of significant value, and the King or Queen were often themselves signatories. 

\subsection{Reconstruction-Accuracy Tests}\label{sec:reconstruction}

Our list data are incomplete and the number of lists $N$ are not much larger than the number of actors $n$. In order to measure the reliability of the partial order (social hierarchy) reconstructions and to study the $(PDP,\rho)\backslash M$ and $(PDP,\rho)\backslash Q$ models performance on different datasets, we perform reconstruction accurcy tests on both models. In particular, we take the last sampled state of partial order $h^{(T)}$ in section \ref{sec:real-app} and generate synthetic data with the same lengths and list-membership as the real data $$y_i\sim p(\cdot|h^{(T)}[o_i],p,\theta),\ i=1,...,N$$ for both periods 1131-1133 and 1100-1103. If our models are correct then the data contain enough information to reconstruct the truth accurately. We consider the following four synthetic datasets for each time period: 
\begin{itemize}
    \item Simulation 1: error-free, $y_i\sim\mathcal{U}(\L[h^{(T)}[y^{obs}_i]]), \forall i\in [N]$.
    \item Simulation 2: data-list with random error. For $i\in [N]$, we simulate $y_i\sim \mathcal{U}(\L[h^{(T)}[y^{obs}_i]])$ from the noise-free model; we select a pair of actors $a,b\in y_i$ uniform at random, $a\neq b$, and put them in the same order as they appear in the data. If $a\prec b$ in $y_i$ but $a\succ b$ in $y^{obs}_i$ then exchange the positions of $a$ and $b$ in $y_i$ leaving all else unchanged.\footnote{Simulation 2 introduces around 5\% error to the synthetic data. }
    \item Simulation 3: data-list with Mallows error given $\theta^*$, $$y_i\sim p^{(M)}(\cdot|h^{(T)}[o_i],\theta),\forall i\in[N].$$ 
    \item Simulation 4: data-list with queue-jumping error given $p^*$, $$y_i\sim p^{(Q)}(\cdot|h^{(T)}[o_i],p),\forall i\in[N].$$
\end{itemize}
Simulation 2 blends a little of the real data into the synthetic data. It uses a noise process that is neither Mallows nor queue-jumping, so the fitted models are all (mildly) misspecified. Note that the $\theta^*$ and $p^*$ values in Simulation 3 and 4 are chosen so that they produce roughly the same level of noise in the data-lists. The experiments presented in this section choose $\theta^*=2.7$ and $p^*=0.05$.\footnote{With these values $(\theta^*=2.7,p^*=0.05)$ the proportion of lists with errors is $\sim 43\%$ error for the 1131-1133 data structure ($N=21$) and $\sim 24\%$ for the 1100-1103 data structure ($N=13$; lists tend to be shorter so lower error probability, see \ref{app:lists}).} We fit both the $(PDP,\rho)\backslash M$ and $(PDP,\rho)\backslash Q$ models on above simulations. Each MCMC is run for $1e5$ iterations. We omit our convergence analysis on this synthetic data, but ESS values and trace plots showed convergence as good as or better than the convergence and mixing we report for the real data (given below).

We summarise the sampled partial orders using consensus partial orders: $h^{con}(\epsilon)$ includes order relation/edge $(i\succ_{h^{con}(\epsilon)} j)$ if the relation appears more than $\epsilon T$ times in the $T$ MCMC samples after thinning. The true and consensus partial orders for all experiments in this section are given in \ref{app:reconstruction}. We calculate the proportion of the true-positive and false-positive relations for $h^{con}(\epsilon), \epsilon\in [0,1]$ and construct the receiver operating characteristic (ROC) curves in \Fig~\ref{fig:roc1} (for simulation 1 and 2) and \Fig~\ref{fig:roc2} (for simulation 3 and 4) respectively. The ROC curves show the proportion of false-positive and true-positive relations given concensus order threshold $\epsilon$. The true positive rate (TPR) and false positive rate (FPR) increase with decreasing $\epsilon$ from $(0,0)$ at $\epsilon=1$ (the consensus order is empty) to $(1,1)$ at $\epsilon=0$ (complete graph). 
\begin{figure}[h!]
  \centering
  \includegraphics[width=0.85\linewidth]{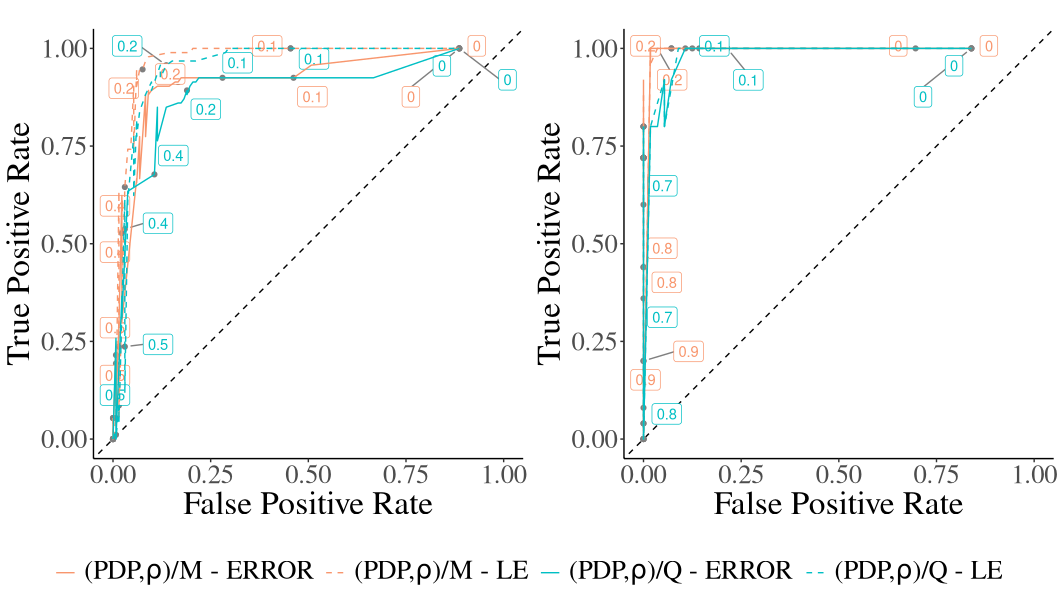}
  \caption{The receiver operating characteristic (ROC) curves for synthetic data - simulation 1 (dashed line) \& 2 (solid line) - using synthetic data with 1131-1133 (left) and 1100-1103 (right) list membership structures, analysed with the $(PDP,\rho)\backslash M$ (orange) and $(PDP,\rho)\backslash Q$ (blue) models. The true/false positive rates are ploted against $\epsilon$ the threshold to construct concensus order $h^{con}(\epsilon), \epsilon\in[0,1]$. }
  \label{fig:roc1}
\end{figure}

\begin{figure}[h!]
  \centering
  \includegraphics[width=0.85\linewidth]{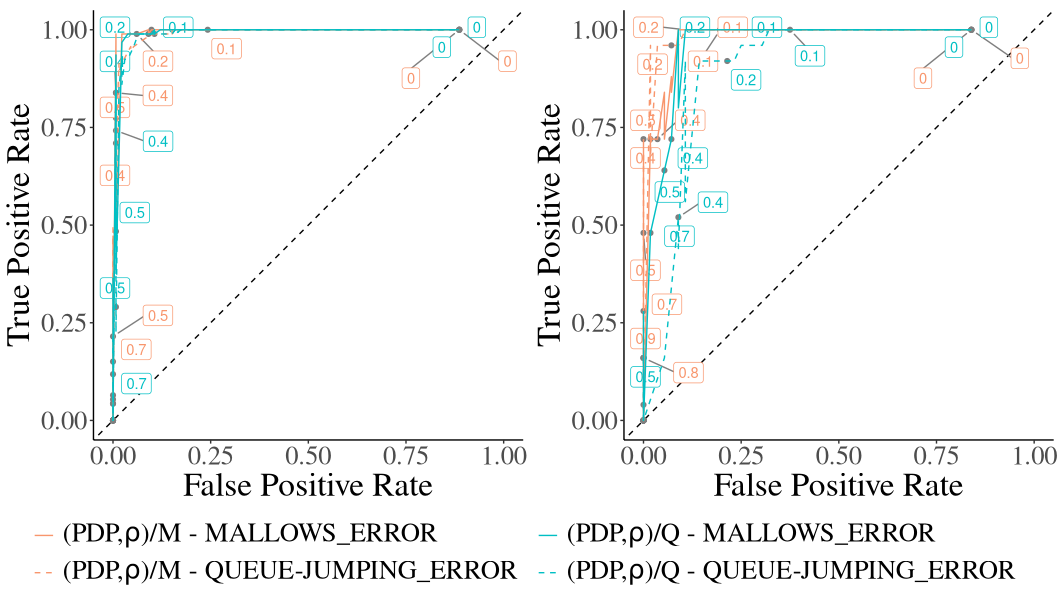}
  \caption{The receiver operating characteristic (ROC) curves for synthetic data - simulation 3 (solid line) \& 4 (dashed line) - using synthetic data with 1131-1133 (left) and 1100-1103 (right) list membership structures, analysed with the $(PDP,\rho)\backslash M$ (orange) and $(PDP,\rho)\backslash Q$ (blue) models. The true/false positive rates are ploted against $\epsilon$ the threshold to construct concensus order $h^{con}(\epsilon), \epsilon\in[0,1]$. }
  \label{fig:roc2}
\end{figure}
For each simulated data set there is $\epsilon$ giving high true-positive and low false-positive reconstructed relation fractions: if our model is accurate then we reconstruct relations well. Both \Fig~\ref{fig:roc1} and \Fig~\ref{fig:roc2} have an elbow closest $(0,1)$ around $\epsilon=0.2$. This value gives a good balance of true- and false-positives for synthetic reconstruction analyses on the $(PDP,\rho)\backslash M$ and $(PDP,\rho)\backslash Q$ models for both time periods. The consensus orders with $\epsilon=0.2$ and the `true' partial orders for simulation are shown in \ref{app:reconstruction}. We use the same threshold when we present consensus orders on real data in section \ref{sec:real-app}. From \Fig~\ref{fig:roc1} and \Fig~\ref{fig:roc2}, both models reconstruct the true order relation well. We observe that the $(PDP,\rho)\backslash M$ model slightly outperforms the $(PDP,\rho)\backslash Q$ model at partial order reconstruction in most scenarios. Under the 1131-1133 structured data, the two models appear to be robust to the Mallows or queue-jumping error type and both reconstruct the truth well.

\subsection{Application on the `Royal Acta' Data}\label{sec:real-app}

We perform Bayesian inference for the $(PDP,\rho)$-partial order model with the queue-jumping\footnote{We choose queue-jumping ``up'' (rather than ``down'') for this specific dataset. For example, we know of cases where a lower status bishop is the nephew of a higher status bishop and is promoted to appear immediately after his uncle in the list. Experiments with the two queue-jumping directions in \cite{nicholls2023bayesian} gave very similar results.} or Mallows observation model. We employ Metropolis-Hasting MCMC for such inference. Each chain is run for $1e5$ iterations where we record every $2n$ steps. We choose a burn-in period of $500\times 2n$ for all chains. See \ref{app:results} for Effective Sample Sizes and \ref{app:trace} for MCMC output traces. These appear to show good MCMC convergence and mixing.

The posterior distributions on partial order depth are shown in \Fig~\ref{fig:depths}. Both the $(PDP,\rho)\backslash M$ and the $(PDP,\rho)\backslash Q$ models conclude similar depths in general. It is clear that the hierarchy is not a total order (depths close to $n$ have low posterior probability in both analyses). The mean posterior depths are around 8.8 for both models in 1131-1133 (with $n=15$ bishops). In 1100-1103 (where $n=9$), we obtain posterior mean depths of 3.9 for $(PDP,\rho)\backslash M$ and 4.8 for $(PDP,\rho)\backslash Q$. The partial orders in 1100-1103 are flatter in general as reflected on the consensus orders in \Fig~\ref{fig:concensus-mallows} and \Fig~\ref{fig:concensus-qj} for the Mallows and the queue-jumping observation models respectively. The consensus orders from different observation models (Mallows or queue-jumping) are similar despite some slight relation difference, showing the result's robustness to different observation models. Bishop 3 (Robert, de Limesey, bishop of Chester), bishop 5 (Robert, Bloet, bishop of Lincoln) and bishop 8 (William, Giffard, bishop of Winchester, 1100-1129) are strongly tied during 1100-1103. During 1131-1133, both models suggest relatively high posterior probability of ties between Bishops 1 (Roger, Bishop of Salisbury), bishop 12 (Gilbert, the Universal, bishop of London) and bishop 13 (Henry, de Blois, Bishop of Winchester, 1129-1171). These are indicated in the `heatmaps' in \Fig~\ref{fig:heatmap-m} and \Fig~\ref{fig:heatmap-qj} which display the probability of different pairs of actors fall in the same cluster.  \\

\begin{minipage}{\linewidth}
    \centering
    \begin{tikzpicture}[thick,scale=0.8, every node/.style={scale=0.6}]
        \node[draw, circle, minimum width=.75cm] (10) at (0, .2) {$10$};
        \node[draw, circle, minimum width=.75cm] (5) at (0, .9) {$5$};
        \node[draw, circle, minimum width=.75cm] (14) at (0,1.6) {$14$};
        \node[draw, circle, minimum width=.75cm] (7) at (0, 2.3) {$7$};      
        \node[draw=blue!60, rectangle, minimum width=7cm,minimum height=1cm] (16) at (1, 3) {};
        \node[draw, circle, minimum width=.75cm] (3) at (-3, 3) {$3$};
        \node[draw, circle, minimum width=.75cm] (1) at (-1, 3) {$1$};
        \node[draw, circle, minimum width=.75cm] (12) at (1, 3) {$12$};
        \node[draw, circle, minimum width=.75cm] (13) at (3, 3) {$13$};
        \node[draw, circle, minimum width=.75cm] (4) at (-3, -.15) {$4$};
        \node[draw, circle, minimum width=.75cm] (8) at (0, -1.2) {$8$};
        \node[draw, circle, minimum width=.75cm] (15) at (0, -0.5) {$15$};
        \node[draw, circle, minimum width=.75cm] (2) at (-1.5, -3.3) {$2$};
        \node[draw, circle, minimum width=.75cm] (9) at (0, -1.9) {$9$};
        \node[draw, circle, minimum width=.75cm] (11) at (0, -2.6) {$11$};
        \node[draw, circle, minimum width=.75cm] (6) at (-1.5, -4) {$6$};
        \draw[-latex] (1) -- (7);
        \draw[-latex] (3) -- (7);
        \draw[-latex] (12) -- (7);
        \draw[-latex] (13) -- (7);
        \draw[-latex] (7) -- (14);
        \draw[-latex] (5) -- (10);
        \draw[-latex] (14) -- (5);
        \draw[-latex] (10) -- (15);
        \draw[-latex] (15) -- (8);
        \draw[-latex] (8) -- (9);
        \draw[-latex] (9) -- (11);
        \draw[-latex] (11) -- (2);
        \draw[-latex] (4) -- (2);
        \draw[-latex] (3) -- (4);
        \draw[-latex] (2) -- (6);

        \node[draw, circle, minimum width=.75cm] (102) at (8, 2.5) {$2$};
        \node[draw, circle, minimum width=.75cm] (101) at (5, .5) {$1$};
        \node[draw, circle, minimum width=.75cm] (103) at (6.5, .5) {$3$};
        \node[draw, circle, minimum width=.75cm] (105) at (8, .5) {$4$};
        \node[draw, circle, minimum width=.75cm] (108) at (9.5, .5) {$5$};
        \node[draw, circle, minimum width=.75cm] (104) at (11, .5) {$8$};
        \node[draw=blue!60, rectangle, minimum width=9cm,minimum height=1cm] (109) at (8, .5) {};
        \node[draw, circle, minimum width=.75cm] (106) at (8, -1.5) {$6$};
        \node[draw, circle, minimum width=.75cm] (109) at (9.5, -3.5) {$9$};
        \node[draw, circle, minimum width=.75cm] (107) at (6.5, -3.5) {$7$};
        \draw[-latex,red] (102) -- (101);
        \draw[-latex,red] (102) -- (103);
        \draw[-latex] (102) -- (105);
        \draw[-latex,red] (102) -- (108);
        \draw[-latex] (102) -- (104);
        \draw[-latex] (101) -- (106);
        \draw[-latex] (103) -- (106);
        \draw[-latex] (105) -- (106);
        \draw[-latex] (108) -- (106);
        \draw[-latex] (104) -- (106);
        \draw[-latex] (106) -- (109);
        \draw[-latex] (106) -- (107);
    \end{tikzpicture}
        
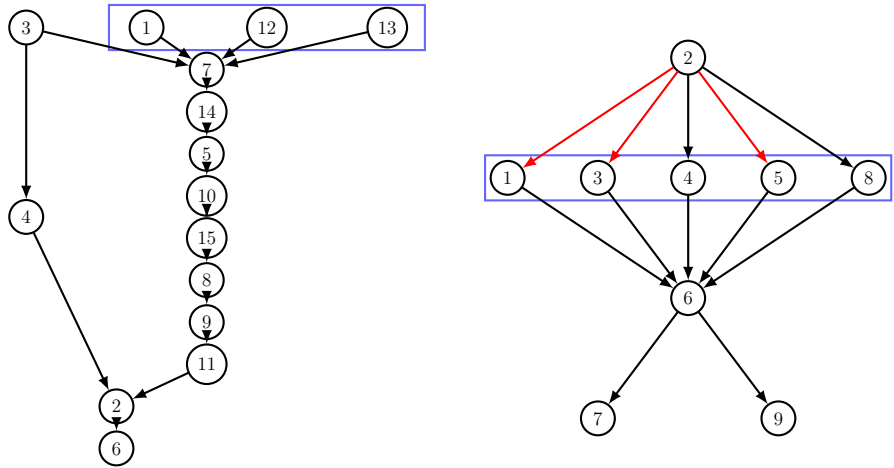
\captionof{figure}{The consensus orders for the $(PDP,\rho)\backslash M$ model on 1131-1133 (left) and 1100-1103 (right) Royal Acta (bishop) data. We conclude an edge if such order relation has more than 0.2 posterior probability (inferred from section \ref{sec:reconstruction}). An edge is colored red if it has more than 0.9 posterior probability. The \textit{blue boxes} indicate the tie relations with more than 0.5 posterior probabilities. }\label{fig:concensus-mallows}
\end{minipage}

\begin{minipage}{\linewidth}
    \centering
    \begin{tikzpicture}[thick,scale=.85, every node/.style={scale=0.65}]
        \node[draw, circle, minimum width=.75cm] (9) at (0,-1.5) {$9$};
        \node[draw, circle, minimum width=.75cm] (10) at (0,-.6) {$10$};
        \node[draw, circle, minimum width=.75cm] (5) at (0,.1) {$5$};
        \node[draw, circle, minimum width=.75cm] (15) at (1, -1.05) {$15$};    
        \node[draw, circle, minimum width=.75cm] (7) at (0, 1.5) {$7$};
        \node[draw, circle, minimum width=.75cm] (1) at (-1, 2.9) {$1$};
        \node[draw, circle, minimum width=.75cm] (14) at (0, .8) {$14$};
        \node[draw, circle, minimum width=.75cm] (13) at (-2, 2.2) {$12$};
        \node[draw, circle, minimum width=.75cm] (12) at (0, 2.2) {$13$};
        \node[draw, circle, minimum width=.75cm] (8) at (-1, -1.05) {$8$};
        \node[draw, circle, minimum width=.75cm] (2) at (2, 2.2) {$3$};
        \node[draw, circle, minimum width=.75cm] (6) at (1.7, -3.35) {$6$};
        \node[draw, circle, minimum width=.75cm] (3) at (1.7, -2.65) {$2$};
        \node[draw, circle, minimum width=.75cm] (4) at (3.4, 0.2) {$4$};
        \node[draw, circle, minimum width=.75cm] (11) at (0, -2.2) {$11$};
        \node[draw=blue!60, rectangle, minimum width=3.8cm,minimum height=.8cm] () at (0, -1.05) {};
        \draw[-latex] (1) -- (12);
        \draw[-latex] (1) -- (13);
        \draw[-latex] (2) -- (7);
        \draw[-latex] (13) -- (7);
        \draw[-latex] (7) -- (14);
        \draw[-latex] (12) -- (7);
        \draw[-latex] (5) -- (10);
        \draw[-latex] (11) -- (3);
        \draw[-latex] (10) -- (15);
        \draw[-latex] (10) -- (8);
        \draw[-latex] (14) -- (5);
        \draw[-latex] (8) -- (9);
        \draw[-latex] (9) -- (11);
        \draw[-latex] (3) -- (6);
        \draw[-latex] (4) -- (3);
        \draw[-latex] (15) -- (9);

        \node[draw, circle, minimum width=.75cm] (102) at (9, 2.3) {$2$};
        \node[draw, circle, minimum width=.75cm] (101) at (8, 1.3) {$1$};
        \node[draw, circle, minimum width=.75cm] (103) at (5.6, -.2) {$3$};
        \node[draw, circle, minimum width=.75cm] (105) at (7.2, -.2) {$5$};
        \node[draw, circle, minimum width=.75cm] (108) at (8.8, -.2) {$8$};
        \node[draw, circle, minimum width=.75cm] (104) at (10.4, -.2) {$4$};
        \node[draw=blue!60, rectangle, minimum width=5.5cm,minimum height=1cm] (109) at (7.2, -.2) {};
        \node[draw, circle, minimum width=.75cm] (106) at (8, -1.7) {$6$};
        \node[draw, circle, minimum width=.75cm] (109) at (8.8, -3) {$9$};
        \node[draw, circle, minimum width=.75cm] (107) at (7.05, -3) {$7$};
        \draw[-latex] (102) -- (101);
        \draw[-latex] (101) -- (103);
        \draw[-latex] (101) -- (105);
        \draw[-latex] (101) -- (108);
        \draw[-latex] (102) -- (104);
        \draw[-latex] (103) -- (106);
        \draw[-latex] (105) -- (106);
        \draw[-latex] (108) -- (106);
        \draw[-latex] (104) -- (106);
        \draw[-latex] (106) -- (109);
        \draw[-latex] (106) -- (107);
    \end{tikzpicture}
        \captionof{figure}{The concensus orders for the $(PDP,\rho)\backslash Q$ model on the 1131-1133 (left) and 1100-1103 (right) 'royal acta' (bishop) witness list data. We conclude an edge if such order relation has more than 0.2 posterior probability (inferred from section \ref{sec:reconstruction}). An edge is colored red if it has more than 0.9 posterior probability. The \textit{blue boxes} indicate the tie relations with more than 0.5 posterior probabilities. }\label{fig:concensus-qj}
\end{minipage}

\begin{figure}[h!]
  \centering
  \includegraphics[width=\linewidth]{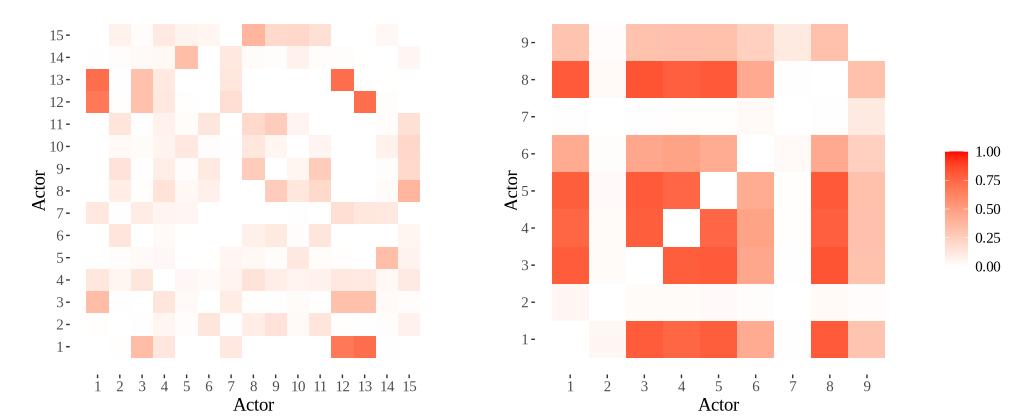}
  \caption{The ‘heatmaps’ for the estimated clusters for the $(PDP,\rho)\backslash M$ model for time periods 1131-1133 (left) and 1100-1103 (right). }
  \label{fig:heatmap-m}
\end{figure} 

\begin{figure}[h!]
  \centering
  \includegraphics[width=\linewidth]{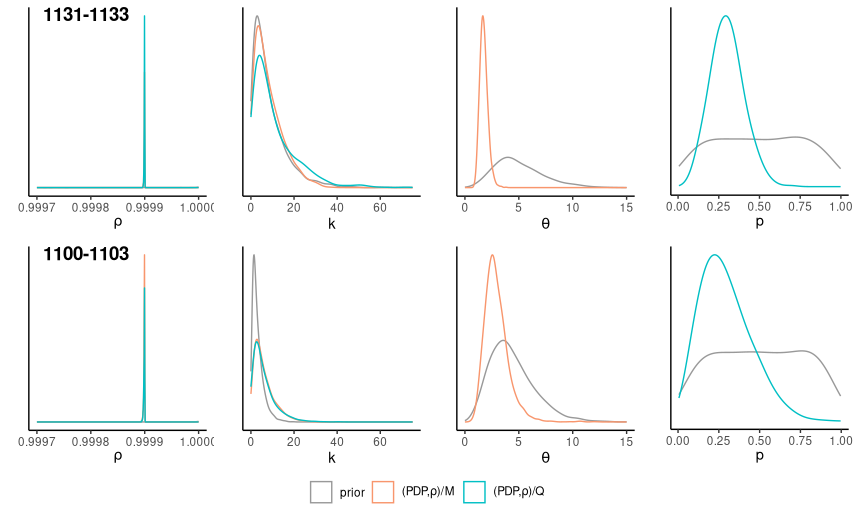}
  \caption{The prior/posterior distributions for key parameters for the $(PDP,\rho)\backslash M$ (pink) and the $(PDP,\rho)\backslash Q$ (blue) models on periods 1131-1133 (upper) and 1100-1103 (lower). The prior distributions are colored in gray. }
  \label{fig:posteriors}
\end{figure} 

We present the prior and posterior distributions on some key parameters, as is shown in \Fig~\ref{fig:posteriors}. \Fig~\ref{fig:depths} shows the prior and posterior distributions of the number of clusters and partial order depths in each model. Both the posterior distributions for $K$ and $\rho$ follow similar behaviour across the two models in different time periods. Prior sensitivity analysis in \ref{app:sensitivity} shows that the posterior of $K$ is less sensitive to the prior choice during 1131-1133 (where there is more data). A prior mean around $n/2$ appears to be the optimal choice among $\Bar{k}\in\{n,n/2,n/4\}$. We therefore present the inference results for prior $\Bar{k}=n/2$ in this section. We notice that the posterior for $\rho$ tends to concentrate on values close to one. This is because $K$ and $\rho$ are only weakly identifiable. For any reasonable geometric prior on $K$, the model can fit well by adjusting the $\rho$ posterior (as $K$ gets larger the hazard for paths to cross increases, so make $\rho$ closer to one for flatter paths and a similar distribution over partial orders). According to our sensitivity analysis to $K$-prior, when the $K$-prior concentrates on a small mean-$K$, the posterior for $\rho$ is supported on  lower values. However, the freedom of choosing both $K$ and $\rho$ provides us more flexible control over partial order depth in the prior. The depth distribution in Figure~\ref{fig:depthp} cannot be realised with a single parameter of choice. 

In the Mallows observation model, the dispersion parameter $\theta$ captures the level of noise in the observed list. The $\theta$ priors are adjusted to allow roughly 10\% of noise in lists as we expect errors were rare in these data. The posterior means for $\theta$ in the $(PDP,\rho)\backslash M$ models are 1.75 and 2.94 in 1131-1133 and 1100-1103 respectively, somewhat lower than the prior mean, so somewhat higher noise than expected a priori. In particular, the witness lists in period 1131-1133 appear to have higher noise compared to 1100-1103. This observation is also supported in the inference from the $(PDP,\rho)\backslash Q$ model. With the queue-jumping observation model, the parameter $p$ indicates the error probability - the probability for a bishop to jump up the queue. The $p$-priors are relatively flat, whereas the posterior distributions for $p$ both indicate an error rate of around 30\% and a relatively concentrated. 

Overall, we obtain consistent inference results from the $(PDP,\rho)\backslash M$ and $(PDP,\rho)\backslash Q$ models in both time periods. We conclude that \textit{Roger, Bishop of Salisbury}, \textit{Henry, de Blois, Bishop of Winchester, 1129-1171} (tied) and \textit{Gilbert, the Universal, bishop of London} are of the highest power hierarchy in 1131-1133 while \textit{Everard, bishop of Norwich} is at the bottom of such order relation. Similarly, in 1100-1103, \textit{Maurice, bishop of London} is of the highest hierarchy while \textit{Ranulf, Flambard, bishop of Durham} and \textit{Roger, Bishop of Salisbury} has the lowest. Numerous factors might contribute to this social power relation, e.g. the area one bishop represent, the age of the bishop, etc. 
We further conclude that there are slight recording errors in the order of witness signature, which can be effectively captured by our model.   

\subsection{Model Comparison}

The $(PDP,\rho)$-partial order model is preferred over the $(K,n,\rho)$-partial order model according to the Bayes Factor. If $M_0$ is the $(K,n,\rho)$-partial order model (with $K$ fixed to $\lceil n/2\rceil$) and $M_1$ is the $(PDP,\rho)$-partial order model, then $M_0$ is nested in $M_1$ (when each actor is in a separate cluster in $M_1$ so $C(S)=n$ (the number of clusters in $S$), $S=(\{1\},\dots,\{n\})$ and the random $K$-value in $M_1$ takes the value $\lceil n/2\rceil$). The Savage-Dickey Density Ratio $B_{10}=\frac{\pi_{M_1}(C=n,K=\lceil n/2 \rceil)}{\pi_{M_1}(C=n,K=\lceil n/2 \rceil|\mathbf{y})}$, which gives the Bayes factor in this case, is easily estimated using samples from the $M_1$ posterior and prior. Values bigger than one are evidence for model $M_1$. As is shown in table \ref{table:bf-tie}, the Bayes factors under different observation models all favor the $(PDP,\rho)$-partial order model for both 1131-1133 and 1100-1103. Clustering and varying $K$ effectively reduce the dimension of $Z$-matrix without compromising the fit. Ties enforce equal relations to other actors, and it seems the data ``likes'' a prior that encourages this sort of structure. The Bayes factors values parsimony so this may also contribute to the weighting.  

\begin{table}[h!]
\centering
\begin{tabular}{r*2c}
\toprule
& \multicolumn{2}{c}{1131-1133} \\
\midrule
 & Mallows & Queue-jumping\\
\midrule
 $B_{10}$ & $4.2$ & $3.2$ \\
 Evidence for $M_1$ & ``substantial'' & ``weak''\\ 
 \midrule
 & \multicolumn{2}{c}{1100-1103}\\
 \midrule
 & Mallows & Queue-jumping\\
 \midrule
 $B_{10}$ & $7.1$ & $33$\\
  Evidence for $M_1$ & ``substantial'' & ''strong''\\
 \bottomrule
\end{tabular}
\caption{Bayes factors $B_{10}$ between the $(PDP,\rho)$-partial order model ($M_1$ \textit{with ties and variable latent matrix dimension}) and the $(K,\lceil n/2 \rceil,\rho)$-partial order model ($M_0$, with $K=\lceil n/2\rceil$ fixed) with either the Mallows or queue-jumping observation models in periods 1131-1133 and 1100-1103. A large Bayes factor indicates preference towards the $(PDP.\rho)$-partial order model. Interpretation follows Jefferys.}
\label{table:bf-tie}
\end{table}

We further compare our models with the popular ranking models - the Plackett-Luce and Mallows Mixture. It would be hard to estimate Bayes factors in this case so we use the expected log pointwise predictive density ($elpd$, \cite{vehtari2017practical}) as our model comparison criterion. This is a predictive loss which can be estimated using the WAIC (\cite{Watanabe2012}). We present $elpd_{waic}$-estimates in table \ref{table:model-comparison}. The partial order models have significantly higher $elpd_{waic}$, and are therefore strongly preferred compared to the total order methods on the 1131-1133 and 1100-1103 Bishop witness list data. Under the $(PDP,\rho)$-partial order model, the Mallows and queue-jumping models perform almost equally well - with the Mallows model being slightly preferred. This is as expected as the Mallows model allows bi-directional errors while the queue-jumping model only allows queue-jumping-up. For the Mallows or Plackett-Luce mixture models, we choose the number of mixture components $M$ that provides the highest $elpd_{waic} (M\in\{1,\dots,10\})$. The Mallows mixture appears to outperform the Plackett-Luce mixture. 

\begin{table}[h!]
\centering
\resizebox{\textwidth}{!}{%
\begin{tabular}{*3c}
\toprule
& \multicolumn{2}{c}{$elpd_{waic}$}\\
\midrule
Models & 1131-1133 & 1100-1103\\
\midrule
$(PDP,\rho)\backslash M$ & -48.2 (9.5) & -23.3 (10.4)\\
$(PDP,\rho)\backslash Q$ & -56.6 (10.4) & -26.0 (11.4) \\
Mallows Mixture (Kendall-tau) & -72.9 (13.0) (M=1) & -33.9 (12.0) (M=3)\\
Plackett-Luce Mixture & -173.2 (17.8) (M=4) & -77.5 (10.4) (M=3) \\
\bottomrule
\end{tabular}}
\caption{Model comparison among $(PDP,\rho)\backslash M$, $(PDP,\rho)\backslash Q$, Plackett-Luce Mixture and Mallows Mixture (with Kendall-tau distance) on time periods 1131-1133 and 1100-1103. We use $elpd_{waic}$ as the model comparison criterion. The results are presented as `estimation (standard error)'. For the Mallows and Plackett-Luce mixture models, we only consider models with 1-10 mixture components. Here we report the best $elpd_{waic}$ with its corresponding number of mixture component $M$. }
\label{table:model-comparison}
\end{table}

\section{Conclusion and Future Work}

We have proposed a new non-parametric statistical model for partial order estimation under the Bayesian framework. The new $(PDP,\rho)$-partial order model improves on the model set out in \cite{nicholls122011partial} by allowing actors to be tied. The more complex model remains projective. The Mallows error model improves on the queue-jumping error model of \cite{nicholls122011partial}. We also estimate, rather than fix, the column dimension of the latent matrix $Z$ introduced by \cite{nicholls122011partial} to model the prior for partial orders. The tied-structure of the new model further reduces the dimension of the latent $Z$-matrix by grouping its rows into equivalence classes. It changes the prior distribution over partial orders to promote more `bucket' order - like structures. This is appealing from a modeling perspective.
%The ties are formulated via a Dirichlet process to classify the actors before fitting a partial order model. Depending on the data generating process, the ties effectively capture the inherenet equivalence between actors, e.g. the unobserved equivalence on power relation between the bishops. It distinguishes from the incomparability in partial orders such that the tied actors are believed to be the same on a certain matrix rather than completely incomparable in nature. 
Experiments on the `Royal Acta' data showed that the tied model out-performed the nested model without ties and fixed column dimension. 

% Different observation models can be chosen given different data generating processes. 
% Both the queue-jumping error model and the Mallows model relax the noise free partial order model by allowing error (but including the noise-free partial order as a special case). We perform Bayesian inference on both models on the twelfth century witness list data. The posterior distribution is sampled via a MCMC scheme. We also considered adapting the Plackett-Luce model to list data from partial orders. To increase model versatility, we assume separate actor weights to each data-list. This promotes a desirable feature that the rank-vector for the actor weights are uniform random linear extensions. This setup, however, results in an inefficient inference scheme.

We also added some theory for the model given in \cite{nicholls122011partial} in the noise free case. We gave necessary and sufficient conditions for the posterior to concentrate on the true partial order. We can ask, if the true model was a partial order model, and we fit say, a Mallows or Plackett-Luce model, what behaviour do we get? Is it possible that the posterior converges to a mixture of rankings rather than a single ranking, and thereby captures some of the ``uncertainty'' which a partial order allows. For the simplest possible example, consider the partial order $h_0$ in figure \ref{POex} with three linear extensions in figure \ref{LEex}. Let $N$ denote the number of ranked lists sampled uniformly at random from the linear extensions of partial order $h_0$. We take $N\rightarrow\infty$. It is not too hard to see that the Mallows model posterior with Hamming or Cayley distance must concentrate on $(1,3,2,4,5)$ in the large data limit (rather than say, a posterior which is uniform on the linear extensions of $h_0$ as we might hope). If we could get a Mallows-mixture to converge to the uniform distribution on the linear extensions of the true partial order it would then be possible to reconstruct the true partial order by intersecting total orders sampled from the posterior. Such an analysis is probably impossible in principle: fitting partial order models is a computational task in $\#P$; fitting a Mallows model with any likelihood which can be evaluated in polynomial time is not in $\#P$. It is unlikely we can solve the harder problem by solving the easier problem. Fundamentally though, the reason we work with partial orders is that we think social hierarchies are well-represented by partial orders, that any total order assumption is too strong, and so if we want to reconstruct partial orders they should be an object in the model.

% We can consider our approach to be equivalent to updating linear extensions in which case the partial order of interest can be defined as the intersection order from the linear extensions. Although this is physically motivated, its construct imposes difficulty on controlling the depths of the partial orders. This linear extension approach may be similar to the Mallows mixture model but with varied number of clusters. Ideally, the number of clusters in each update should be equal to the dimension for the partial order of interest. One may define a fixed number of clusters of $n\geq dim(h)$. This is similar to the latent matrix setup. However, the lack of control over depth still persists. 

Future research may consider including covariates in the partial order model, following \cite{nicholls2023bayesian} but adding our non-parametric approach, or exploring a different definition of ties where the tied actors have to appear together in a linear extension. 

% Our flexible setup also allows for future improvement. One advantage of the Plackett-Luce observation model is that its parameterisation provides a flexible way to include covariates. A possible way is to modify the likelihood function is to assume $L(\alpha_j; y_j)= \prod_{m=1}^n \frac{e^{\alpha_{j, y_{j,m}}+\mathbf{\beta}X}}{\sum_{m' = m}^n e^{\alpha_{j, y_{j, m'}}+\mathbf{\beta}X}}$, where $X$ stands for the covariates and $\mathbf{\beta}$ represents their coefficients. This setup enables us to assess and quantify the effect of different factors to the final partial order relation. Another possible extension to the partial order model is to add in a temporal component via a stochastic process on partial orders. The model presented in the current paper is static, although we could analyse temporal datasets by taking discrete snapshots. Including a time component to our model would allow for more convenient model fitting and easier assessment on the time effect. Further future work may allow the vectors of $\alpha$'s from one list to the other to be correlated, for example, as a realisation of a Dirichlet process. This would include the existing model, the noise free partial order model, and Plackett-Luce itself (with a single vector of $\alpha$'s) in a single framework.

%% For citations use: 
%%       \citet{<label>} ==> Lamport (1994)
%%       \citep{<label>} ==> (Lamport, 1994)
%%
% Example citation, See \citet{lamport94}.

\bibliographystyle{elsarticle-harv} 
\bibliography{main}

%% The Appendices part is started with the command \appendix;
%% appendix sections are then done as normal sections
\appendix
\newpage 

\section{The intersection order is the maximum likelihood estimator}\label{appx-MLE}

The material in the appendix follows the ideas of \cite{beerenwinkel2007conjunctive} forly closely.

\begin{proposition}
Given noise-free data lists $\mathbf{y} = \{y_1, y_2, \dots, y_N\}$ with full lengths, i.e. $y_j\in\P_n \, \forall j\in[N]$, and the uniform observation model on linear extensions $$f(y_j|h) = \frac{\mathbbm{1}_{y_j\in \mathcal{L}[h]}}{|\mathcal{L}[h]|}.$$ The intersection order $h_{int}$ is the maximum likelihood estimator of the true partial order.
\end{proposition}

\begin{proof}
The likelihood function given the uniform observation model is 
\begin{align*}
    L(h; \mathbf{y}) & = \prod_{j=1}^N \frac{\mathbbm{1}_{y_j \in \mathcal{L}[h]}}{|\mathcal{L}[h]|}
     = \frac{1}{|\mathcal{L}[h]|^N}\prod_{j=1}^N \mathbbm{1}_{y_j \in \mathcal{L}[h]}
\end{align*}
The intersection order $h_{int}$ admits $u \succ_{h_{int}} v$ with $u, v \in [n] \text{ and } u\neq v$ if and only if $u$ appears before $v$ in all $\mathbf{y}$. For $j\in[N]$ and $1\leq a<b\leq n$, we have either $y_{j, a}\sim y_{j,b}$ or $y_{j, a}\succ y_{j,b}$ in $h_{int}$. Therefore, $\prod_{j=1}^N \mathbbm{1}_{y_j \in \mathcal{L}[h_{int}]} = 1$.

Let $h_{r} \in\H_{[n]}$ with $h_{r} \neq h_{int}$ be some partial order  such that $y_j \in \mathcal{L}[h_{r}]\, \forall j\in[N]$. If $h_r$ admits one more order relation than $h_{int}$, at least one $y_j$ will be violated. So the condition $y_j \in \mathcal{L}[h_{r}]\, \forall j\in[N]$ wouldn't hold. That is, for all the order relations in $h_{int}$, there exists at least one pair of actors $u,v\in [n], u\neq v$ with $(u \succ_{h_{int}} v)$ such that $u\|_{h_r} v$. This gives $|\mathcal{L}[h_{int}]| < |\mathcal{L}[h_r]|$. Therefore, 
\begin{equation*}
    L(h_{int};Y) > L(h;Y)\; \forall h\in\H_{[n]}, h \neq h_{int}
\end{equation*}
Hence, the intersection order $h_{int}$ is the maximum likelihood estimator of the true partial order given full-length noise-free list data and the uniform observation model on linear extensions. 
\end{proof}

\section{Partial orders on three actors}\label{appx-edge}

The 19 possible edge sets $E_3$ are shown in \Fig~\ref{PO3}.
\begin{multicols}{3}
    \begin{minipage}{0.9\linewidth}
    \centering
    \begin{tikzpicture}[thick,scale=1.07, every node/.style={scale=0.8}]
        \node[draw, circle, minimum width=.1cm] (1) at (-0.5, 0.5) {$1$};
        \node[draw, circle, minimum width=.1cm] (2) at (0.5, 0.5) {$2$};
        \node[draw, circle, minimum width=.1cm] (3) at (0, -0.25) {$3$};
        \coordinate (a1) at (-0.33, 0.435) ;
        \coordinate (a2) at (-0.33, 0.565) ;
        \coordinate (b1) at (0.33, 0.435) ;
        \coordinate (b2) at (0.33, 0.565) ;
        \draw[-latex] (a1) -- (b1);
        \draw[-latex,color=blue] (b2) -- (a2);
    \end{tikzpicture}
    %\captionof{figure}{A suborder with actors $o=\{2,4,5\}\in [n]$ of the partial order in \Fig~\ref{POex}.}\label{subPO}
    \end{minipage}%
    \begin{minipage}{0.9\linewidth}
    \centering
    \begin{tikzpicture}[thick,scale=1.07, every node/.style={scale=0.8}]
        \node[draw, circle, minimum width=.1cm] (1) at (-0.5, 0.5) {$1$};
        \node[draw, circle, minimum width=.1cm] (2) at (0.5, 0.5) {$2$};
        \node[draw, circle, minimum width=.1cm] (3) at (0, -0.25) {$3$};
        \coordinate (a1) at (-0.30, 0.37) ;
        \coordinate (a2) at (-0.42, 0.35) ;
        \coordinate (b1) at (-0.07, -0.1) ;
        \coordinate (b2) at (-0.2, -0.12) ;
        \draw[-latex] (a1) -- (b1);
        \draw[-latex,color=blue] (b2) -- (a2);
    \end{tikzpicture}
    \end{minipage}%
    \begin{minipage}{0.9\linewidth}
    \centering
    \begin{tikzpicture}[thick,scale=1.07, every node/.style={scale=0.8}]
        \node[draw, circle, minimum width=.1cm] (1) at (-0.5, 0.5) {$1$};
        \node[draw, circle, minimum width=.1cm] (2) at (0.5, 0.5) {$2$};
        \node[draw, circle, minimum width=.1cm] (3) at (0, -0.25) {$3$};
        \coordinate (a1) at (0.30, 0.37) ;
        \coordinate (a2) at (0.42, 0.35) ;
        \coordinate (b1) at (0.07, -0.1) ;
        \coordinate (b2) at (0.2, -0.12) ;
        \draw[-latex] (a1) -- (b1);
        \draw[-latex,color=blue] (b2) -- (a2);
    \end{tikzpicture}
    \end{minipage}
\end{multicols}

\begin{multicols}{3}
    \begin{minipage}{0.9\linewidth}
    \centering
    \begin{tikzpicture}[thick,scale=1.07, every node/.style={scale=0.8}]
        \node[draw, circle, minimum width=.1cm] (1) at (-0.5, 0.5) {$1$};
        \node[draw, circle, minimum width=.1cm] (2) at (0.5, 0.5) {$2$};
        \node[draw, circle, minimum width=.1cm] (3) at (0, -0.25) {$3$};
        \coordinate (a1) at (-0.33, 0.435) ;
        \coordinate (a2) at (-0.33, 0.565) ;
        \coordinate (b1) at (0.33, 0.435) ;
        \coordinate (b2) at (0.33, 0.565) ;
        \coordinate (c1) at (-0.30, 0.37) ;
        \coordinate (c2) at (-0.44, 0.37) ;
        \coordinate (d1) at (-0.07, -0.1) ;
        \coordinate (d2) at (-0.2, -0.12) ;
        \draw[-latex] (a1) -- (b1);
        \draw[-latex,color=blue] (b2) -- (a2);
        \draw[-latex] (c1) -- (d1);
        \draw[-latex,color=blue] (d2) -- (c2);
    \end{tikzpicture}
    %\captionof{figure}{A suborder with actors $o=\{2,4,5\}\in [n]$ of the partial order in \Fig~\ref{POex}.}\label{subPO}
    \end{minipage}%
    \begin{minipage}{0.9\linewidth}
    \centering
    \begin{tikzpicture}[thick,scale=1.07, every node/.style={scale=0.8}]
        \node[draw, circle, minimum width=.1cm] (1) at (-0.5, 0.5) {$1$};
        \node[draw, circle, minimum width=.1cm] (2) at (0.5, 0.5) {$2$};
        \node[draw, circle, minimum width=.1cm] (3) at (0, -0.25) {$3$};
        \coordinate (a1) at (-0.30, 0.37) ;
        \coordinate (a2) at (-0.42, 0.35) ;
        \coordinate (b1) at (-0.07, -0.1) ;
        \coordinate (b2) at (-0.2, -0.12) ;
        \coordinate (c1) at (0.30, 0.37) ;
        \coordinate (c2) at (0.42, 0.35) ;
        \coordinate (d1) at (0.07, -0.1) ;
        \coordinate (d2) at (0.2, -0.12) ;
        \draw[-latex] (a1) -- (b1);
        \draw[-latex,color=blue] (b2) -- (a2);
        \draw[-latex] (c1) -- (d1);
        \draw[-latex,color=blue] (d2) -- (c2);
    \end{tikzpicture}
    \end{minipage}%
    \begin{minipage}{0.9\linewidth}
    \centering
    \begin{tikzpicture}[thick,scale=1.07, every node/.style={scale=0.8}]
        \node[draw, circle, minimum width=.1cm] (1) at (-0.5, 0.5) {$1$};
        \node[draw, circle, minimum width=.1cm] (2) at (0.5, 0.5) {$2$};
        \node[draw, circle, minimum width=.1cm] (3) at (0, -0.25) {$3$};
        \coordinate (a1) at (-0.33, 0.435) ;
        \coordinate (a2) at (-0.33, 0.565) ;
        \coordinate (b1) at (0.33, 0.435) ;
        \coordinate (b2) at (0.33, 0.565) ;
        \coordinate (c1) at (0.30, 0.37) ;
        \coordinate (c2) at (0.42, 0.35) ;
        \coordinate (d1) at (0.07, -0.1) ;
        \coordinate (d2) at (0.2, -0.12) ;
        \draw[-latex] (b1) -- (a1);
        \draw[-latex,color=blue] (a2) -- (b2);
        \draw[-latex] (c1) -- (d1);
        \draw[-latex,color=blue] (d2) -- (c2);
    \end{tikzpicture}
    \end{minipage}
\end{multicols}

\begin{multicols}{3}
    \begin{minipage}{0.9\linewidth}
    \centering
    \begin{tikzpicture}[thick,scale=1.07, every node/.style={scale=0.8}]
        \node[draw, circle, minimum width=.1cm] (1) at (-0.5, 0.5) {$1$};
        \node[draw, circle, minimum width=.1cm] (2) at (0.5, 0.5) {$2$};
        \node[draw, circle, minimum width=.1cm] (3) at (0, -0.25) {$3$};
        \coordinate (a1) at (-0.33, 0.435) ;
        \coordinate (a2) at (-0.33, 0.565) ;
        \coordinate (b1) at (0.33, 0.435) ;
        \coordinate (b2) at (0.33, 0.565) ;
        \coordinate (c1) at (-0.30, 0.37) ;
        \coordinate (c2) at (-0.42, 0.35) ;
        \coordinate (d1) at (-0.07, -0.1) ;
        \coordinate (d2) at (-0.2, -0.12) ;
        \draw[-latex] (b1) -- (a1);
        \draw[-latex,color=blue] (a2) -- (b2);
        \draw[-latex] (c1) -- (d1);
        \draw[-latex,color=blue] (d2) -- (c2);
    \end{tikzpicture}
    %\captionof{figure}{A suborder with actors $o=\{2,4,5\}\in [n]$ of the partial order in \Fig~\ref{POex}.}\label{subPO}
    \end{minipage}%
    \begin{minipage}{0.9\linewidth}
    \centering
    \begin{tikzpicture}[thick,scale=1.07, every node/.style={scale=0.8}]
        \node[draw, circle, minimum width=.1cm] (1) at (-0.5, 0.5) {$1$};
        \node[draw, circle, minimum width=.1cm] (2) at (0.5, 0.5) {$2$};
        \node[draw, circle, minimum width=.1cm] (3) at (0, -0.25) {$3$};
        \coordinate (a1) at (-0.30, 0.37) ;
        \coordinate (a2) at (-0.42, 0.35) ;
        \coordinate (b1) at (-0.07, -0.1) ;
        \coordinate (b2) at (-0.2, -0.12) ;
        \coordinate (c1) at (0.30, 0.37) ;
        \coordinate (c2) at (0.42, 0.35) ;
        \coordinate (d1) at (0.07, -0.1) ;
        \coordinate (d2) at (0.2, -0.12) ;
        \draw[-latex] (b1) -- (a1);
        \draw[-latex,color=blue] (a2) -- (b2);
        \draw[-latex] (c1) -- (d1);
        \draw[-latex,color=blue] (d2) -- (c2);
    \end{tikzpicture}
    \end{minipage}%
    \begin{minipage}{0.9\linewidth}
    \centering
    \begin{tikzpicture}[thick,scale=1.07, every node/.style={scale=0.8}]
        \node[draw, circle, minimum width=.1cm] (1) at (-0.5, 0.5) {$1$};
        \node[draw, circle, minimum width=.1cm] (2) at (0.5, 0.5) {$2$};
        \node[draw, circle, minimum width=.1cm] (3) at (0, -0.25) {$3$};
        \coordinate (a1) at (-0.33, 0.435) ;
        \coordinate (a2) at (-0.33, 0.565) ;
        \coordinate (b1) at (0.33, 0.435) ;
        \coordinate (b2) at (0.33, 0.565) ;
        \coordinate (c1) at (0.30, 0.37) ;
        \coordinate (c2) at (0.42, 0.35) ;
        \coordinate (d1) at (0.07, -0.1) ;
        \coordinate (d2) at (0.2, -0.12) ;
        \draw[-latex] (a1) -- (b1);
        \draw[-latex,color=blue] (b2) -- (a2);
        \draw[-latex] (c1) -- (d1);
        \draw[-latex,color=blue] (d2) -- (c2);
    \end{tikzpicture}
    \end{minipage}
\end{multicols}

\begin{minipage}{\linewidth}
\centering
\begin{tikzpicture}[thick,scale=1.07, every node/.style={scale=0.8}]
    \node[draw, circle, minimum width=.1cm] (1) at (-0.75, 0) {$1$};
    \node[draw, circle, minimum width=.1cm] (2) at (0, 0) {$2$};
    \node[draw, circle, minimum width=.1cm] (3) at (0.75, 0) {$3$};
\end{tikzpicture}

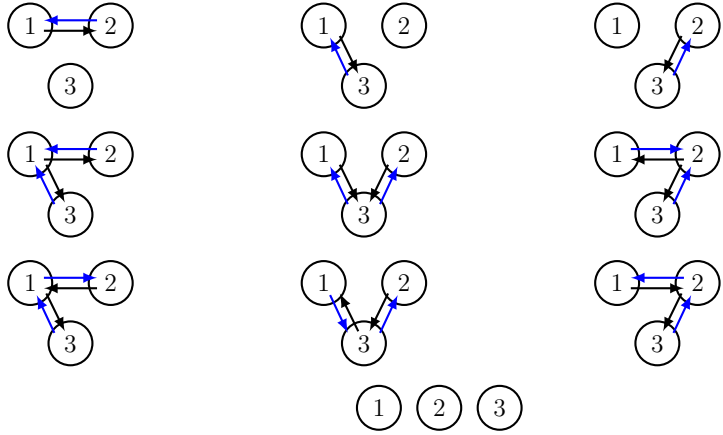
\captionof{figure}{All possible partial orders given 3 actors. The blue and black edges represent different partial orders respectively.}\label{PO3}
\end{minipage}%

\section{Proof of Proposition~\ref{prop:projective-kn}}\label{appx-kr-projectivity}

{\bf Proposition~\ref{prop:projective-kn}.}
{\it The $(K,n)$-partial order model distributions are projective.}

\begin{proof}
For $h\in\H_{[n]}$ let $Z[h]=\{Z'\in R^{n\times K}: h(Z')=h\}$ be the set of $Z$-matrices generating $h$ so that if $\pi_{[n]}(h)=\Pr(H=h)$ under the $(K,n)$-model then $\pi_{[n]}(h)=\Pr(Z\in Z[h])$. Let $o=[n-1]$. 
%For any $h\in \H_{n-1}$ let $\H_{h,n}=\{h'\in\H_{[n]}: h'[o]=h\}$ be the set of partial orders on $n$ actors which contain $h$ as a sub-order. 
Let $H_{n-1}\sim \pi_{[n-1]}$ and $H_n\sim \pi_{[n]}$. It is sufficient for marginal consistency to show that $H_n[o]\sim H_{n-1}$, that is $\Pr(H_{n-1}=h)=\Pr(H_n[o]=h)$ for all $h\in \H_{n-1}$. We begin on the RHS. 

Let $Z\in {[0,1]}^{n\times K}$ with $Z_{i,{1\dv K}}\sim \mathcal{U}(0,1)^K,\ i=1,...,n$ so that $H_n\sim h(Z)$.
We have $$\Pr(H_n[o]=h)=\Pr(h(Z)[o]=h).$$
Next we observe that $h(Z)[o]=h(Z[o,{1\dv K}])$ and so $$\Pr(h(Z)[o]=h)=\Pr(h(Z[o,{1\dv K}])=h).$$
This is because relations between any two rows of $Z$ are not effected by the values in any other rows so the suborder $h(Z)[o]$ we get by dropping actors outside $o$
is the same as the suborder we get if we compute it on the rows of $Z$ belonging to the actors in $o$. Now let $Z'\in \mathbb{R}^{(n-1)\times K}$ with $Z'_{i,{1\dv K}}\sim \mathcal{U}(0,1)^K,\ i=1,...,n-1$.
We have $Z'\sim Z[o,{1\dv K}]$ since the rows are all independent. So we may write 
$$\Pr(h(Z[o,{1\dv K}])=h)=\Pr(h(Z')=h)$$ and we are done. 
\end{proof}

\section{Proof of Proposition~\ref{prior-mc}}\label{appx-projectivity}

{\bf Proposition~\ref{prior-mc}.}
{\it For $h^*\in \H^*_n$ let $\pi_{[n]}(h^*)=\Pr(H^*=h^*)$ be the distribution of the random partial order with ties generated by the process in Definition~\ref{def:gen-proc-ties-model} and given in Equation~\ref{eq:marginal-prior-ties}.
The family of distributions $\pi_{[n]}, n\ge 1$ is projective.}

\begin{proof}
Let $o=[n-1]$, let $H^*_n\sim \pi_{[n]}(\cdot|\rho,k)$ defined in Equation~\ref{candprior} and let $H^*_{n-1}\sim \pi_{[n-1]}(\cdot|\rho,k)$.  We wish to show that $H^*_n[o]\sim H^*_{n-1}$. It is again sufficient to show $H^*_n[o]|\rho,k\sim H^*_{n-1}|\rho,k$ for every fixed $\rho$ and $K$; if $\Pr(H^*_n[o]=h^*|\rho,k)=\Pr(H^*_{n-1}=h^*|\rho,k)$ then the marginals $\Pr(H^*_n[o]=h^*)=\Pr(H^*_{n-1}=h^*)$ are equal (just average over $\rho$ and $K$ on both sides).  
In order to verify $H^*_n[o]\sim H^*_{n-1}$ we start with $\Pr(H^*_n[o]=h^*|\rho,k)$. First of all, by the generative model, $H^*_n\sim h(Z(Z^*,S))$ if $S\sim P_{[n],\eta_b,\eta_a}$ and 
\[
Z^*_{c,1\dv K}\sim \mathcal{MVN}(\mathbf{0}, \mathbf{\Sigma}_\rho)
\] 
i.i.d. for $c=1,...,C$. Substituting for $H^*_n$,
\[
\Pr(H^*_n[o]=h^*|\rho,k)=\Pr(h(Z(Z^*,S))[o]=h^*|\rho,k).
\] 
Now 
\[
h(Z(Z^*,S))[o]=h(Z(Z^*,S)[o,1\dv k])
\]
because the relations between actors $1,...,n-1$ determined by $h(\cdot)$ acting on $Z[o,1\dv k]$ are not effected by the presence or values of $Z[n,1\dv k]$ so we just get the suborder.

Now let $S^{-n}=(S_1,...,S_{C^{-n}})$ be the partition $S$ with $n$ removed. Let $Z^{*,-n}$ be the corresponding parameter matrix. If $n$ is in a partition by itself, which must be $S_C=\{n\}$, then that set is deleted and $C^{-n}=C-1$; otherwise $C^{-n}=C$. If $C^{-n}=C$ then $Z^{*,-n}=Z^*$ and otherwise $Z^{*,-n}=Z^*_{1\dv C-1,1\dv k}$, so the last row of $Z^*$ is dropped.

In terms of these quantities, $Z(Z^*,S)[o,1\dv k]=Z(Z^{*,-n},S^{-n})$, since again, if we remove $n$ from the partition $S$ then use $S^{-n}$ with $Z^{*,-n}$ to determine the rows of $Z$ for $1,...,n-1$ in $Z$ we get the same reduced $Z$ as if we determined the entries in $Z$ from $(Z^*,S)$ and then dropped row $n$ from $Z$.
We have then 
\[\Pr(H^*_n[o]=h^*|\rho,k)=\Pr(h(Z(Z^{*,-n},S^{-n}))=h^*|\rho,k).\]

The two parameter Poisson-Dirichlet process generates exchangeable random partitions \cite{pitman2006combinatorial}. This indicates each item is treated independently, as if they are the `last arrival'. 

The two parameter Poisson-Dirichlet process generates exchangeable random partitions \cite{pitman2006combinatorial}. This indicates each item is treated independently, as if they are the `last arrival'. Given the formulation of the two-parameter Poisson-Dirichlet process , we have $S^{-n}\sim P_{[n-1],\eta_b,\eta_a}$. In addition, the latent matrix setup gives $Z^{*,-n}_{c,1\dv k}\sim \mathcal{MVN}(\mathbf{0}, \mathbf{\Sigma}_\rho)$ i.i.d. for $c=1,...,C^{-n}$. It follows that
$h(Z(Z^{*,-n},S^{-n}))$ is a realisation from the generative model at fixed $\rho, k$ on $[n-1]$ and hence $h(Z(Z^{*,-n},S^{-n}))\sim H^*_{n-1}$.
We conclude that $\Pr(H^*_n[o]=h^*|\rho,k)=\Pr(H^*_{n-1}=h^*|\rho,k)$ for every fixed $\rho,k$, and hence, from the preamble $H^*_n[o]\sim H^*_{n-1}$. We have shown marginal consistency for the special case of $o=[n-1]$ but the actor labels are exchangeable,
so we may remove any subset $o\subset [n]$ by removing one entry in $[n]\setminus o$ at a time, sorting it to be the last ``arrival'' in the Poisson-Dirichlet process clustering labels.
\end{proof}

\section{Proof of Proposition \ref{Prop:MC-unif}}\label{appx-MC}

The key here is to show the chain is irreducible. This is well known, for example, \cite{Karzanov91} make use of this fact for a closely related chain.

{\bf Proposition \ref{Prop:MC-unif}.}
{\it Consider a Markov chain $\{X_t\}_{t\ge 1}$, with state space $X_t\in \mathcal{L}[h[o_j]]$ for some fixed $j\in\{1,\dots,N\}$. Suppose that at step $t$ we have $X_t=x$. An entry $i\sim U\{1,...,n_j\}$ is chosen uniformly at random. If $i=n_j$ then we reject and set $X_{t+1}=x$ and otherwise $x'=(x_1,...,x_{i+1},x_i,...,x_{n_j})$.
If $x'\in \mathcal{L}[h[o_j]]$ then we set $X_{t+1}=x'$ and otherwise $X_{t+1}=x$.
The process converges in distribution to the uniform distribution over linear extensions of $h[o_j]$, that is, for $l\in \L[h[o_j]]$,
\begin{equation*}
    P(X_n=l)\overset{t\to\infty}{\to}\frac{1}{|\mathcal{L}[h[o_j]]|}\mathbbm{1}_{l\in \mathcal{L}[h[o_j]]}. 
\end{equation*}
}

\begin{proof}
This finite-state Markov chain $\{X_t\}_{t\ge 1}$ is irreducible. Let $l,l'\in\mathcal{L}[h[o_j]]$ with $l\ne l'$ be a pair of linear exentsions in the target space. Let $a=l_1$ and $b=l'_1$ be the first (or ``top'') entries in the lists and suppose $a\ne b$, so the top entries in the lists do not match. Suppose $l'_i=a$ with $i>1$. We will show that we can swap $a$ up to the top of $l'$. Suppose $c=l'_{i-1}$ with $c=b$ if $i=2$. Now we must be able to swap $a$ and $c$ in $l'$ without violating an order in $h[o_j]$ as $a$ appears below $c$ in $l'$ but above $c$ in $l$ (as $a$ is top in $l$ so must be above $c$) so we have both orders $l=(a,...,c,...)$ and $l'=(b,...c,a,...)$. It follows that we must have $a\|_h c$; they are unordered in $h$. Iterating this procedure we can move from $l'$ to $l$ by a finite sequence of swaps each of which have probability $1/n_j$.

It is aperiodic (since it is irreducible and rejects when $i=n_j$, at least), and hence ergodic, so it has a unique stationary distribution $\pi$ over $\mathcal{L}[h[o_j]]$. For $l,l'\in\mathcal{L}[h[o_j]]$ with $l\ne l'$ let $p_{l,l'}=\Pr(X_{t+1}=l'|X_t=l)$ denote the off diagonal elements of the transition matrix. We have
\[
p_{l,l'}=\left\{\begin{array}{cc}
     1/n_j & \mbox{if $l,l'$ differ by a neighbor-swap} \\
     0 & \mbox{otherwise}.
\end{array}\right.
\]
Since $p$ is symmetric, its columns sum to one, so it admits the uniform distribution as a left eigenvector, $\mathbf{1}^T p=\mathbf{1}^T$ with eigenvalue one, where $\mathbf{1}$ is a vector of $|\mathcal{L}[h[o_j]]|$-ones, so 
\[
\pi_l = \frac{1}{|\mathcal{L}[h[o_j]]|},\quad l\in \mathcal{L}[h[o_j]].
\]
\end{proof}

\section{Model Consistency for the Partial Order Model with the Error-Free Observation Model}\label{C-EF}

% \begin{proposition}\label{PONE-Consistency}
%     For $k\in[N]$, denote $o_k \subseteq [n]$ a subset of actors corresponding to the data list $y_k$.  When for any pair of actors $\forall (i,j)\in[n]\times [n]$ with $i\neq j$, $\exists k\in [N]$ such that $(i,j)\in o_k$, the partial order model with the error-free observation model is consistent.
% \end{proposition}

Recall from the notation of Section~\ref{sec:asymptotic-post} that $O=(o_i)_{i=1}^{2^n-1}$ is the set of all subsets of the actor indices and $I\subseteq [2^n-1]$ is a set of the indices of sets in $O$ for which we have observations, so we only accumulate lists for the subsets $o_i,\ i\in I$ not for every subset in $O$. For $i\in I$ we have $N_i$ realisations $y_{i,1},...,y_{i,N_k}$ of the noise free observation model $p^{(P)}(\cdot|\htrue[o_i])\propto 1/|\L[\htrue[o_i]]|$.

{\bf Proposition~\ref{prop:PONE-Consistency}.}
{\it Let $h^\dagger\in \H_{[n]}$ be given and suppose $y_{i,j}\sim p^{(P)}(\cdot|h^\dagger[o_i])$ jointly independent for all $i\in I$ and $j=1,\dots N_i$.
    Let $\pi(h|\mathbf{y}_{1:N})$ be the posterior of the $(K,n,\rho)\backslash P$ model with $K\ge \lfloor n/2\rfloor$, so that
    $$\pi(h|\mathbf{y}_{1:N}) = \int_{Z:h(Z)=h}\int_{\rho=0}^1 \pi(Z,\rho|\mathbf{y}) d\rho dZ$$
    where
    $$\pi(Z,\rho|\mathbf{y}_{1:N}) \propto p^{(P)}(\mathbf{y}_{1:N}|h(Z))\pi(Z|\rho)\pi_\rho(\rho).$$
    If for each pair of actors $(i,j)\in[n]\times [n]$ with $i\neq j$, there exists $k\in I$ such that $\{i,j\}\subseteq o_k$ then $\pi(\htrue|\mathbf{y}_{1:N})\to 1$ as $\min_{i\in I} N_i\to \infty$ for every $h^\dagger\in\H_{[n]}$.}

\begin{proof}
Assume $\htrue\in \H_{[n]}$ to be the true partial order on $[n]$ behind the observed linear extensions $\mathbf{y}_{1:N}$.
    We denote $o_i=\{o_{i,1},o_{i,2},\dots,o_{i,k_i}\}\subseteq [n]$ a subset of actors with size $|o_i|=k_i\leq n$ where $i$ can take values from $\{1,\dots, 2^n-1\}$. The suborders of $\htrue$ are $\htrue[o_i]$.

    Assume we observe $N$ linear extensions $\mathbf{y}_{1:N}$. Assume there are $N_i$ lists $(y_{i,1},\dots,y_{i,N_i})$ from each suborder $\htrue[o_i]$ where $y_{i,j}=(y_{i,j,1},\dots,y_{i,j,k_i})\in\P_{o_i}$ is a list, $j = 1,\dots,N_i$. Then the total number of observed lists is $N = \sum_{i=1}^{2^n-1}N_i$. For a general partial order $h\in\H_{[n]}$, and under the error-free observation model, we would have
    \begin{equation}
        y_{i,j}\sim P(y_{i,j}|h[o_i]) = \frac{1}{|\mathcal{L}[h[o_i]]|}\mathbbm{1}_{y_{i,j}\in\mathcal{L}[h[o_i]]}. 
    \end{equation}
    Denote by $I$ the indices of suborders appearing in $\mathbf{y}_{1:N}$,
    \begin{equation}
        Pr(\mathbf{y}_{1:N}|h) = \prod_{i\in I}\prod_{j = 1}^{N_i} P(y_{i,j}|h[o_i]). 
    \end{equation}
    Let $M_{i,r}=\sum_{j=1}^{N_i}\mathbbm{1}_{y_{i,j}=r}$ be the number of times a list $r\in\P_{o_i}$ appears in $\{y_{i,1},\dots,y_{i,N_i}\}$. 
    % Then,
    % \begin{equation}
    %     Pr(\mathbf{y}|h) = \prod_{i=1}^{2^n-1}\prod_{r\in\P_{o_i}} P(r|h[o_i])^{M_{i,r}}. 
    % \end{equation}
    If we are considering the posterior for $h$ then, since $\mathbf{y}_{1:N}$ would be linear extensions of $h$, we would have $y_{i,j}\in\mathcal{L}[h[o_i]], \forall i\in I,j\in[N_i]$, which gives $M_{i,r}=0$ if $r\notin\mathcal{L}[h[o_i]]$. The $(M_{i,r})_{r\in\mathcal{L}[h[o_i]]}$ values follow the multinomial distribution, such that,
    \begin{equation}\label{M-MN}
        (M_{i,r})_{r\in\mathcal{L}[h[o_i]]} \sim \text{Multinomial}(N_i, \frac{1}{|\mathcal{L}[h[o_i]]|},\dots,\frac{1}{|\mathcal{L}[h[o_i]]|}).
    \end{equation}
    This gives that for any $h$,
    \begin{equation}
        Pr(\mathbf{y}_{1:N}|h) = \prod_{i \in I}\prod_{r\in\mathcal{L}[h[o_i]]} P(r|h[o_i])^{M_{i,r}}. 
    \end{equation} 
    
     To show model consistency, we would like to show
    \begin{equation}\label{eq:EF-app-consistency}
        \lim_{N \to \infty} \frac{\pi(h|\mathbf{y}_{1:N})}{\pi(\htrue|\mathbf{y}_{1:N})} = \lim_{N \to \infty} \frac{Pr(\mathbf{y}_{1:N}|h)\pi(h)}{Pr(\mathbf{y}_{1:N}|\htrue)\pi(\htrue)} =
        \begin{cases}
        0 & \quad \text{when } h \neq \htrue\text{;}\\
        1 & \quad \text{when } h = \htrue\text{.} 
        \end{cases}
    \end{equation}
    
    Our prior $\pi(h)>0, \forall h\in\H_{[n]}$ in equation (\ref{prior-h}) as $K\ge \lfloor n/2\rfloor$ so the set $\{Z\in \mathbb{R}^{n\times K}: h(Z)=h\}$ has positive probability for every $h\in \H_{[n]}$. What remains to be shown is that, if $h\neq \htrue$, $\frac{Pr(\mathbf{y}_{1:N}|h)}{Pr(\mathbf{y}_{1:N}|\htrue)} \to 0$ when $N_i\to\infty$, $\forall i\in I$. Let $h[o_i]$ denote a suborder of $h$ corresponding to data lists $\mathbf{y}_i$. Consider 
    \begin{equation}
        \frac{Pr(\mathbf{y}_{1:N}|h)}{Pr(\mathbf{y}_{1:N}|\htrue)} = \prod_{i \in I}\frac{P(\mathbf{y}_i|h[o_i])}{P(\mathbf{y}_i|\htrue[o_i])} = \prod_{i \in I}R_i.
    \end{equation}
    Since $\htrue$ is the true partial order, we have $y_{k,j}\in\mathcal{L}[\htrue[o_k]]$. This gives 
    \begin{equation}
        R_i = \prod_{r\in\mathcal{L}[\htrue[o_i]]}\left(\frac{P(r|h[o_i])}{P(r|\htrue[o_i])} \right)^{M_{i,r}} = \prod_{r\in\mathcal{L}[\htrue[o_i]]}\left(\frac{|\mathcal{L}[\htrue[o_i]]|}{|\mathcal{L}[h[o_i]]|}\mathbbm{1}_{r\in\mathcal{L}[h[o_i]]} \right)^{M_{i,r}}.
    \end{equation}
    Given equation (\ref{M-MN}), and taking $h=\htrue$, we have $(\frac{M_{i,r}}{N_i})_{r\in\mathcal{L}[\htrue[o_i]]}\overset{a.s.}{\to}\frac{1}{|\mathcal{L}[\htrue[o_i]]|}$ as $N_i\to\infty$. This gives
    \begin{equation}
        R_i^{1/N_i}\overset{a.s.}{\to} \prod_{r\in\mathcal{L}[\htrue[o_i]]}\left(\frac{|\mathcal{L}[\htrue[o_i]]|}{|\mathcal{L}[h[o_i]]|}\mathbbm{1}_{r\in\mathcal{L}[h[o_i]]}\right)^{1/|\mathcal{L}[\htrue[o_i]]|}
    \end{equation}
    as $N_i\to\infty$. 
    For convenience, set $A = \mathcal{L}[\htrue[o_i]]$ and $B = \mathcal{L}[h[o_i]], i \in I$. 
    
    There are three different scenarios for the relation between $A$ and $B$.
    \begin{enumerate}
        \item If $A\setminus B\neq\emptyset$, then $\mathbbm{1}_{r\in B}=0$ for some $r\in A$. We have $R_i^{1/N_i}\overset{a.s.}{\to}0$. This gives
        \begin{equation}
            R_i\overset{a.s.}{\to} 0 \text{ as } N_i\to\infty \text{ when } A\setminus B\neq \emptyset. 
        \end{equation}
        \item If $A = B$, then $\mathcal{L}[\htrue[o_i]] =\mathcal{L}[h[o_i]]$ then $R_i$ is equal to 1 (and two partial orders with the same linear extensions are equal so $\htrue[o_i]=h[o_i]$).  
        \item If $A \subset B$ then $|A|<|B|$. In that case
        \begin{align}
            R_i^{1/N_i}\overset{a.s.}{\to} \prod_{r\in A}\left(\frac{|A|}{|B|}\right)^{1/|A|} = \left(\frac{|A|}{|B|}\right)^{|A|/|A|} = \frac{|A|}{|B|}<1,
        \end{align}
        so $R_i\overset{a.s.}{\to}0$ in this case.
        % This gives 
        % \begin{equation}
        %     R_i\overset{a.s.}{\to} \left(\frac{|A|}{|B|}\right)^{N_i} \text{ as } N_i\to\infty. 
        % \end{equation}
    \end{enumerate}
    From cases 1 and 3 we see that when $\htrue[o_i]\neq h[o_i]$, $R_i\overset{a.s.}{\to}0$ as $N_i\to\infty, i \in I$. Assembling these results,
    \begin{equation}
        \frac{\pi(h|\mathbf{y}_{1:N})}{\pi(\htrue|\mathbf{y}_{1:N})} = \frac{\pi(h)}{\pi(\htrue)}\prod_{i\in I} R_i \overset{N\to\infty}{\to}
        \begin{cases}
        0 & \quad \text{if } h[o_i] \neq \htrue[o_i]\text{ for any } i \in I\text{;}\\
        \frac{\pi(h)}{\pi(\htrue)} & \quad \text{if } h[o_i] = \htrue[o_i]\text{ for all } i \in I\text{.} 
        \end{cases}
    \end{equation}
    What's left to show is that if $h[o_i] = \htrue[o_i]$ for every $i\in I$ then we must have $h = \htrue$. This is where the last condition on $I$ comes in. We want to show
    \begin{equation}\label{NE-pos}
        \frac{\pi(h|\mathbf{y}_{1:N})}{\pi(\htrue|\mathbf{y}_{1:N})}\overset{N\to\infty}{\to}
        \begin{cases}
            0 & \quad \text{when } h \neq \htrue\text{;}\\
            1 & \quad \text{when } h = \htrue\text{,} 
        \end{cases}
    \end{equation}
    for any $h\in\H_{[n]}$ and the true partial order $\htrue\in\H_{[n]}$. This requires the following claim. 
    \begin{claim}
        \textbf{Condition (*)}: $\forall (i,j)\in[n]\times[n]$ with $i\neq j$, $\exists k\in I$ such that $(i,j)\in o_k$. Every pair in $[n]\times[n]$ appears together at least once in the sets for which we have rank information. 
        
        Equation (\ref{NE-pos}) holds for every $h, \htrue\in\H_{[n]}$ where $\htrue$ is the true partial order, if and only if condition (*) holds. 
    \end{claim}
    \begin{proof}
        If condition (*) holds, and $h[o_k] = \htrue[o_k]$ for every $k\in I$ then for every $(i,j)\in[n]\times[n]$, we have $\{i,j\}\subseteq o_k$ for some $k$. It follows that $h[o_k][\{i,j\}] = \htrue[o_k][\{i,j\}]$ so
        $h[\{i,j\}]=\htrue[\{i,j\}]\in\{i\succ j, i\prec j, i\|j\}$ (the suborder $h[o_k][\{i,j\}]$ of a suborder $h[o_k]$ is just the suborder $h[\{i,j\}]$ of the original partial order $h$). So $h=\htrue$.
        
        If condition (*) does not hold for some pair $(i,j)$, we give a counter-example constructed so that transitivity doesn't inform the missing relation. Suppose there is no $k\in I$ such that $\{i,j\}\subseteq o_k$ for some $\{i,j\}\subseteq [n]$. Take $h$ and $\htrue$ to be empty partial orders and let $h^{i\succ j}$, $h^{i\prec j}$ and $h^{i\|j}$ be three different partial orders obtained by adding order relations  $i\succ j$, $i\prec j$ and $i\|j$ to $h$. Now $h^{i\succ j}$, $h^{i\prec j}$ and $h^{i\|j}$ agree with $\htrue$ on all $o_k, k\in I$ so the posterior ratio converges to the prior ratio. Therefore, if condition (*) is not satisfied,  equation (\ref{NE-pos}) does not hold for every $\htrue$. 
    \end{proof}

    The fact that a condition as strong as Condition (*) is needed is surprising as we might hope transitivity would save us. Whilst this would be the case for any individual partial order, we need it to hold simultaneously for all partial orders, as $\htrue$ is unknown. 
\end{proof}

\section{Prior Properties: Proportion of Vertex-Series-Parallel Orders and Bucket Orders}\label{app:vspbo}

The proportion of vertex-series-parallel orders (VSPs) and bucket orders (BOs) under prior are shown in Table~\ref{tab:propn_vsp_bo}. 

\begin{table}[h!]
    \centering
    \begin{tabular}{c c c c c}
        \toprule
         & \multicolumn{2}{c}{Tied Prior} & \multicolumn{2}{c}{Non-Tied Prior}\\
         Number of Actors & \%VSP & \%BO & \%VSP & \%BO\\
         \midrule
         5  & 98.436\% & 79.187\% & 90.548\% & 64.787\% \\
         10 & 83.910\% & 50.833\% & 49.931\% & 31.803\%\\
         15 & 64.682\% & 37.063\% & 34.776\% & 21.909\%\\
         \bottomrule
    \end{tabular}
    \caption{The probability a random partial order under the tied prior is a VSP or BO. }
    \label{tab:propn_vsp_bo}
\end{table}

\section{An Alternative Tied Scenario}\label{appx:tie2}

Ties may be taken to impose an extra constraint in the observation model:  tied actors `must appear together'. Essentially, tied actors enter a list as if they were a single composite actor. In this setup tied actors are treated as a single node in a linear extension with an ``internal'' permutation over tied nodes taken uniformly at random. In \Fig~\ref{POTex}, instead of regarding $2\|_{h_0^*} 3$ and $2\|_{h_0^*} 4$, one may treat actors $(3,4)$ as one node. There will then only be $4$ linear extensions as are shown in \Fig~\ref{letienoisefree}. This construction does not seem relevant for our application, but may be of interest for further work.\\

\begin{minipage}{0.94\linewidth}
    \centering
    \begin{tikzpicture}[thick,scale=0.8, every node/.style={scale=0.8}]
        \node[draw, circle, minimum width=.1cm] (1) at (0, 2) {$1$};
        \node[draw, circle, minimum width=.1cm] (2) at (0, 1) {$2$};
        \node[draw=blue!60, rectangle, minimum width=1cm,minimum height=2cm] (6) at (0, -0.5) {};
        \node[draw, circle, minimum width=.1cm] (3) at (0, 0) {$3$};
        \node[draw, circle, minimum width=.1cm] (4) at (0, -1) {$4$};
        \node[draw, circle, minimum width=.1cm] (5) at (0, -2) {$5$};
        \draw[-latex] (1) -- (2);
        \draw[-latex] (2) -- (3);
        \draw[-latex] (3) -- (4);
        \draw[-latex] (4) -- (5);
    \end{tikzpicture}
    \begin{tikzpicture}[thick,scale=0.8, every node/.style={scale=0.8}]
        \node[draw, circle, minimum width=.1cm] (1) at (0, 2) {$1$};
        \node[draw, circle, minimum width=.1cm] (2) at (0, 1) {$2$};
        \node[draw=blue!60, rectangle, minimum width=1cm,minimum height=2cm] (6) at (0, -0.5) {};
        \node[draw, circle, minimum width=.1cm] (3) at (0, 0) {$4$};
        \node[draw, circle, minimum width=.1cm] (4) at (0, -1) {$3$};
        \node[draw, circle, minimum width=.1cm] (5) at (0, -2) {$5$};
        \draw[-latex] (1) -- (2);
        \draw[-latex] (2) -- (3);
        \draw[-latex] (3) -- (4);
        \draw[-latex] (4) -- (5);
    \end{tikzpicture}
    \begin{tikzpicture}[thick,scale=0.8, every node/.style={scale=0.8}]
        \node[draw, circle, minimum width=.1cm] (1) at (0, 2) {$1$};
        \node[draw, circle, minimum width=.1cm] (2) at (0, 1) {$3$};
        \node[draw, circle, minimum width=.1cm] (3) at (0, 0) {$4$};
        \node[draw=blue!60, rectangle, minimum width=1cm,minimum height=2cm] (6) at (0, 0.5) {};
        \node[draw, circle, minimum width=.1cm] (4) at (0, -1) {$2$};
        \node[draw, circle, minimum width=.1cm] (5) at (0, -2) {$5$};
        \draw[-latex] (1) -- (2);
        \draw[-latex] (2) -- (3);
        \draw[-latex] (3) -- (4);
        \draw[-latex] (4) -- (5);
    \end{tikzpicture}
    \begin{tikzpicture}[thick,scale=0.8, every node/.style={scale=0.8}]
        \node[draw, circle, minimum width=.1cm] (1) at (0, 2) {$1$};
        \node[draw=blue!60, rectangle, minimum width=1cm,minimum height=2cm] (6) at (0, 0.5) {};
        \node[draw, circle, minimum width=.1cm] (2) at (0, 1) {$4$};
        \node[draw, circle, minimum width=.1cm] (3) at (0, 0) {$3$};
        \node[draw, circle, minimum width=.1cm] (4) at (0, -1) {$2$};
        \node[draw, circle, minimum width=.1cm] (5) at (0, -2) {$5$};
        \draw[-latex] (1) -- (2);
        \draw[-latex] (2) -- (3);
        \draw[-latex] (3) -- (4);
        \draw[-latex] (4) -- (5);
    \end{tikzpicture}
    \captionof{figure}{All four linear extensions for the tied partial order in \Fig~\ref{POTex}.} \label{letienoisefree}                                                       
    \end{minipage}
    
\section{The Plackett-Luce Model}\label{app:PL}

This appendix complements a similar presentation introducing the Mallows noise model in Section~\ref{sec:mallows}.

The Plackett-Luce model, presented independently in \cite{plackett1975analysis} and \cite{luce1959possible}, is a widely used model for rank data, straightforwardly fitting rank-subset data.
The Plackett-Luce model enjoys extensive popularity because of its interpretability and tractability. It is projective so applies straightforwardly to either a complete ranking of all actors, a partial ranking of a subset of actors or the partial ranking of top actors \citep{bradley1952rank}. The literature extends the Plackett-Luce model to different inference schemes. For example, \cite{guiver2009bayesian} proposed an efficient Bayesian inference scheme for the Plackett-Luce model by applying expectation propagation and \cite{caron2012efficient} gives an efficient MCMC scheme targeting a class of Placket-Luce posterior distributions.

Let $\alpha=(\alpha_1,\dots,\alpha_n)$ be a set of actor weights. The Plackett-Luce model defines the probability for ranking $y\in\mathcal{P}_n$ as
\begin{equation*}
    p(y|\alpha) = \prod_{k=1}^n \frac{e^{\alpha_k}}{\sum_{m=k}^n e^{\alpha_m}}.
\end{equation*}
We can modify this usual setup for our purposes, as we can have {\em separate} actor weights for {\em each} list and integrate over these. This will seem counter-intuitive but is needed to preserve the structure of the underlying noise-free model. For $j=1,...,N$, let $\alpha_j=(\alpha_{j,1}, ..., \alpha_{{j,n}}) \in \mathbb{R}^n$ denote the (logit-scale) weights, so that $\alpha_{j,k}$ is the weight for actor $k=1,...,n$ in list $j$ (if they are present). The actor-ordered weights for actors $y_{j,1}, \dots, y_{j,n_j}$ in data-list $j=1,...,N$ are then $\alpha_{j,y_j} = \{\alpha_{j,y_{j,1}}, ..., \alpha_{j,y_{j,n_j}}\}$. The Plackett-Luce model assigns probability mass function
\begin{equation}\label{Plackett-Luce}
    p^{(PL)}(y_j|\mathbf{\alpha}_j) = \prod_{k = 1}^{n_j} \frac{e^{\alpha_{j,y_{j,k}}}}{\sum_{m = k}^{n_j} e^{\alpha_{j,y_{j,m}}}}.
\end{equation}

The Plackett-Luce model is marginally consistent. We can set the prior for $\alpha$ to be $\pi(\alpha|h,\sigma)=\prod_{j=1}^N \pi(\alpha_j|h,\sigma)$ where 
\begin{equation}\label{eq:pl-prior}
    \alpha_j|h,\sigma \propto \mathcal{MVN}(\alpha_j; 0_K,\sigma^2 \mathbbm{I}_n)\mathbbm{1}_{\{R(\alpha_j)\in\mathcal{L}[h]\}}.
\end{equation}
If we denote the rank vector for $\alpha_j$ as $R(\alpha_j)$ where $R(\alpha_j)_i = \sum_{k=1}^n \mathbbm{1}_{\{\alpha_{j,k}\leq \alpha_{j,i}\}}$, $R(\alpha_j)\sim \mathcal{U}(\mathcal{L}[h])$ under this prior distribution. In addition, we show that the normalising constant in equation \ref{eq:pl-prior} is $|\mathcal{L}[h]|/ n!$. There has been literature that considers the Plackett-Luce model on tied ranking data \citep{pl20} and tied actors \citep{henderson2022modelling,10.1093/imaman/dpaa027}.

Though theoretically appealing, the setup requires $\alpha_j \in \mathbb{R}^n$ for $j=1,\dots,N$, which is a large space to explore. We implemented sudo-marginal method to effectively integrate over $\alpha$, where the approximated list of linear extensions are generated via the Bubley-Dyer algorithm \citep{bubley1999faster}. However, the $(PDP,\rho)\backslash M$ model (as discussed in section \ref{sec:mallows}) still shows a large computational advantage in comparison and seems to show no disadvantages as a model for the data. In addition, \cite{liu2019model} shows that the Mallow's model is preferred over Plackett-Luce model. The Plackett-Luce model produces a less certain concensus estimation (higher MSE) in their potato ranking example. A similar result is obtained in our model comparison between the Plackett-Luce mixture and Mallows mixture models in table \ref{table:model-comparison}. We therefore choose the Mallow's noise observation model instead of Plackett-Luce in this paper, but we include the Plackett-Luce Mixture for model comparison in section \ref{sec:application}. 
\pagebreak

\section{MCMC Algorithm}\label{MCMC}

This section summarises the MCMC algorithm used to sample posterior distributions to the parameters of interest. We implemented the partial order with ties model with two noisy observation models - the Mallow's observation model (M) and the queue-jumping model (Q). 

\begin{algorithm}[H]
\caption{MCMC for the Partial Order with Ties Model}
    \SetKwInOut{Input}{input}
    \SetKwFunction{noconflict}{noconflict}
    \SetKwFunction{potie}{potie}
    \SetKwFunction{ncol}{ncol}
    \SetKwFunction{condP}{condP}
    \SetKwProg{Fn}{Function}{:}{}

    \Input{$Z, \rho, S, \theta, p$}
    
    \Fn{\potie{$Z$, $S$}}{
        h =  $\mathbf{0}_{n\times n}$;\\ 
        K = \ncol($Z$);\\
        \ForEach{$i,j\in[n]\times[n], i\neq j$}{
            \If{$Z_{i,m} > Z_{j,m}\, \forall m\in [K]$}{
                $h_{i,j} = 1$; 
            }
        }
        \ForEach{$c\in[|S|]$}{
            $h_{i,j}=0 \,\forall i,j\in S_c$; 
        }
    \KwRet $h$
    }
    
    \hrulefill\emph{Update for $K$ \& $Z$}\hrulefill
    
    Sample $K^*\in\{K-1, K+1\}$ with $\rho_{K,K+1}=\rho_{K,K-1}=0.5$\;
    
    \If{$K^*>0$}{
    \eIf{$K^*=K+1$}{
        \ForEach{$c \in \{1,\dots,|S|\}$} 
            {Set $Z^*_{\cdot,1:K}=Z$ \;
            
            Sample $Z^*_{c,K^*}$ such that $Z^*_{c,\cdot}\sim MVN(\mathbf{0}_{K^*},\Sigma_{\rho})$}
        }{Set $Z^* = Z_{\cdot,1:K^*}$}

        $h^* \leftarrow$ \potie{$Z^*$, $S$}\;
        
        \eIf{model=M}{
            $\phi_K = \frac{\pi(K^*)p^{(Q)}(\mathbf{y}|h(Z^*),\theta)}{\pi(K)p^{(M)}(\mathbf{y}|h(Z),\theta)}$
            }{
            $\phi_K = \frac{\pi(K^*)p^{(Q)}(\mathbf{y}|h(Z^*),p)}{\pi(K)p^{(M)}(\mathbf{y}|h(Z^*),p)}$
        }
        Sample $u_K \sim \mathcal{U}(0,1)$\;
        
        \If{$\phi_K > u_K$}{
            $Z \leftarrow Z^*$; $K \leftarrow K^*$; $h \leftarrow h^*$\;
        }
    }

\end{algorithm}

\newpage

\begin{algorithm}[H]
\caption{MCMC for the Partial Order with Ties Model (part 2)}
    \SetKwInOut{Input}{input}
    \SetKwFunction{noconflict}{noconflict}
    \SetKwFunction{potie}{potie}
    \SetKwFunction{condPL}{condPL}
    \SetKwFunction{condP}{condP}
    \SetKwProg{Fn}{Function}{:}{}

        \hrulefill\emph{Update for $S$ \& $Z$}\hrulefill
    
    Sample $j \in [n]$\;
    
    Set $Z^* \leftarrow Z$;
    
    Simulate $Z^*_{|S^{-j}|+1,\cdot}\sim MVN(\mathbf{0}_K,\Sigma_{\rho})$ $\forall i\in \{1,\dots,N\}$ \;
    
    \eIf{model=M}{
        
        Sample a new cluster $c$ for $j\in\{1,2,\dots,|S^{-j}|+2\}$ with probability proportional to 
        {\scriptsize
        \begin{align*}
            & ((n_1^{-j}-\phi_\alpha)p^{(M)}(\mathbf{y}|h(Z_{S^{-j}\cup \{j\in S_1^{-j}\}}),\theta), \dots, (n^{-j}_{|S^{-j}|}-\phi_\alpha)p^{(M)}(\mathbf{y}|h(Z_{S^{-j}\cup \{j\in S_{|S_{-j}|}^{-j}\}}),\theta),\\ 
            & (\eta_b + \eta_a |S^{-j}|) p^{(M)}(\mathbf{y}|h(Z^*_{S^{-j}\cup \{j\in S_{|S^{-j}|+1}^{-j}\}}),\theta), (\eta_b + \eta_a |S^{-j}| p^{(M)}(\mathbf{y}|h(Z),\theta));
        \end{align*}
        }
    }{
    
        Sample a new cluster $c$ for $j\in\{1,2,\dots,|S^{-j}|+2\}$ with probability proportional to
        {\scriptsize
        \begin{align*}
            & ((n_1^{-j}-\phi_\alpha)p^{(Q)}(\mathbf{y}|h(Z_{S^{-j}\cup \{j\in S_1^{-j}\}},p), \dots, (n^{-j}_{|S^{-j}|}-\phi_\alpha)p^{(Q)}(\mathbf{y}|h(Z_{S^{-j}\cup \{j\in S_{|S^{-j}|}^{-j}\}}),p),\\
            & (\eta_b + \eta_a |S^{-j}|) p^{(Q)}(\mathbf{y}|h(Z^*_{S^{-j}\cup \{j\in S_{|S^{-j}|+1}^{-j}\}}),p), (\eta_b + \eta_a |S^{-j}|) p^{(Q)}(\mathbf{y}|h(Z),p))
        \end{align*}
        }
    }
    \If{$c=|S^{-j}|+1$}{$S \leftarrow S^{-j}\cup \{j\in S_{|S^{-j}|}^{-j}\};Z\leftarrow Z^*; \alpha \leftarrow \alpha^*$}
    \If {$c \in \{1,\dots,|S^{-j}|\}$}{$S \leftarrow S^{-j}\cup \{j\in S_c^{-j}\};Z\leftarrow Z_{S^{-j}\cup \{j\in S_c^{-j}\}}$}
    
    \hrulefill    \emph{Update for $Z$}\hrulefill
    
    Set $Z^* \leftarrow Z$;
    
    Sample $r \in \{1, \dots, |S|\}$, $c \in \{1, \dots, K\}$ and simulate $Z^* _{r,c}$ such that $Z^*_{r,\cdot}\sim MVN(Z_{r,\cdot},\Sigma_\rho)$;
    
    $h^* \leftarrow$\potie{$Z^*,S$};
    
    \eIf{model = M}{
        $\phi_Z = \frac{\pi(Z^*|\rho,K,S)p^{(M)}(\mathbf{y}|h(Z^*),\theta)}{\pi(Z|\rho,K,S)p^{(M)}(\mathbf{y}|h(Z),\theta)}$
    }{
        $\phi_Z = \frac{p^{(Q)}(\mathbf{y}|h(Z^*),p)\pi(Z^*|\rho,K,S)}{p^{(Q)}(\mathbf{y}|h(Z),p)\pi(Z|\rho,K,S)}$
    }
    Sample $u_Z \sim \mathcal{U}(0, 1)$\;
    \If{$\phi_Z > u_Z$}{$Z\leftarrow Z^*; h \leftarrow h^*$.}
    
\end{algorithm}
    
\newpage

\begin{algorithm}[H]
\caption{MCMC for the Partial Order with Ties Model (part 2)}
    \SetKwInOut{Input}{input}
    \SetKwFunction{noconflict}{noconflict}
    \SetKwFunction{potie}{potie}
    \SetKwFunction{condPL}{condPL}
    \SetKwFunction{condP}{condP}
    \SetKwProg{Fn}{Function}{:}{}

    \hrulefill    \emph{Update for $\rho$}\hrulefill
    
    Sample $\delta_\rho \sim \mathcal{U}(w_\rho, \frac{1}{w_\rho})$;\Comment{$w_\rho \in (0,1)$ is a constant}
    
    Set $\rho^* = 1-\delta_\rho(1-\rho)$ then 
    $\phi_\rho = \frac{\pi(\rho^*)\pi(Z|\rho^*,K,S)}{\pi(\rho)\pi(Z|\rho,K,S)\delta_\rho};$
    
    Sample $u_\rho \sim \mathcal{U}(0,1)$;
    
    \If{$u_\rho < \phi_\rho \; \&\; \rho^* < 1$}{$\rho \leftarrow \rho^*$.}
    
    \hrulefill    \emph{Update for $p$}\hrulefill
    
    \If{model=Q}{

    Let $r = log(\frac{p}{1-p})$ and sample $r^* \sim \mathcal{N}(r,1)$;

    Set $p^* \leftarrow \frac{1}{1+e^{-r^*}}$
    
        $\phi_p = \frac{\pi(p^*)p^{(Q)}(\mathbf{y}|h(Z),p^*)}{\pi(p)p^{(Q)}(\mathbf{y}|h(Z),p)};$
        
        Sample $u_p \sim \mathcal{U}(0,1)$; 
        
        \If{$u_p < \phi_p$}{$p \leftarrow p^*$.}
        }
    
    \hrulefill    \emph{Update for $\theta$}\hrulefill
    
    \If{model = M}{
    Sample $\theta^* \sim \mathcal{N}(\theta, 0.5)$
    
    \If{$\theta^*>0$}{
        $\phi_\theta = \frac{p^{(M)}(\mathbf{y}|h(Z),\theta^*)\pi_\theta(\theta^*)}{p^{(M)}(\mathbf{y}|h(Z),\theta)\pi_\theta(\theta)}$; 
        
        Sample $u_\theta \sim \mathcal{U}(0,1)$; 
        
        \If{$u_\theta < \phi_\theta$}{$\theta \leftarrow \theta^*$.}
        }
    }\end{algorithm}

\newpage

\section{Application - the `Royal Acta' (Bishops) Data }

\subsection{The Data Lists}\label{app:lists}

Witness lists between 1131 and 1133: 
\begin{multicols}{2}
    \begin{minipage}{\linewidth}
        [1] 3 4 8

        [2] 1 13 12 14  7 11  6

        [3] 1 13  7 14 10

        [4] 14 11

        [5] 1 13  7

        [6] 1 13  7  8 15  9

        [7] 3 8 4 2

        [8] 13  1  5

        [9] 13  1  5

        [10] 13  7  5

        [11] 13  1  5

        [12] 1 13  5

        [13] 1  7 13 14  5 10 15  8

        [14] 15  8

        [15] 12  7  9

        [16] 12 13  7 10  9 11  6

        [17] 3 4 2

        [18] 4 3

        [19] 7 13 12  1

        [20] 1 13  7  5 14  9

        [21] 6 1 7
    \end{minipage}%
    \begin{minipage}{\linewidth}
        1: \textit{Roger, Bishop of Salisbury}\\
        2: \textit{Richard, Bishop of Bayeux}\\
        3: \textit{John, Bishop of Lisieux}\\
        4: \textit{Ouen, Bishop of Evreux}\\
        5: \textit{Bernard, Bishop of St David's}\\
        6: \textit{Everard, bishop of Norwich}\\
        7: \textit{Alexander, Bishop of Lincoln}\\
        8: \textit{John, Bishop of Sees}\\
        9: \textit{Seffrid, Bishop of Chichester}\\
        10: \textit{John, Bishop of Rochester}\\
        11: \textit{Simon, Bishop of Worcester}\\
        12: \textit{Gilbert, the Universal, bishop of London}\\
        13: \textit{Henry, de Blois, Bishop of Winchester, 1129-1171}\\
        14: \textit{Robert, de Bethune, Bishop of Hereford}\\
        15: \textit{Algar, Bishop of Coutances}\\
    \end{minipage}
\end{multicols}

Witness lists between 1100 and 1103: 

\begin{multicols}{2}
    \begin{minipage}{\linewidth}
        [1] 2 5 7

        [2] 2 8 5 3 1 6 4 7

        [3] 2 5 6

        [4] 2 9

        [5] 5 8

        [6] 5 9

        [7] 1 2 3 5 4 6 7

        [8] 5 8

        [9] 1 5

        [10] 8 7

        [11] 8 7

        [12] 2 4 1 3 8

        [13] 2 1
    \end{minipage}%
    \begin{minipage}{\linewidth}
        1: \textit{Gundulf, bishop of Rochester}\\
        2: \textit{Maurice, bishop of London}\\
        3: \textit{Robert, de Limesey, bishop of Chester}\\
        4: \textit{Ralph, bishop of Chichester}\\
        5: \textit{Robert, Bloet, bishop of Lincoln}\\
        6: \textit{Samson, bishop of Worcester}\\
        7: \textit{Ranulf, Flambard, bishop of Durham}\\
        8: \textit{William, Giffard, bishop of Winchester, 1100-1129}\\
        9: \textit{Roger, Bishop of Salisbury}\\
    \end{minipage}
\end{multicols}

\subsection{Additional Results}\label{app:results}

Four experiments are conducted on the `Royal Acta' (Bishop) 1131-1133 and 1100-1103 data lists for both the $(PDP,\rho)\backslash M$ and $(PDP,\rho)\backslash Q$ models. Each MCMC chain was run for $1e5$ iterations. We record every $2n$ steps and chose a burn-in period of $500(\times 2n)$ for all chains. The list of experiments and effective sample sizes to each key parameter is shown in table \ref{tab:ess}. 

\begin{table}[h!]
    \centering
    \begin{tabular}{c c c c c c}
        \toprule
        \multicolumn{2}{c}{} & \multicolumn{4}{c}{Effective Sample Sizes (ESSs)} \\
        \midrule
        Time Period  & Model & $K$ & $\rho$ & $\theta$ & $p$ \\
        \midrule
        \multirow{2}{4em}{1131-1133} & $(PDP,\rho)\backslash M$ & 240.60 & 136.32 & 1318.60 & - \\
        & $(PDP,\rho)\backslash Q$ & 104.32 & 174.88 & - & 873.53 \\
        \midrule
        \multirow{2}{4em}{1100-1103} & $(PDP,\rho)\backslash M$ & 605.91 & 340.02 & 899.20 & - \\
        & $(PDP,\rho)\backslash Q$ & 444.18 & 358.95 & - & 1917.23 \\ 
        \bottomrule
    \end{tabular}
    \caption{The effective sample sizes (ESSs) to the key parameters for each experiment. }
    \label{tab:ess}
\end{table}

Section \ref{app:pos} and section \ref{app:trace} present the posterior distributions and trace plots for these key parameters respectively. 

\subsubsection{Posterior Distributions}\label{app:pos}

This section displays some additional result on prior and posterior distributions. Fig.~\ref{fig:depths} gives both the prior and posterior distributions on the number of clusters (on actors) and partial order depths. 

\begin{figure}[h!]
  \centering
  \includegraphics[width=0.65\linewidth]{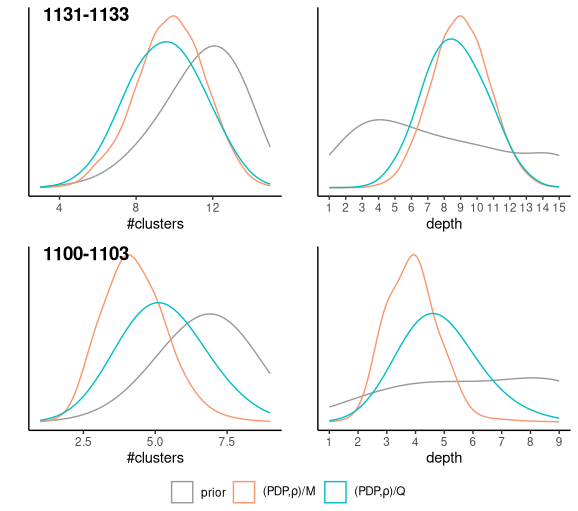}
  \caption{The prior/posterior distributions for the number of clusters (left) and partial order depths (right) for the $(PDP,\rho)\backslash M$ (pink) and $(PDP,\rho)\backslash Q$ (blue) model on periods 1131-1133 (upper) and 1100-1103 (lower). The prior distributions are colored in gray. }
  \label{fig:depths}
\end{figure} 

We further present cluster samples from the MCMC for the $(PDP,\rho)\backslash Q$ model in the heatmap in fig \ref{fig:heatmap-qj}. The heatmap for $(PDP,\rho)\backslash M$ is in figure \ref{fig:heatmap-m}.

\begin{figure}[h!]
  \centering
  \includegraphics[width=0.7\linewidth]{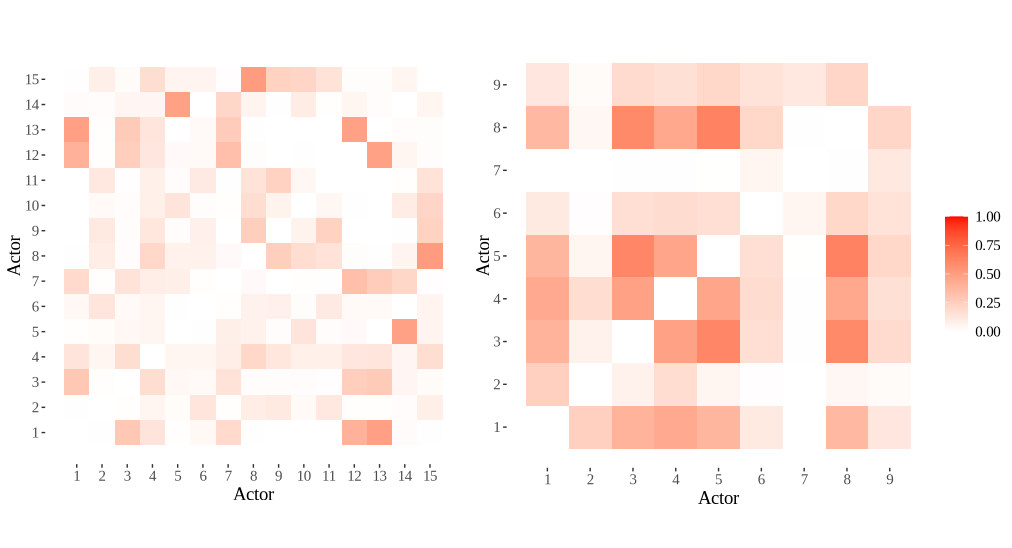}
  \caption{The ‘heatmaps’ for the estimated clusters for the $(PDP,\rho)\backslash Q$ model for time periods 1131-1133 (left) and 1100-1103 (right). }
  \label{fig:heatmap-qj}
\end{figure} 

\pagebreak
\subsubsection{Key Parameter Trace Plots}\label{app:trace}

The trace plots to the key parameters are shown in figure \ref{fig:trace-m} for the $(PDP,\rho)\backslash M$ model and figure \ref{fig:trace-qj} for the $(PDP,\rho)\backslash Q$ model. All trace plots display great mixing and convergence. 

\begin{figure}[h!]
\centering
\begin{minipage}{.4\textwidth}
  \centering
  \includegraphics[width=\linewidth]{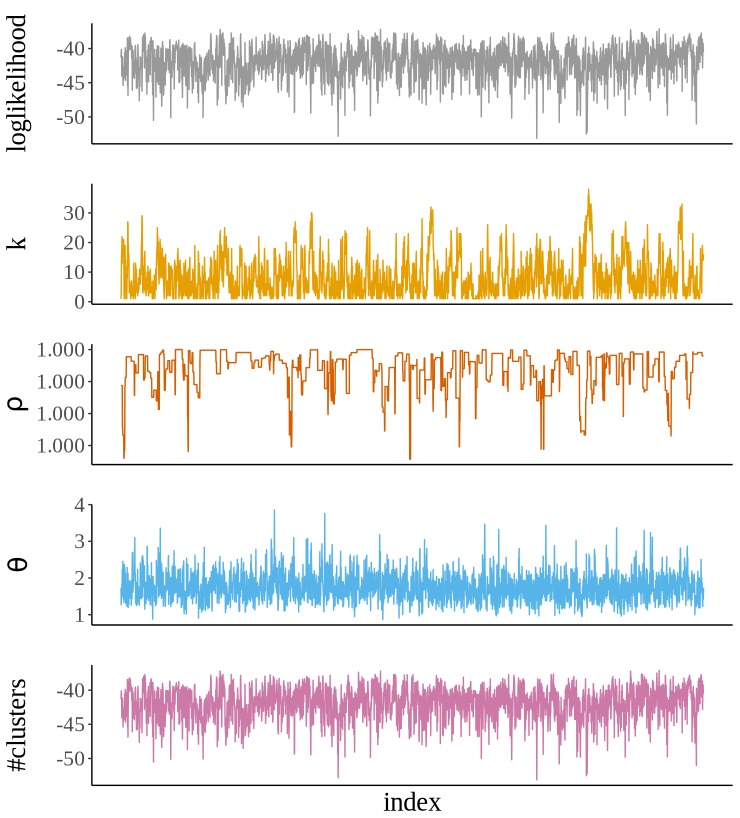}
\end{minipage}%
\begin{minipage}{.4\textwidth}
  \centering
  \includegraphics[width=\linewidth]{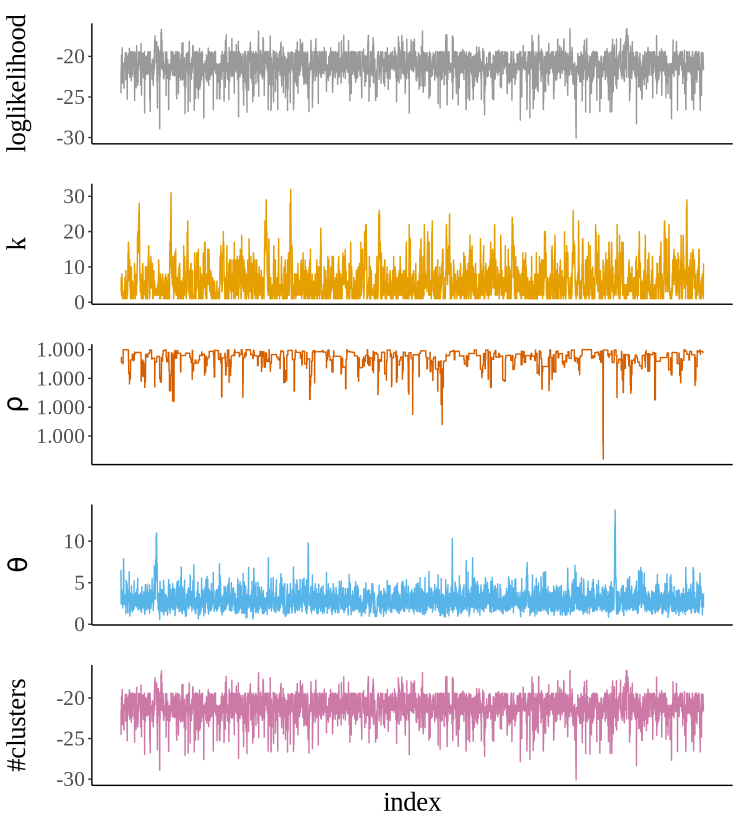}
\end{minipage}
\captionof{figure}{The trace plots for the $(PDP,\rho)\backslash M$ model on the 1131-1133 (left) and 1100-1103 (right) `royal acta' (bishop) witness list data. We present the trace plot for the log-likelihood (grey), the number of clusters (pink) and key parameters (latent matrix column dimension parameter $K$ - orange, the depth control parameter $\rho$ - red and the dispersion parameter for the Mallow's observation model $\theta$ - blue). The burn-in periods are removed from the trace plots. }\label{fig:trace-m}
\end{figure}

\begin{figure}[h!]
\centering
\begin{minipage}{.4\textwidth}
  \centering
  \includegraphics[width=\linewidth]{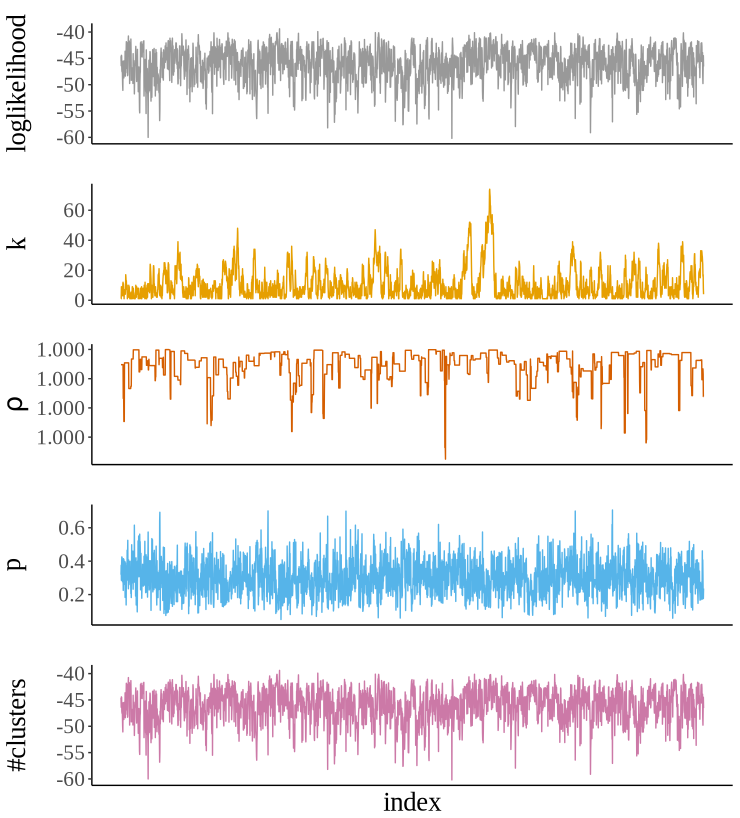}
\end{minipage}%
\begin{minipage}{.4\textwidth}
  \centering
  \includegraphics[width=\linewidth]{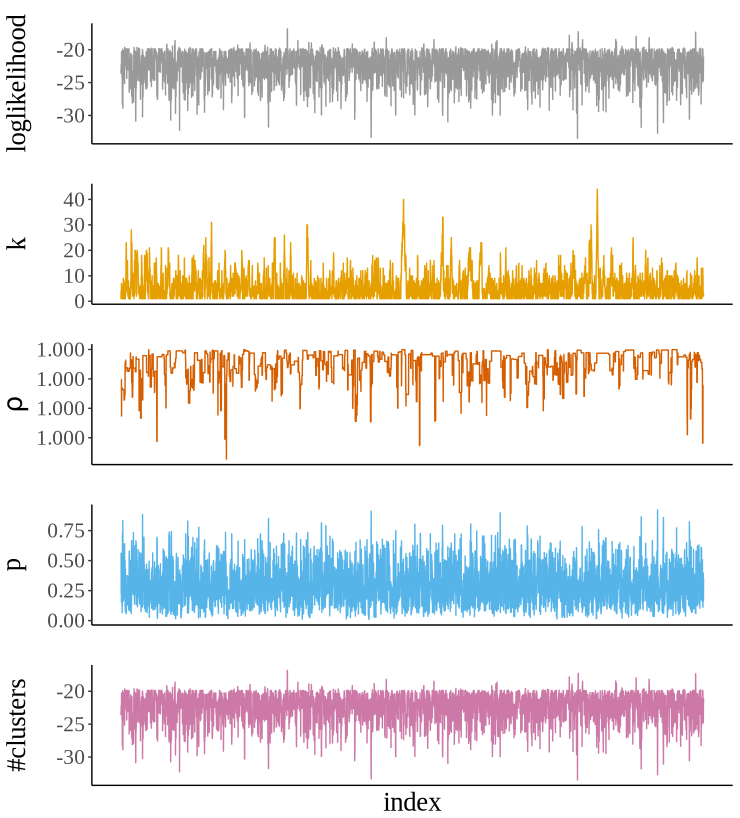}
\end{minipage}
\captionof{figure}{The trace plots for the $(PDP,\rho)\backslash Q$ model on the 1131-1133 (left) and 1100-1103 (right) `royal acta' (bishop) witness list data. We present the trace plot for the log-likelihood (grey), the number of clusters (pink) and key parameters (latent matrix column dimension parameter $K$ - orange, the depth control parameter $\rho$ - red and the error probability in the queue-jumping observation observation model $p$ - blue).}\label{fig:trace-qj}
\end{figure}

\pagebreak

\subsection{Sensitivity Analysis on $K$ Prior}\label{app:sensitivity}

This section investigates the models' ($(PDP,\rho)\backslash M$ and $(PDP,\rho)\backslash Q$) sensitivity to the priors on the latent matrix column dimension parameter $K$. According to \cite{watt2015inference}, $k=n/2$ is sufficient to recover any unknown partial order. We explicit a geometric prior distribution on $K$ as $k\sim Geo(p_k)$, where $p_k$ is a pre-determined hyperparameter - such that the means of $K$ are in $\{n,n/2,n/4\}$. The priors on other parameters are adjusted accordingly to obtain a relatively flat depth distribution. Table \ref{tab:sensitivity} records the prior distributions on $K$ for different sensitivity analysis on the real data. 
\begin{table}[h]
    \centering
    \resizebox{\textwidth}{!}{%
    \begin{tabular}{c c c c c}
    \toprule
        Time Period & Model & $K$-Prior & PDP-Prior & $\rho$-Prior \\
        \midrule
        \multirow{6}{4em}{1131-1133 ($n=15$)} & $(PDP,\rho)\backslash M$ & $Geo(0.06) (\bar{k}=n)$ & $PDP(\eta_b=3, \eta_a = 0.7)$ & $Beta(1,1/6)$\\
        & $(PDP,\rho)\backslash Q$ & $Geo(0.06) (\bar{k}=n)$ & $PDP(\eta_b=3, \eta_a = 0.7)$ & $Beta(1,1/6)$ \\
        & $(PDP,\rho)\backslash M$ & $Geo(0.12) (\bar{k}=n/2)$ & $PDP(\eta_b=6, \eta_a = 0.6)$ & $Beta(1,1/4)$\\
        & $(PDP,\rho)\backslash Q$ & $Geo(0.12) (\bar{k}=n/2)$ & $PDP(\eta_b=6, \eta_a = 0.6)$ & $Beta(1,1/4)$ \\
        & $(PDP,\rho)\backslash M$ & $Geo(0.21) (\bar{k}=n/4)$ & $PDP(\eta_b=12, \eta_a = 0.55)$ & $Beta(1,1/4)$\\
        & $(PDP,\rho)\backslash Q$ & $Geo(0.21) (\bar{k}=n/4)$ & $PDP(\eta_b=12, \eta_a = 0.55)$ & $Beta(1,1/4)$ \\
        \midrule
        \multirow{6}{4em}{1100-1103 ($n=9$)} & $(PDP,\rho)\backslash M$ & $Geo(0.1) (\bar{k}=n)$ & $PDP(\eta_b=5, \eta_a = 0.4)$ & $Beta(1,1/4)$\\
        & $(PDP,\rho)\backslash Q$ & $Geo(0.1) (\bar{k}=n)$ & $PDP(\eta_b=5, \eta_a = 0.4)$ & $Beta(1,1/4)$\\
        & $(PDP,\rho)\backslash M$ & $Geo(0.18) (\bar{k}=n/2)$ & $PDP(\eta_b=5, \eta_a = 0.4)$ & $Beta(1,1/3)$\\
        & $(PDP,\rho)\backslash Q$ & $Geo(0.18) (\bar{k}=n/2)$ & $PDP(\eta_b=5, \eta_a = 0.4)$ & $Beta(1,1/3)$\\
        & $(PDP,\rho)\backslash M$ & $Geo(0.31) (\bar{k}=n/4)$ & $PDP(\eta_b=5, \eta_a = 0.4)$ & $Beta(1,1/3)$\\
        & $(PDP,\rho)\backslash Q$ & $Geo(0.31) (\bar{k}=n/4)$ & $PDP(\eta_b=5, \eta_a = 0.4)$ & $Beta(1,1/3)$\\
        \bottomrule
    \end{tabular}}
    \caption{Sensitivity analysis on $K$-priors. }
    \label{tab:sensitivity}
\end{table}

We display the prior and posterior distributions for $K$ in \Fig~\ref{fig:k-sensitivity-m} and \Fig~\ref{fig:k-sensitivity-qj} for $(PDP,\rho)\backslash M$ and $(PDP,\rho)\backslash Q$ models respectively. When $\Bar{k}\geq n/2$, the posterior distributions are close to the prior distributions in both periods - indicating their high sensitivity to the prior choice. However, the posterior distributions on $K$ are updated to higher $K$ values for both the $(PDP,\rho)\backslash M$ and $(PDP,\rho)\backslash Q$ models in period 1131-1133. No similar behaviour is observed in period 1100-1103, which is partially a result of data limitation. 

Given the result by \cite{watt2015inference}, a $k\geq n/2$ is sufficient for the latent matrix model to recover any unknown partial order. Having a large mass on the $k\geq n/2$ values in posterior will give the model enough freedom to explore the space of partial order without much constraints. We therefore conclude a choice of $K$-prior with $\Bar{K}=n/2$ is the most optimal and present the corresponding result in section \ref{sec:real-app}. 

\begin{figure}[h!]
  \centering
  \includegraphics[width=0.9\linewidth]{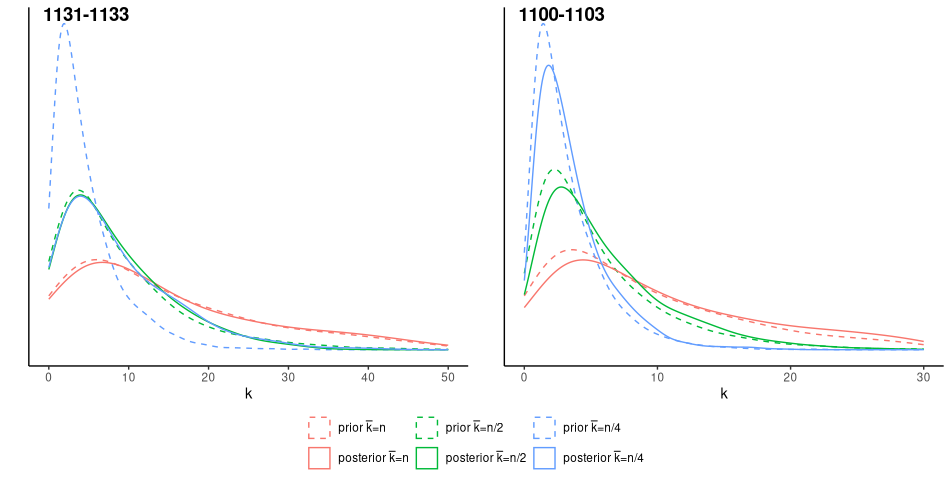}
  \caption{The prior (dashed) and posterior (solid) distributions for $K$ from the $(PDP,\rho)\backslash M$ model for sensitivity analysis on $K$-prior. We explicit the $K$ priors to have means of $n$ (red), $n/2$ (green) or $n/4$ (blue) (details in table \ref{tab:sensitivity}). The result for time period 1131-1133 is displayed on left, and 1100-1103 on right.}
  \label{fig:k-sensitivity-m}
\end{figure} 

\begin{figure}[h!]
  \centering
  \includegraphics[width=0.9\linewidth]{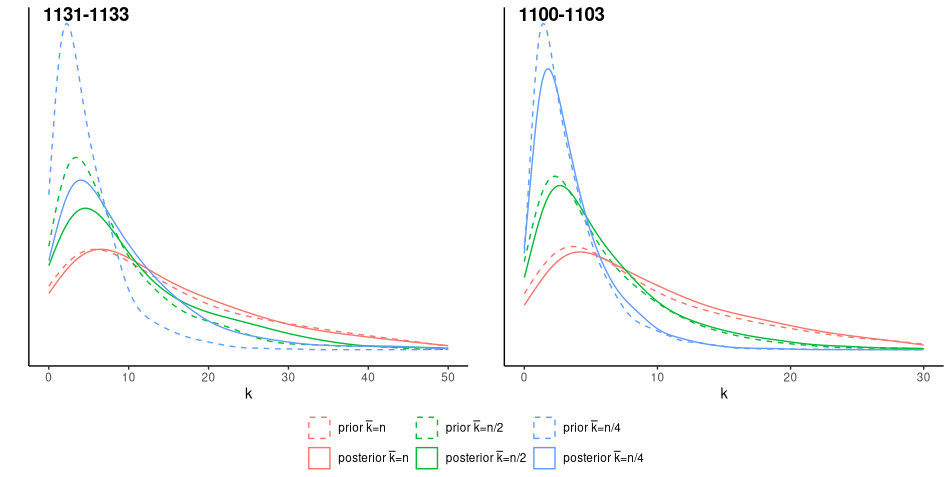}
  \caption{The prior (dashed) and posterior (solid) distributions for $K$ from the $(PDP,\rho)\backslash Q$ model for sensitivity analysis on $K$-prior. We explicit the $K$ priors to have means of $n$ (red), $n/2$ (green) or $n/4$ (blue) (details in table \ref{tab:sensitivity}). The result for time period 1131-1133 is displayed on left, and 1100-1103 on right.}
  \label{fig:k-sensitivity-qj}
\end{figure} 

\pagebreak

\subsection{Synthetic Reconstruction Test}\label{app:reconstruction}

The synthetic reconstruction test detailed in section \ref{sec:reconstruction} is designed to test the model's ability to reconstruct the `true' partial order given data-lists under different observation models. We consider four scenarios for each time period. 

\begin{itemize}
    \item Simulation 1: error-free, $y_i\sim\mathcal{U}(\L[h^{(T)}[y^{obs}_i]]), \forall i\in [N]$.
    \item Simulation 2: data-list with random error. For $i\in [N]$, we simulate $y_i\sim \mathcal{U}(\L[h^{(T)}[y^{obs}_i]])$ from the noise-free model; we select a pair of actors $a,b\in y_i$ uniform at random, $a\neq b$, and put them in the same order as they appear in the data. If $a\prec b$ in $y_i$ but $a\succ b$ in $y^{obs}_i$ then exchange the positions of $a$ and $b$ in $y_i$ leaving all else unchanged.
    \item Simulation 3: data-list with Mallow's error given $\theta^*$, $$y_i\sim p^{(M)}(\cdot|h^{(T)}[o_i],\theta),\forall i\in[N].$$ 
    \item Simulation 4: data-list with queue-jumping error given $p^*$, $$y_i\sim p^{(Q)}(\cdot|h^{(T)}[o_i],p),\forall i\in[N].$$
\end{itemize}

The simulations are generated based on the true partial order in Figure~\ref{fig:concensus-syn-mallows-3133} for period 1131-1133 and the true partial order in Figure~\ref{fig:concensus-syn-mallows-0003} for period 1100-1103. For the 1131-1133 structured data, we show the concensus orders $h^{con}(\epsilon=0.2)$ from the $(PDP,\rho)\backslash M$ model in \Fig~\ref{fig:concensus-syn-3133-m} and the $(PDP,\rho)\backslash Q$ model in \Fig~\ref{fig:concensus-syn-3133-q}. For the 1100-1103 structured data, we show the concensus orders $h^{con}(\epsilon=0.2)$ from the $(PDP,\rho)\backslash M$ model in \Fig~\ref{fig:concensus-syn-0003-m} and the $(PDP,\rho)\backslash Q$ model in \Fig~\ref{fig:concensus-syn-0003-q}. \\

\begin{minipage}{\linewidth}
    \centering
    \begin{tikzpicture}[thick,scale=1, every node/.style={scale=0.6}]
        \node[draw, circle, minimum width=.75cm] (102) at (8,-.2) {$2$};
        \node[draw, circle, minimum width=.75cm] (104) at (7.3,1) {$4$};
        \node[draw, circle, minimum width=.75cm] (103) at (8.7,.4) {$3$};
        \node[draw, circle, minimum width=.75cm] (101) at (6, 1) {$1$}; 
        \node[draw, circle, minimum width=.75cm] (112) at (8.7, 1) {$12$}; 
        \node[draw, circle, minimum width=.75cm] (113) at (10, 1) {$13$}; 
        \node[draw, circle, minimum width=.75cm] (105) at (8, -1) {$5$}; 
        \node[draw, circle, minimum width=.75cm] (107) at (7, -.6) {$7$};
        \node[draw, circle, minimum width=.75cm] (114) at (9, -.6) {$14$}; 
        \node[draw, circle, minimum width=.75cm] (110) at (8, -1.8) {$10$}; 
        \node[draw, circle, minimum width=.75cm] (108) at (7, -2.6) {$8$}; 
        \node[draw, circle, minimum width=.75cm] (115) at (9, -2.6) {$15$}; 
        \node[draw, circle, minimum width=.75cm] (109) at (7, -3.4) {$9$}; 
        \node[draw, circle, minimum width=.75cm] (111) at (9, -3.4) {$11$}; 
        \node[draw, circle, minimum width=.75cm] (106) at (8, -4.2) {$6$}; 

        \draw[-latex] (101) -- (102);
        \draw[-latex] (104) -- (102);
        \draw[-latex] (113) -- (102);
        \draw[-latex] (112) -- (103);
        \draw[-latex] (103) -- (102);
        \draw[-latex] (102) -- (107);
        \draw[-latex] (102) -- (114);
        \draw[-latex] (107) -- (105);
        \draw[-latex] (114) -- (105);
        \draw[-latex] (105) -- (110);
        \draw[-latex] (110) -- (108);
        \draw[-latex] (110) -- (115);
        \draw[-latex] (108) -- (111);
        \draw[-latex] (108) -- (109);
        \draw[-latex] (115) -- (111);
        \draw[-latex] (115) -- (109);
        \draw[-latex] (109) -- (106);
        \draw[-latex] (111) -- (106); 
    \end{tikzpicture}
        \captionof{figure}{The true partial order for list simulation (1131-1133 structured data). }\label{fig:concensus-syn-mallows-3133}
\end{minipage}

{
\fbox{\resizebox{\dimexpr\linewidth-2\fboxrule-2\fboxsep}{!}{
\begin{minipage}[t]{.2\linewidth}
\centering
\begin{tikzpicture}[thick,scale=.6,every node/.style={scale=.4}]
    \node[draw,circle,minimum width=.75cm] (13) at (0,1.6) {$13$};
    \node[draw,circle,minimum width=.75cm] (3) at (-2,1.6) {$3$};
    \node[draw,circle,minimum width=.75cm] (7) at (0,0) {$7$};
    \node[draw,circle,minimum width=.75cm] (1) at (0,0.8) {$1$};
    \node[draw,circle,minimum width=.75cm] (12) at (2,0.8) {$12$};
    \node[draw,circle,minimum width=.75cm] (2) at (-2,0.8) {$2$};
    \node[draw,circle,minimum width=.75cm] (14) at (2,0) {$14$};
    \node[draw,circle,minimum width=.75cm] (5) at (0,-0.8) {$5$};
    \node[draw,circle,minimum width=.75cm] (10) at (0,-1.6) {$10$};
    \node[draw,circle,minimum width=.75cm] (8) at (0,-2.4) {$8$};
    \node[draw,circle,minimum width=.75cm] (15) at (-1,-3) {$15$};
    \node[draw,circle,minimum width=.75cm] (11) at (1,-3) {$11$};
    \node[draw,circle,minimum width=.75cm] (4) at (1,2.2) {$4$};
    \node[draw,circle,minimum width=.75cm] (9) at (-1,-3.8) {$9$};
    \node[draw,circle,minimum width=.75cm] (6) at (-1,-4.6) {$6$};

    \draw[-latex] (4) -- (3);
    \draw[-latex] (4) -- (13);
    \draw[-latex] (4) -- (12);
    \draw[-latex] (3) -- (2);
    \draw[-latex] (13) -- (1);
    \draw[-latex] (12) -- (14);
    \draw[-latex] (3) -- (1);
    \draw[-latex] (1) -- (14);
    \draw[-latex] (2) -- (14);
    \draw[-latex] (2) -- (7);
    \draw[-latex] (1) -- (7);
    \draw[-latex] (12) -- (7);
    \draw[-latex] (7) -- (5);
    \draw[-latex] (14) -- (5);
    \draw[-latex] (5) -- (10);
    \draw[-latex] (10) -- (8);
    \draw[-latex] (8) -- (15);
    \draw[-latex] (8) -- (11); 
    \draw[-latex] (15) -- (9); 
    \draw[-latex] (9) -- (6); 
\end{tikzpicture}
{\tiny $(PDP,\rho)\backslash M-\text{LE}$}
\end{minipage}%
\hspace{.5cm}
\begin{minipage}[t]{.2\linewidth}
\centering
\begin{tikzpicture}[thick,scale=.6,every node/.style={scale=.4}]
    \node[draw,circle,minimum width=.75cm] (13) at (0,1.6) {$13$};
    \node[draw,circle,minimum width=.75cm] (12) at (0,2.4) {$12$};
    \node[draw,circle,minimum width=.75cm] (1) at (0,0) {$1$};
    \node[draw,circle,minimum width=.75cm] (4) at (0,0.8) {$4$};
    \node[draw,circle,minimum width=.75cm] (3) at (2,0) {$3$};
    \node[draw,circle,minimum width=.75cm] (15) at (2,-3.6) {$15$};
    \node[draw,circle,minimum width=.75cm] (14) at (2,-0.8) {$14$};
    \node[draw,circle,minimum width=.75cm] (7) at (0,-0.8) {$7$};
    \node[draw,circle,minimum width=.75cm] (5) at (0,-1.6) {$5$};
    \node[draw,circle,minimum width=.75cm] (10) at (0,-2.4) {$10$};
    \node[draw,circle,minimum width=.75cm] (8) at (0,-3.2) {$8$};
    \node[draw,circle,minimum width=.75cm] (11) at (0,-4) {$11$};
    \node[draw,circle,minimum width=.75cm] (9) at (3,-5) {$9$};
    \node[draw,circle,minimum width=.75cm] (2) at (2,-4.4) {$2$};
    \node[draw,circle,minimum width=.75cm] (6) at (1,-5) {$6$};

    \draw[-latex] (12) -- (13);
    \draw[-latex] (13) -- (4);
    \draw[-latex] (4) -- (1);
    \draw[-latex] (1) -- (7);
    \draw[-latex] (7) -- (5);
    \draw[-latex] (5) -- (10);
    \draw[-latex] (8) -- (11);
    \draw[-latex] (10) -- (8);
    \draw[-latex] (8) -- (15);
    \draw[-latex] (15) -- (2);
    \draw[-latex] (2) -- (9);
    \draw[-latex] (2) -- (6);
    \draw[-latex] (4) -- (3);
    \draw[-latex] (1) -- (14);
    \draw[-latex] (3) -- (7);
    \draw[-latex] (14) -- (5);
    \draw[-latex] (3) -- (14);
\end{tikzpicture}
{\tiny $(PDP,\rho)\backslash M-\text{ERROR}$}
\end{minipage}%
\begin{minipage}[t]{.2\linewidth}
\centering
\begin{tikzpicture}[thick,scale=.6,every node/.style={scale=.4}]
    \node[draw,circle,minimum width=.75cm] (13) at (2,1.1) {$13$};
    \node[draw,circle,minimum width=.75cm] (12) at (1,1.1) {$12$};
    \node[draw,circle,minimum width=.75cm] (1) at (-2,1.1) {$1$};
    \node[draw,circle,minimum width=.75cm] (4) at (0,1.1) {$4$};
    \node[draw,circle,minimum width=.75cm] (3) at (-1,1.1) {$3$};
    \node[draw,circle,minimum width=.75cm] (15) at (1,-2.2) {$15$};
    \node[draw,circle,minimum width=.75cm] (14) at (1.5,-1.1) {$14$};
    \node[draw,circle,minimum width=.75cm] (7) at (-0.5,-1.1) {$7$};
    \node[draw,circle,minimum width=.75cm] (5) at (-1.5,-1.1) {$5$};
    \node[draw,circle,minimum width=.75cm] (10) at (.5,-1.1) {$10$};
    \node[draw,circle,minimum width=.75cm] (8) at (0,-2.2) {$8$};
    \node[draw,circle,minimum width=.75cm] (11) at (.5,-4.4) {$11$};
    \node[draw,circle,minimum width=.75cm] (9) at (.5,-3.3) {$9$};
    \node[draw,circle,minimum width=.75cm] (2) at (0,0) {$2$};
    \node[draw,circle,minimum width=.75cm] (6) at (.5,-5.5) {$6$};

    \draw[-latex] (1) -- (2);
    \draw[-latex] (3) -- (2);
    \draw[-latex] (4) -- (2);
    \draw[-latex] (12) -- (2);
    \draw[-latex] (13) -- (2);
    \draw[-latex] (2) -- (5);
    \draw[-latex] (2) -- (7);
    \draw[-latex] (2) -- (10);
    \draw[-latex] (2) -- (14);
    \draw[-latex] (10) -- (8);
    \draw[-latex] (10) -- (15);
    \draw[-latex] (15) -- (9);
    \draw[-latex] (8) -- (9);
    \draw[-latex] (9) -- (11);
    \draw[-latex] (11) -- (6);
\end{tikzpicture}
{\tiny $(PDP,\rho)\backslash M-\text{MALLOW'S ERROR}$}
\end{minipage}%
\hspace{.5cm}
\begin{minipage}[t]{.24\linewidth}
\centering
\begin{tikzpicture}[thick,scale=.6,every node/.style={scale=.4}]
    \node[draw,circle,minimum width=.75cm] (13) at (1,2.2) {$13$};
    \node[draw,circle,minimum width=.75cm] (12) at (1,1.1) {$12$};
    \node[draw,circle,minimum width=.75cm] (1) at (-1,2.2) {$1$};
    \node[draw,circle,minimum width=.75cm] (4) at (0,2.2) {$4$};
    \node[draw,circle,minimum width=.75cm] (3) at (0,1.1) {$3$};
    \node[draw,circle,minimum width=.75cm] (15) at (.5,-1.1) {$15$};
    \node[draw,circle,minimum width=.75cm] (14) at (2.5,-1.1) {$14$};
    \node[draw,circle,minimum width=.75cm] (7) at (-1.5,-1.1) {$7$};
    \node[draw,circle,minimum width=.75cm] (5) at (-2.5,-1.1) {$5$};
    \node[draw,circle,minimum width=.75cm] (10) at (1.5,-1.1) {$10$};
    \node[draw,circle,minimum width=.75cm] (8) at (-.5,-1.1) {$8$};
    \node[draw,circle,minimum width=.75cm] (11) at (0,-3.3) {$11$};
    \node[draw,circle,minimum width=.75cm] (9) at (0,-2.2) {$9$};
    \node[draw,circle,minimum width=.75cm] (2) at (0,0) {$2$};
    \node[draw,circle,minimum width=.75cm] (6) at (0,-4.4) {$6$};

    \draw[-latex] (1) -- (2);
    \draw[-latex] (3) -- (2);
    \draw[-latex] (4) -- (3);
    \draw[-latex] (12) -- (2);
    \draw[-latex] (13) -- (12);
    \draw[-latex] (2) -- (5);
    \draw[-latex] (2) -- (7);
    \draw[-latex] (2) -- (10);
    \draw[-latex] (2) -- (14);
    \draw[-latex] (2) -- (8);
    \draw[-latex] (2) -- (15);
    \draw[-latex] (15) -- (9);
    \draw[-latex] (8) -- (9);
    \draw[-latex] (9) -- (11);
    \draw[-latex] (11) -- (6);
\end{tikzpicture}
{\tiny $(PDP,\rho)\backslash M-\text{QJ ERROR}$}
\end{minipage}%
}}
\captionof{figure}{The concensus orders for the $(PDP,\rho)\backslash M$ model on the synthetic data with 1131-1133 structure (simulation 1-4 from left to right). We conclude an edge if such order relation has more than 0.2 posterior probability (inferred from section \ref{sec:reconstruction}). An edge is colored red if it has more than 0.9 posterior probability. }\label{fig:concensus-syn-3133-m}
}
\vspace*{0.2in}
{
\fbox{\resizebox{\dimexpr\linewidth-2\fboxrule-2\fboxsep}{!}{
\begin{minipage}[t]{.2\linewidth}
\centering
\begin{tikzpicture}[thick,scale=.6,every node/.style={scale=.4}]
    \node[draw,circle,minimum width=.75cm] (13) at (1.5,1.4) {$13$};
    \node[draw,circle,minimum width=.75cm] (3) at (.5,1.4) {$3$};
    \node[draw,circle,minimum width=.75cm] (7) at (-.5,.6) {$7$};
    \node[draw,circle,minimum width=.75cm] (1) at (-.5,1.4) {$1$};
    \node[draw,circle,minimum width=.75cm] (12) at (-1.5,1.4) {$12$};
    \node[draw,circle,minimum width=.75cm] (2) at (0,-0.8) {$2$};
    \node[draw,circle,minimum width=.75cm] (14) at (.5,.6) {$14$};
    \node[draw,circle,minimum width=.75cm] (5) at (0,0) {$5$};
    \node[draw,circle,minimum width=.75cm] (10) at (0,-1.6) {$10$};
    \node[draw,circle,minimum width=.75cm] (8) at (0,-2.4) {$8$};
    \node[draw,circle,minimum width=.75cm] (15) at (-1,-3) {$15$};
    \node[draw,circle,minimum width=.75cm] (11) at (1,-3) {$11$};
    \node[draw,circle,minimum width=.75cm] (4) at (0,2) {$4$};
    \node[draw,circle,minimum width=.75cm] (9) at (-1,-3.8) {$9$};
    \node[draw,circle,minimum width=.75cm] (6) at (-1,-4.6) {$6$};

    \draw[-latex] (4) -- (3);
    \draw[-latex] (4) -- (13);
    \draw[-latex] (4) -- (12);
    \draw[-latex] (3) -- (14);
    \draw[-latex] (4) -- (1);
    \draw[-latex] (1) -- (14);
    \draw[-latex] (5) -- (2);
    \draw[-latex] (1) -- (14);
    \draw[-latex] (2) -- (10);
    \draw[-latex] (3) -- (7);
    \draw[-latex] (1) -- (7);
    \draw[-latex] (7) -- (5);
    \draw[-latex] (14) -- (5);
    \draw[-latex] (10) -- (8);
    \draw[-latex] (8) -- (15);
    \draw[-latex] (8) -- (11); 
    \draw[-latex] (15) -- (9); 
    \draw[-latex] (9) -- (6); 
\end{tikzpicture}
{\tiny $(PDP,\rho)\backslash Q-\text{LE}$}
\end{minipage}%
\begin{minipage}[t]{.2\linewidth}
\centering
\begin{tikzpicture}[thick,scale=.6,every node/.style={scale=.4}]
    \node[draw,circle,minimum width=.75cm] (13) at (0,1.6) {$13$};
    \node[draw,circle,minimum width=.75cm] (12) at (0,2.4) {$12$};
    \node[draw,circle,minimum width=.75cm] (4) at (0,0) {$4$};
    \node[draw,circle,minimum width=.75cm] (1) at (0,0.8) {$1$};
    \node[draw,circle,minimum width=.75cm] (14) at (1,-1.2) {$14$};
    \node[draw,circle,minimum width=.75cm] (15) at (0,-3.2) {$15$};
    \node[draw,circle,minimum width=.75cm] (5) at (-1,-1.2) {$5$};
    \node[draw,circle,minimum width=.75cm] (3) at (0,-0.8) {$3$};
    \node[draw,circle,minimum width=.75cm] (7) at (0,-1.6) {$7$};
    \node[draw,circle,minimum width=.75cm] (10) at (-1,-2) {$10$};
    \node[draw,circle,minimum width=.75cm] (8) at (-2,-3.2) {$8$};
    \node[draw,circle,minimum width=.75cm] (11) at (-1,-2.8) {$11$};
    \node[draw,circle,minimum width=.75cm] (9) at (-1,-3.6) {$9$};
    \node[draw,circle,minimum width=.75cm] (2) at (-1,-4.4) {$2$};
    \node[draw,circle,minimum width=.75cm] (6) at (-1,-5.2) {$6$};

    \draw[-latex] (12) -- (13);
    \draw[-latex] (13) -- (1);
    \draw[-latex] (1) -- (4);
    \draw[-latex] (4) -- (3);
    \draw[-latex] (3) -- (7);
    \draw[-latex] (3) -- (14);
    \draw[-latex] (3) -- (5);
    \draw[-latex] (5) -- (10);
    \draw[-latex] (10) -- (11);
    \draw[-latex] (11) -- (8);
    \draw[-latex] (11) -- (15);
    \draw[-latex] (8) -- (9);
    \draw[-latex] (15) -- (9);
    \draw[-latex] (9) -- (2);
    \draw[-latex] (2) -- (6);
    \draw[-latex] (3) -- (14);
\end{tikzpicture}
{\tiny $(PDP,\rho)\backslash Q-\text{ERROR}$}
\end{minipage}%
\begin{minipage}[t]{.2\linewidth}
\centering
\begin{tikzpicture}[thick,scale=.6,every node/.style={scale=.4}]
    \node[draw,circle,minimum width=.75cm] (13) at (2,1) {$13$};
    \node[draw,circle,minimum width=.75cm] (12) at (1,1) {$12$};
    \node[draw,circle,minimum width=.75cm] (1) at (-2,1) {$1$};
    \node[draw,circle,minimum width=.75cm] (4) at (0,1) {$4$};
    \node[draw,circle,minimum width=.75cm] (3) at (-1,1) {$3$};
    \node[draw,circle,minimum width=.75cm] (15) at (-.5,-3) {$15$};
    \node[draw,circle,minimum width=.75cm] (14) at (1,-1) {$14$};
    \node[draw,circle,minimum width=.75cm] (7) at (0,-1) {$7$};
    \node[draw,circle,minimum width=.75cm] (5) at (-1,-1) {$5$};
    \node[draw,circle,minimum width=.75cm] (10) at (-1,-2) {$10$};
    \node[draw,circle,minimum width=.75cm] (8) at (-1.5,-3) {$8$};
    \node[draw,circle,minimum width=.75cm] (11) at (-1,-5) {$11$};
    \node[draw,circle,minimum width=.75cm] (9) at (-1,-4) {$9$};
    \node[draw,circle,minimum width=.75cm] (2) at (0,0) {$2$};
    \node[draw,circle,minimum width=.75cm] (6) at (-1,-6) {$6$};

    \draw[-latex] (1) -- (2);
    \draw[-latex] (3) -- (2);
    \draw[-latex] (4) -- (2);
    \draw[-latex] (12) -- (2);
    \draw[-latex] (13) -- (2);
    \draw[-latex] (2) -- (5);
    \draw[-latex] (2) -- (7);
    \draw[-latex] (5) -- (10);
    \draw[-latex] (2) -- (14);
    \draw[-latex] (10) -- (8);
    \draw[-latex] (10) -- (15);
    \draw[-latex] (15) -- (9);
    \draw[-latex] (8) -- (9);
    \draw[-latex] (9) -- (11);
    \draw[-latex] (11) -- (6);
\end{tikzpicture}
{\tiny $(PDP,\rho)\backslash Q-\text{MALLOW'S ERROR}$}
\end{minipage}%
\begin{minipage}[t]{.24\linewidth}
\centering
\begin{tikzpicture}[thick,scale=.6,every node/.style={scale=.4}]
    \node[draw,circle,minimum width=.75cm] (13) at (1,2) {$13$};
    \node[draw,circle,minimum width=.75cm] (12) at (1,1) {$12$};
    \node[draw,circle,minimum width=.75cm] (1) at (2,1) {$1$};
    \node[draw,circle,minimum width=.75cm] (4) at (0,2) {$4$};
    \node[draw,circle,minimum width=.75cm] (3) at (-1,1) {$3$};
    \node[draw,circle,minimum width=.75cm] (15) at (2,-2) {$15$};
    \node[draw,circle,minimum width=.75cm] (14) at (0,-1) {$14$};
    \node[draw,circle,minimum width=.75cm] (7) at (-1.5,-1) {$7$};
    \node[draw,circle,minimum width=.75cm] (5) at (3,-1) {$5$};
    \node[draw,circle,minimum width=.75cm] (10) at (1.5,-1) {$10$};
    \node[draw,circle,minimum width=.75cm] (8) at (1,-2) {$8$};
    \node[draw,circle,minimum width=.75cm] (11) at (1.5,-4) {$11$};
    \node[draw,circle,minimum width=.75cm] (9) at (1.5,-3) {$9$};
    \node[draw,circle,minimum width=.75cm] (2) at (0,0) {$2$};
    \node[draw,circle,minimum width=.75cm] (6) at (1.5,-5) {$6$};

    \draw[-latex] (1) -- (2);
    \draw[-latex] (3) -- (2);
    \draw[-latex] (4) -- (3);
    \draw[-latex] (4) -- (12);
    \draw[-latex] (12) -- (2);
    \draw[-latex] (13) -- (12);
    \draw[-latex] (13) -- (1);
    \draw[-latex] (2) -- (5);
    \draw[-latex] (2) -- (7);
    \draw[-latex] (2) -- (10);
    \draw[-latex] (2) -- (14);
    \draw[-latex] (10) -- (8);
    \draw[-latex] (10) -- (15);
    \draw[-latex] (15) -- (9);
    \draw[-latex] (8) -- (9);
    \draw[-latex] (9) -- (11);
    \draw[-latex] (11) -- (6);
\end{tikzpicture}
{\tiny $(PDP,\rho)\backslash Q-\text{QJ ERROR}$}
\end{minipage}%
}}
\captionof{figure}{The concensus orders for the $(PDP,\rho)\backslash Q$ model on the synthetic data with 1131-1133 structure (simulation 1-4 from left to right). We conclude an edge if such order relation has more than 0.2 posterior probability (inferred from section \ref{sec:reconstruction}). An edge is colored red if it has more than 0.9 posterior probability. }\label{fig:concensus-syn-3133-q}
}

\begin{minipage}{\linewidth}
    \centering
    \begin{tikzpicture}[thick,scale=1, every node/.style={scale=0.8}]

        \node[draw, circle, minimum width=.75cm] (102) at (8,1) {$2$};
        \node[draw, circle, minimum width=.75cm] (104) at (8,0) {$4$};
        \node[draw, circle, minimum width=.75cm] (103) at (7,0) {$3$};
        \node[draw, circle, minimum width=.75cm] (101) at (6, 0) {$1$}; 
        \node[draw, circle, minimum width=.75cm] (105) at (9, 0) {$5$}; 
        \node[draw, circle, minimum width=.75cm] (107) at (7.5,-2.5) {$7$};
        \node[draw, circle, minimum width=.75cm] (108) at (10, 0) {$8$};  
        \node[draw, circle, minimum width=.75cm] (109) at (8.5, -1) {$9$}; 
        \node[draw, circle, minimum width=.75cm] (106) at (7.5, -1.7) {$6$}; 

        \draw[-latex] (102) -- (101);
        \draw[-latex] (102) -- (103);
        \draw[-latex] (102) -- (104);
        \draw[-latex] (102) -- (105);
        \draw[-latex] (102) -- (108);
        \draw[-latex] (103) -- (109);
        \draw[-latex] (104) -- (109);
        \draw[-latex] (105) -- (109);
        \draw[-latex] (108) -- (109);
        \draw[-latex] (101) -- (106);
        \draw[-latex] (109) -- (106);
        \draw[-latex] (106) -- (107);

    \end{tikzpicture}
        \captionof{figure}{The true partial order for list simulation (1131-1133 structured data). }\label{fig:concensus-syn-mallows-0003}
        \vspace*{0.2in}
\end{minipage}

{
\fbox{\resizebox{\dimexpr\linewidth-2\fboxrule-2\fboxsep}{!}{
\begin{minipage}[t]{.2\linewidth}
\centering
\begin{tikzpicture}[thick,scale=.6,every node/.style={scale=.4}]
        \node[draw, circle, minimum width=.75cm] (102) at (8,1) {$2$};
        \node[draw, circle, minimum width=.75cm] (104) at (8,0) {$4$};
        \node[draw, circle, minimum width=.75cm] (103) at (7,0) {$3$};
        \node[draw, circle, minimum width=.75cm] (101) at (6, 0) {$1$}; 
        \node[draw, circle, minimum width=.75cm] (105) at (9, 0) {$5$}; 
        \node[draw, circle, minimum width=.75cm] (107) at (7.5,-2) {$7$};
        \node[draw, circle, minimum width=.75cm] (108) at (10, 0) {$8$};  
        \node[draw, circle, minimum width=.75cm] (109) at (8.5, -2) {$9$}; 
        \node[draw, circle, minimum width=.75cm] (106) at (8, -1) {$6$}; 

        \draw[-latex,red] (102) -- (101);
        \draw[-latex,red] (102) -- (103);
        \draw[-latex,red] (102) -- (104);
        \draw[-latex,red] (102) -- (105);
        \draw[-latex,red] (102) -- (108);
        \draw[-latex] (103) -- (106);
        \draw[-latex] (104) -- (106);
        \draw[-latex] (105) -- (106);
        \draw[-latex] (108) -- (106);
        \draw[-latex] (101) -- (106);
        \draw[-latex] (106) -- (109);
        \draw[-latex] (106) -- (107);
\end{tikzpicture}
{\tiny $(PDP,\rho)\backslash M-\text{LE}$}
\end{minipage}%
\hspace{.5cm}
\begin{minipage}[t]{.2\linewidth}
\centering
\begin{tikzpicture}[thick,scale=.6,every node/.style={scale=.4}]
        \node[draw, circle, minimum width=.75cm] (102) at (8,1) {$2$};
        \node[draw, circle, minimum width=.75cm] (104) at (8,0) {$4$};
        \node[draw, circle, minimum width=.75cm] (103) at (7,0) {$3$};
        \node[draw, circle, minimum width=.75cm] (101) at (6, 0) {$1$}; 
        \node[draw, circle, minimum width=.75cm] (105) at (9, 0) {$5$}; 
        \node[draw, circle, minimum width=.75cm] (107) at (7.5,-2) {$7$};
        \node[draw, circle, minimum width=.75cm] (108) at (10, 0) {$8$};  
        \node[draw, circle, minimum width=.75cm] (109) at (8.5, -2) {$9$}; 
        \node[draw, circle, minimum width=.75cm] (106) at (8, -1) {$6$}; 

        \draw[-latex,red] (102) -- (101);
        \draw[-latex,red] (102) -- (103);
        \draw[-latex,red] (102) -- (104);
        \draw[-latex,red] (102) -- (105);
        \draw[-latex,red] (102) -- (108);
        \draw[-latex] (103) -- (106);
        \draw[-latex] (104) -- (106);
        \draw[-latex] (105) -- (106);
        \draw[-latex] (108) -- (106);
        \draw[-latex] (101) -- (106);
        \draw[-latex] (106) -- (109);
        \draw[-latex] (106) -- (107);
\end{tikzpicture}
{\tiny $(PDP,\rho)\backslash M-\text{ERROR}$}
\end{minipage}%
\hspace{.5cm}
\begin{minipage}[t]{.2\linewidth}
\centering
\begin{tikzpicture}[thick,scale=.6,every node/.style={scale=.4}]
        \node[draw, circle, minimum width=.75cm] (102) at (8.5,1.6) {$2$};
        \node[draw, circle, minimum width=.75cm] (104) at (8,.8) {$4$};
        \node[draw, circle, minimum width=.75cm] (103) at (7,.8) {$3$};
        \node[draw, circle, minimum width=.75cm] (101) at (8.5, 0) {$1$}; 
        \node[draw, circle, minimum width=.75cm] (105) at (9, .8) {$5$}; 
        \node[draw, circle, minimum width=.75cm] (107) at (8,-1.6) {$7$};
        \node[draw, circle, minimum width=.75cm] (108) at (10, .8) {$8$};  
        \node[draw, circle, minimum width=.75cm] (109) at (9, -1.6) {$9$}; 
        \node[draw, circle, minimum width=.75cm] (106) at (8.5, -.8) {$6$}; 

        \draw[-latex,red] (102) -- (103);
        \draw[-latex,red] (102) -- (104);
        \draw[-latex,red] (102) -- (105);
        \draw[-latex,red] (102) -- (108);
        \draw[-latex] (103) -- (101);
        \draw[-latex] (104) -- (101);
        \draw[-latex] (105) -- (101);
        \draw[-latex] (108) -- (101);
        \draw[-latex] (101) -- (106);
        \draw[-latex] (106) -- (109);
        \draw[-latex] (106) -- (107);
\end{tikzpicture}
{\tiny $(PDP,\rho)\backslash M-\text{MALLOW'S ERROR}$}
\end{minipage}%
\begin{minipage}[t]{.24\linewidth}
\centering
\begin{tikzpicture}[thick,scale=.6,every node/.style={scale=.4}]
        \node[draw, circle, minimum width=.75cm] (102) at (0,1.6) {$2$};
        \node[draw, circle, minimum width=.75cm] (104) at (1,0) {$4$};
        \node[draw, circle, minimum width=.75cm] (103) at (0,0) {$3$};
        \node[draw, circle, minimum width=.75cm] (101) at (-1, .8) {$1$}; 
        \node[draw, circle, minimum width=.75cm] (105) at (2, 0) {$5$}; 
        \node[draw, circle, minimum width=.75cm] (107) at (1,-1.6) {$7$};
        \node[draw, circle, minimum width=.75cm] (108) at (1, .8) {$8$};  
        \node[draw, circle, minimum width=.75cm] (109) at (.5, -.8) {$9$}; 
        \node[draw, circle, minimum width=.75cm] (106) at (0, -1.6) {$6$}; 

        \draw[-latex] (102) -- (101);
        \draw[-latex] (102) -- (108);
        \draw[-latex] (108) -- (103);
        \draw[-latex] (108) -- (104);
        \draw[-latex] (108) -- (105);
        \draw[-latex] (105) -- (109);
        \draw[-latex] (104) -- (109);
        \draw[-latex] (103) -- (109);
        \draw[-latex] (101) -- (109);
        \draw[-latex] (109) -- (106);
        \draw[-latex] (109) -- (107);
\end{tikzpicture}
{\tiny $(PDP,\rho)\backslash M-\text{QJ ERROR}$}
\end{minipage}%
}}
\captionof{figure}{The concensus orders for the $(PDP,\rho)\backslash M$ model on the synthetic data with 1100-1103 structure (simulation 1-4 from left to right). We conclude an edge if such order relation has more than 0.2 posterior probability (inferred from section \ref{sec:reconstruction}). An edge is colored red if it has more than 0.9 posterior probability. }\label{fig:concensus-syn-0003-m}
}
\vspace*{0.2in}
{
\fbox{\resizebox{\dimexpr\linewidth-2\fboxrule-2\fboxsep}{!}{
\begin{minipage}[t]{.2\linewidth}
\centering
\begin{tikzpicture}[thick,scale=.6,every node/.style={scale=.4}]
        \node[draw, circle, minimum width=.75cm] (2) at (1,1.4) {$2$};
        \node[draw, circle, minimum width=.75cm] (3) at (1,0.6) {$3$};
        \node[draw, circle, minimum width=.75cm] (8) at (0,0) {$8$};
        \node[draw, circle, minimum width=.75cm] (1) at (-2, 0) {$1$}; 
        \node[draw, circle, minimum width=.75cm] (5) at (0,-.8) {$5$}; 
        \node[draw, circle, minimum width=.75cm] (6) at (0,-1.6) {$6$};
        \node[draw, circle, minimum width=.75cm] (4) at (2, -.8) {$4$};  
        \node[draw, circle, minimum width=.75cm] (7) at (-1, -2.2) {$7$}; 
        \node[draw, circle, minimum width=.75cm] (9) at (1, -2.2) {$9$}; 

        \draw[-latex] (2) -- (3);
        \draw[-latex] (3) -- (4);
        \draw[-latex] (3) -- (1);
        \draw[-latex] (3) -- (8);
        \draw[-latex] (3) -- (4);
        \draw[-latex] (8) -- (5);
        \draw[-latex] (5) -- (6);
        \draw[-latex] (1) -- (6);
        \draw[-latex] (6) -- (7);
        \draw[-latex] (6) -- (9);
        \draw[-latex] (4) -- (6);
\end{tikzpicture}
{\tiny $(PDP,\rho)\backslash Q-\text{LE}$}
\end{minipage}%
\hspace{.5cm}
\begin{minipage}[t]{.2\linewidth}
\centering
\begin{tikzpicture}[thick,scale=.6,every node/.style={scale=.4}]
        \node[draw, circle, minimum width=.75cm] (2) at (1,2) {$2$};
        \node[draw, circle, minimum width=.75cm] (3) at (1,1) {$3$};
        \node[draw, circle, minimum width=.75cm] (8) at (2,0) {$8$};
        \node[draw, circle, minimum width=.75cm] (1) at (-1, 0) {$1$}; 
        \node[draw, circle, minimum width=.75cm] (5) at (1,0) {$5$}; 
        \node[draw, circle, minimum width=.75cm] (6) at (.5,-1) {$6$};
        \node[draw, circle, minimum width=.75cm] (4) at (0, 0) {$4$};  
        \node[draw, circle, minimum width=.75cm] (7) at (0, -2) {$7$}; 
        \node[draw, circle, minimum width=.75cm] (9) at (1, -2) {$9$}; 

        \draw[-latex,red] (2) -- (3);
        \draw[-latex] (3) -- (4);
        \draw[-latex] (3) -- (1);
        \draw[-latex] (3) -- (8);
        \draw[-latex] (3) -- (5);
        \draw[-latex] (3) -- (4);
        \draw[-latex] (8) -- (6);
        \draw[-latex] (5) -- (6);
        \draw[-latex] (1) -- (6);
        \draw[-latex] (6) -- (7);
        \draw[-latex] (6) -- (9);
        \draw[-latex] (4) -- (6);
\end{tikzpicture}
{\tiny $(PDP,\rho)\backslash Q-\text{ERROR}$}
\end{minipage}%
\hspace{.5cm}
\begin{minipage}[t]{.2\linewidth}
\centering
\begin{tikzpicture}[thick,scale=.6,every node/.style={scale=.4}]
        \node[draw, circle, minimum width=.75cm] (2) at (0,1.4) {$2$};
        \node[draw, circle, minimum width=.75cm] (3) at (-1.5,.6) {$3$};
        \node[draw, circle, minimum width=.75cm] (8) at (1.5,0.2) {$8$};
        \node[draw, circle, minimum width=.75cm] (1) at (0, -1) {$1$}; 
        \node[draw, circle, minimum width=.75cm] (5) at (0,.2) {$5$}; 
        \node[draw, circle, minimum width=.75cm] (6) at (-1.5,-2) {$6$};
        \node[draw, circle, minimum width=.75cm] (4) at (-1.5,-.2) {$4$};  
        \node[draw, circle, minimum width=.75cm] (7) at (0, -2) {$7$}; 
        \node[draw, circle, minimum width=.75cm] (9) at (1.5,-2) {$9$}; 

        \draw[-latex] (2) -- (5);
        \draw[-latex] (2) -- (3);
        \draw[-latex] (5) -- (1);
        \draw[-latex] (2) -- (8);
        \draw[-latex] (3) -- (4);
        \draw[-latex] (8) -- (1);
        \draw[-latex] (1) -- (7);
        \draw[-latex] (1) -- (6);
        \draw[-latex] (1) -- (9);
        \draw[-latex] (4) -- (1);
\end{tikzpicture}
{\tiny $(PDP,\rho)\backslash Q-\text{MALLOW'S ERROR}$}
\end{minipage}%
\begin{minipage}[t]{.24\linewidth}
\centering
\begin{tikzpicture}[thick,scale=.6,every node/.style={scale=.4}]
        \node[draw, circle, minimum width=.75cm] (102) at (0,1.2) {$2$};
        \node[draw, circle, minimum width=.75cm] (104) at (0,-0.6) {$4$};
        \node[draw, circle, minimum width=.75cm] (103) at (0,-1.2) {$3$};
        \node[draw, circle, minimum width=.75cm] (101) at (0,.6) {$1$}; 
        \node[draw, circle, minimum width=.75cm] (105) at (0, 0) {$5$}; 
        \node[draw, circle, minimum width=.75cm] (107) at (1.5,-2.2) {$7$};
        \node[draw, circle, minimum width=.75cm] (108) at (0,1.8) {$8$};  
        \node[draw, circle, minimum width=.75cm] (109) at (0, -1.8) {$9$}; 
        \node[draw, circle, minimum width=.75cm] (106) at (-1.5, -2.2) {$6$}; 

        \draw[-latex] (102) -- (101);
        \draw[-latex] (108) -- (102);
        \draw[-latex] (101) -- (105);
        \draw[-latex] (105) -- (104);
        \draw[-latex] (104) -- (103);
        \draw[-latex] (103) -- (109);
        \draw[-latex] (109) -- (106);
        \draw[-latex] (109) -- (107);
\end{tikzpicture}
{\tiny $(PDP,\rho)\backslash M-\text{QJ ERROR}$}
\end{minipage}%
}}
\captionof{figure}{The concensus orders for the $(PDP,\rho)\backslash Q$ model on the synthetic data with 1100-1103 structure (simulation 1-4 from left to right). We conclude an edge if such order relation has more than 0.2 posterior probability (inferred from section \ref{sec:reconstruction}). An edge is colored red if it has more than 0.9 posterior probability. }\label{fig:concensus-syn-0003-q}
}

\end{document}